\documentclass[reprint,amsmath,amssymb,aps,superscriptaddress,longbibliography,nofootinbib]{revtex4-2}
\usepackage{graphicx}
\usepackage{subfigure}
\usepackage{comment}
\usepackage{dcolumn}
\usepackage{bm}
\usepackage{amsmath,amssymb}
\usepackage{amsthm}
\usepackage{mathrsfs}
\usepackage{indentfirst}
\usepackage{float}
\usepackage{braket}
\usepackage{physics}
\usepackage[utf8]{inputenc}
\usepackage{bm}
\usepackage{graphicx}
\usepackage{tikz}
\usepackage{physics}
\usetikzlibrary{arrows, shapes.gates.logic.US, calc}
\usetikzlibrary{angles,quotes}
\tikzstyle{branch}=[fill, shape=circle, minimum size=3pt, inner sep=0pt]
\usepackage{blochsphere}
\usepackage{mathtools}
\usepackage{color,graphicx} 
\newcommand\D{\!\operatorname{d}\!}
\usepackage{algorithm}
\usepackage[algo2e]{algorithm2e}
\usepackage{bbold}
\usepackage{scalerel}

\usetikzlibrary{quantikz}

\usepackage{hyperref}
\hypersetup{
 colorlinks=true,
 citecolor=red,
 linkcolor=blue,
 urlcolor=blue,
 pdfpagemode=UseNone,
 pdfstartview=FitH}
\usepackage{subfiles}

\newcommand{\orcidicon}[1]{\href{https://orcid.org/#1}{\includegraphics[height=\fontcharht\font`\B]{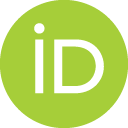}}}

\theoremstyle{definition}
\newtheorem{Obs}{Obs}
\theoremstyle{definition}
\newtheorem{Lemma}{Lemma}
\theoremstyle{definition}

\theoremstyle{definition}
\newtheorem{Definition}{Definition}
\theoremstyle{definition}
\newtheorem{Theorem}{Theorem}

\footnotetext{These authors contributed equally}

\begin{document}
\preprint{APS/123-QED}
\title{Mitigating Coherent Errors through a Decoherence-Resistant \\ Variational Framework employing Stabilizer States}

\author{Giovanni Di Bartolomeo\,\orcidicon{0000-0002-1792-7043}$^\diamond$}
\email{giovanni.dibartolomeo@units.it}
\affiliation{Department of Physics, University of Trieste, Strada Costiera 11, 34151 Trieste, Italy}
\affiliation{Istituto Nazionale di Fisica Nucleare, Trieste Section, Via Valerio 2, 34127 Trieste, Italy}

\author{Giulio Crognaletti\,\orcidicon{0009-0005-6999-8061}$^\diamond$}
\email{giulio.crognaletti@phd.units.it}
\affiliation{Department of Physics, University of Trieste, Strada Costiera 11, 34151 Trieste, Italy}
\affiliation{Istituto Nazionale di Fisica Nucleare, Trieste Section, Via Valerio 2, 34127 Trieste, Italy}

\author{Angelo Bassi\,\orcidicon{0000-0001-7500-387X}}
\affiliation{Department of Physics, University of Trieste, Strada Costiera 11, 34151 Trieste, Italy}
\affiliation{Istituto Nazionale di Fisica Nucleare, Trieste Section, Via Valerio 2, 34127 Trieste, Italy}

\author{Michele Vischi\,\orcidicon{0000-0002-5724-7421}}
\email{vischi.michele@units.it}
\affiliation{Department of Physics, University of Trieste, Strada Costiera 11, 34151 Trieste, Italy}
\affiliation{Istituto Nazionale di Fisica Nucleare, Trieste Section, Via Valerio 2, 34127 Trieste, Italy}

\begin{abstract}
Stabilizer states are a central resource in quantum information processing, underpinning a wide range of applications. While they can be efficiently generated via Clifford circuits, the presence of coherent errors—such as small-angle miscalibrations in native gate implementations—can significantly impact their quality. In this work, we introduce Variational Coherent Error Mitigation (VCEM), a method that employs the stabilizer formalism to suppress coherent errors through variational optimization of native gates parameters. VCEM demonstrates robust performance, remaining largely unaffected by incoherent noise, enabling pre-compensation of coherent errors prior to the application of standard incoherent error mitigation techniques. We demonstrate the effectiveness and robustness of VCEM through numerical simulations. 
\end{abstract}

\maketitle

\section{Introduction}

The preparation of stabilizer states \cite{Englbrecht_symmetries, veitch2014resource, graeme_typical} constitutes a fundamental subroutine in quantum information processing. A stabilizer state is uniquely characterized as the simultaneous +1 eigenstate of a set of mutually commuting Pauli operators. This symmetry property represents a powerful resource that is extensively exploited in Quantum Error Correction (QEC) codes \cite{shor_QEC, knill2000theory,gottesman2009introduction,terhal_QEC,cory_QEC,2024_google_QEC}, Measurement-Based Quantum Computing (MBQC) \cite{one_way_raussendorf,briegel2009measurement,raussendorf2003measurement,Wei_2021,walther2005experimental,ferguson_MBQC}, Quantum Key Distribution (QKD) \cite{shor_QKD,gottesman_secure,renner2008security,leroy_QKD,Yu_QKD}, Distributed Quantum Computing (DQC) \cite{eisert_DQC,Caleffi_2024,beals2013efficient,Wu_2023,main2025distributed}, Quantum Internet (QI) \cite{kimble2008quantum,Cao_QI,azuma_QI,koudia2023quantum}, many-body quantum simulations \cite{Sun_2025,Gu_many_body,Antu_2025,Liu_many_body}, and beyond. Relevant examples of stabilizer states include graph states \cite{hein2006entanglement,hein2004, shettell2020graph}, GHZ states \cite{contreras2019resource,d2004computational}, and logical code states \cite{kitaev2003fault,bonilla2021xzzx,castelnuovo_toric,zhao_surface}.

Although stabilizer states can, in principle, be generated using Clifford circuits—operations that are efficiently classically simulable \cite{gottesman1998}—the generation on demand of high quality stabilizer states is required across all these applications. This necessitates strategies for their robust preparation and protection against detrimental noise \cite{brandhofer2025hardware}.

Most developed techniques focus on reducing the effects of incoherent errors, the ones arising from the interaction of a quantum device with its surrounding environment \cite{nielsen2000quantum, benenti2019principles}. Such errors break unitarity and corrupt any quantum information processing task. However, another important class of errors are coherent ones, i.e. small angle miscalibrations leading to systematic over- or under-rotation errors in the native gates of a quantum device. Even when these errors are small, their effect both on physical and logical qubits can be significant \cite{Wallman_2015,ball_2016,kueng_2016,wallman_2016}. In QEC coherent errors are typically ignored because, when the decoder is optimized for correcting independent incoherent Pauli errors, the coherent contribution to the logical error is negligible at fewer than $\epsilon^{-(n-1)}$ error correction cycles, where $\epsilon$ is the order of magnitude of coherent errors and $n$ is the number of physical qubits \cite{greenbaum2017modeling,Huang_2019,bravyi2018correcting}. Above this number of correction cycles, the coherent logical error causes logical failure faster than the incoherent one, and this have to be counteracted with efficient coherent error mitigation methods at the physical level. This reflects the current status of quantum computation which is expected to persist in the near- to mid-term future.

Few techniques focusing on coherent Quantum Error Mitigation (QEM) are present in the literature \cite{kern2005quantum,endo_QEM,campbell_QEM,endo_practical_QEM,endo_QEM_coherent,orsucci_coherent}. 

In this work, by exploiting the stabilizerness we derive an efficient error mitigation technique based on the variational principle named Variational Coherent Error Mitigation (VCEM).
First the Clifford circuit generating the stabilizer state is expressed in terms of the native gates of a quantum device, then each of such native gate is parametrized. The optimization of a suitable cost function based on the stabilizer operators allows to find the optimal parameters that cancel coherent errors. While the natural application of VCEM is the correction of coherent errors in stabilizer-state preparation, since VCEM operates at the level of native gates, it is applicable as a general-purpose technique to correct coherent errors.

Moreover, we show how VCEM remains robust and effective even in the presence of realistic incoherent errors. Although incoherent errors modify the cost function, it is still possible to find optimal parameters based on the properties of the Clifford circuit and the noise channels.

This fact leads to the main advantage of our technique: coherent errors can be first corrected by VCEM without taking care of incoherent ones. Subsequently, any circuit can run with optimal parameters and it could be paired with incoherent QEM techniques \cite{temme_PEC,van2023probabilistic,filippov2023,bultrini2023unifying}. 

We show the effectiveness of VCEM with accurate numerical simulations. 

\section{Coherent errors on stabilizer states}
\label{sec:stabilizers}

Consider $\mathcal{\bf{P}}_n$ the Pauli group on $n$ qubits. An $n$-qubit state $\ket{\psi_S}$ is called a stabilizer state if there exists a subgroup $\mathcal{\bf{S}}_n \subset \mathcal{\bf{P}}_n$ of dimension $2^n$ and such that 
\begin{equation}
    \hat{S}_i\ket{\psi_S} =\ket{\psi_S}\,\ ,
\end{equation} for every generator $\hat{S}_i$ of $\mathcal{\bf{S}}_n$ with $i =1,\dots,n$ \cite{arab2024lecture}. 
The $\hat{S}_i$ operators are called stabilizers, they are Pauli strings and have eigenvalues $\pm 1$ with degeneracy $2^{n-1}$:
\begin{align}
&\hat{S_i}|\psi_{ik}^{(+)}\rangle = |\psi_{ik}^{(+)}\rangle\,\, \forall k = 1,\dots, 2^{n-1} \\
& \hat{S_i}|\psi_{ik}^{(-)}\rangle = -|\psi_{ik}^{(-)}\rangle\,\, \forall k = 1,\dots, 2^{n-1}
\end{align} 
where the eigenvectors $|\psi_{ik}^{(\pm)}\rangle$ are orthonormal $\langle\psi_{ik}^{(\pm)}|\psi_{ik'}^{(\pm)}\rangle = \delta_{kk'}$.

According to \cite{nielsen2000quantum}, it is always possible to identify a Clifford circuit $\hat{U}_S$, namely a quantum circuit composed of Clifford gates only, such that $\ket{\psi_S} = \hat{U}_S \ket{0}^{\otimes n}$.

In presence of even small amount of coherent errors, the state $\ket{\psi_S}$ becomes non stabilized. Coherent errors can be described as unwanted small extra rotations in the native gates of a quantum device resulting from gates miscalibration. In what follows, we assume that each gate in $\hat{U}_S$ has its own small coherent error which is random but stable (constant), at least during the entire running time. Our aim is to correct for such errors.

Any stabilizer state can be expanded as a superposition of the $|\psi_{ik}^{(\pm)}\rangle$ eigenvectors for each $\hat{S}_i$
\begin{equation}\label{eq:stab_state_expansion}
    |\psi_S\rangle = \sum_{k =1}^{2^{n-1}}\Bigl(a_{ik}|\psi_{ik}^{(+)}\rangle + b_{ik}|\psi_{ik}^{(-)}\rangle\Bigr)\,\ ,
\end{equation}
where $a_{ik}$ are generic amplitudes and $b_{ik} = 0$ for every $k = 1,\dots,2^{n-1}$. 
In the presence of coherent errors the state is modified as $\ket{\psi_S(\vec{\epsilon})} = \hat{U}_S(\vec{\epsilon})\ket{0}^{{\otimes}n}$ where $\vec{\epsilon}$ is a vector of small errors and $\hat{U}_S(\vec{\epsilon})$ is in general non-Clifford.
Analogously to Eq.~\eqref{eq:stab_state_expansion}, the modified state can be written as
\begin{equation}
\label{eq:stab_state_expansion_errors}
|\psi_S(\vec{\epsilon})\rangle = \sum_{k =1}^{2^{n-1}}\Bigl(a_{ik}(\vec{\epsilon})|\psi_{ik}^{(+)}\rangle + b_{ik}(\vec{\epsilon})|\psi_{ik}^{(-)}\rangle\Bigr) \,\ ,
\end{equation}
where $\sum_k(|a_{ik}(\vec{\epsilon})|^2 + |b_{ik}(\vec{\epsilon})|^2) = 1$. In general, $|\psi_S(\vec{\epsilon})\rangle$ is not a stabilizer state, because the $b_{ik}(\vec{\epsilon})$ might be different from zero. In the following we assume that for any small perturbations $|\vec{\epsilon}|\ll1$ at least one $|b_{ik}(\vec{\epsilon})|^2>0$ such that $|\psi_S(\vec{\epsilon})\rangle$ is not a stabilizer state.

\section{Variational coherent error mitigation (VCEM)}
\label{sec:VCEM}

In order to correct coherent errors one can parametrize the circuit $\hat{U}_S(\vec{\theta}+\vec{\epsilon})$ where $\vec{\theta}$ is a vector of variational parameters. The idea is to transpile the circuit into the native gates of a given quantum device and parametrize their rotation angles. The resulting state $|\psi_S(\vec{\theta}+\vec{\epsilon})\rangle = \hat{U}_S(\vec{\theta}+\vec{\epsilon})\ket{0}^{\otimes n}$ can again be expanded analogously to Eq.~\eqref{eq:stab_state_expansion_errors} as 
\begin{equation}\label{eq:stab_state_var_expansion}   |\psi_S(\vec{\theta}+\vec{\epsilon})\rangle = \sum_{k =1}^{2^{n-1}}\Bigl(a_{ik}(\vec{\theta}+ \vec{\epsilon})|\psi_{ik}^{(+)}\rangle + b_{ik}(\vec{\theta}+\vec{\epsilon})|\psi_{ik}^{(-)}\rangle\Bigr)  \,.
\end{equation}
We can define the following cost function
\begin{equation}\label{eq:cost_function}
C(\vec{\theta}) = \sum_{i=1}^{n}C_i(\vec{\theta}) = -\sum_{i=1}^{n}\langle \psi_S(\vec{\theta}+\vec{\epsilon})|\hat{S}_i|\psi_S(\vec{\theta}+\vec{\epsilon})\rangle \, ,
\end{equation}
and we prove that by using Eq.~\eqref{eq:stab_state_var_expansion} this quantity can be optimized with the variational principle, i.e. it has global minimum:
\begin{equation}
\label{eq:cost_function_Si}
\begin{aligned}
    &C_{i}(\vec{\theta}) = -\langle \psi_S(\vec{\theta}+\vec{\epsilon})|\hat{S}_i|\psi_S(\vec{\theta}+\vec{\epsilon})\rangle \\
    &= - \sum_k \big(|a_{ik}(\vec{\theta}+\vec{\epsilon})|^2 -|b_{ik}(\vec{\theta}+\vec{\epsilon})|^2\big) \\
    & = - \sum_k \big(|a_{ik}(\vec{\theta}+\vec{\epsilon})|^2 +|b_{ik}(\vec{\theta}+\vec{\epsilon})|^2\big) +2 \sum_k |b_{ik}(\vec{\theta}+\vec{\epsilon})|^2 \\
    &= - 1 + 2 \sum_k |b_{ik}(\vec{\theta}+\vec{\epsilon})|^2 \geq -1\, ,
\end{aligned}
\end{equation}
by plugging Eq.~\eqref{eq:cost_function_Si} into Eq.~\eqref{eq:cost_function}, one obtains 
\begin{equation}
    C(\vec{\theta}) \geq -n\,\ .
\end{equation}
The global minimum value of the cost function $-n$ is exactly obtained when the state is the stabilizer state. 
In other words, assuming that $|\vec{\epsilon}|$ is small, one can start the optimization in $\vec{\theta} = 0$, that is equivalent to optimize the cost function in a region $\mathcal{B}(-\vec{\epsilon},\delta)$ for $\delta\ll1$, and such that $0 \in \mathcal{B}$, close to the optimal values where the cost function is convex, i.e. the Hessian $\mathbf{H}=\partial_{\theta_{k}}\partial_{\theta_l} C(\vec{\theta}) |_{\vec{\theta}=-\vec{\epsilon}}$ is positive definite. This excludes trainability issues such as barren plateaus \cite{Larocca25} in the ideal unitary case, consistently with the literature \cite{Puig25}. In Sec. \ref{sec:VCEM_incoherent}, we show how the same claim continues to hold also in the presence of incoherent errors.

By finding the minimum of the cost function the resulting optimal parameters $\vec{\theta}_{opt}$ are such that $\vec{\theta}_{opt} =- \vec{\epsilon}$ and coherent errors are corrected. When $\vec{\theta}_{opt}$ is obtained, the corrected circuit is such that $\hat{U}_S(\vec{\theta}_{opt}+\vec{\epsilon}) = \hat{U}_S$. We name this technique Variational Coherent Error Mitigation (VCEM). A pictorial representation of the VCEM circuit is shown in Fig. \ref{fig:VCEM_circuit}. For a simple analytical example see appendix \ref{sec:analytic}.
\begin{figure}[H]
    \centering
    \includegraphics[width=\linewidth]{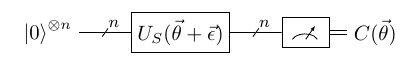}
    \caption{Parametrized circuit to evaluate the VCEM cost function.}
    \label{fig:VCEM_circuit}
\end{figure}

The natural application of VCEM is the correction of coherent errors arising in the preparation of a stabilizer state. However, the method has broader applicability, as the error mitigation is performed at the level of native gates acting on physical qubits. This enables its use as a general-purpose technique, independent of the specific stabilizer state initially considered. In this more general setting, the stabilizer state can, in principle, be freely chosen so as to optimize the correction of coherent errors.

A gate-by-gate correction strategy is generally impractical, since the number of required variational optimizations scales with both the number of qubits and the number of connected qubit pairs. VCEM becomes advantageous when a single, representative stabilizer state is selected, allowing all relevant variational parameters to be incorporated and optimized simultaneously. The choice of stabilizer state can be tailored to the target application: for instance, in simulations of many-body quantum systems, one might select a stabilizer state reflecting the system’s topology—such as spin lattices or molecular structures. Another practical choice is the graph state associated with the connectivity of the specific quantum device.

\section{VCEM in presence of incoherent errors}
\label{sec:VCEM_incoherent}
In the previous section we assumed the presence of coherent errors only, here we show that under standard assumptions VCEM can still be effective in correcting coherent errors also in presence of incoherent ones. Incoherent errors arise from unwanted interaction with the environment \cite{krantz2019quantum, breuer2002theory} and in general their effects can be described through quantum maps \cite{benenti2019principles, nielsen2000quantum}. Thus from this point on we work with density matrices instead of state vectors. We define $\hat{\rho}_S=\ket{\psi_S}\bra{\psi_S}=\mathcal{U}_S\hat{\rho}_0$ where $\mathcal{U}_S=\hat{U}_S\, \cdot\, \hat{U}_S^\dagger$, $\hat{\rho}_0= (\ket{0}\bra{0})^{\otimes n}$ and to ease the notation we omit the dependence on variational parameters $\vec{\theta}$ and coherent errors $\vec{\epsilon}$.

Moreover when needed, we divide the circuit $\mathcal{U}_S$ in $M$ moments $\mathcal{U}_S = \prod_{q=1}^{M}\mathcal{U}_{q}=\mathcal{U}_M\mathcal{U}_{M-1}\dots\mathcal{U}_2\mathcal{U}_{1}$, where each moment $\mathcal{U}_{q}$ is a sub-circuit consisting of native gates acting simultaneously on a number of qubits $\leq n$ with $q = 1,\dots, M$.

\subsection{General Pauli map at the end of the circuit}
\label{subsec:pauli_maps_end_circ}
At this stage we assume that the effect of incoherent errors can be described by a Pauli map (see Def. \ref{def:dpm} in Appendix \ref{appendix:notation_and_lemmas}) acting at the end of the circuit $\mathcal{U}_S$, given by
\begin{equation}\label{eq:pauli_map}
    \mathcal{P}\hat{\rho}_S = \sum_{j = 0}^{4^{n}-1}p_j \hat{P}_j \hat{\rho}_S\hat{P_j}\,,
\end{equation}
where $p_j$ is the probability relative to the Pauli string $\hat{P}_j$. This is depicted in Fig. \ref{fig:VCEM_circuit_noisy}.

\begin{figure}[H]
    \centering
    \includegraphics[width=\linewidth]{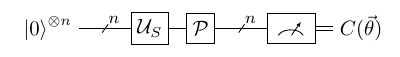}
    \caption{Parametrized circuit to evaluate the modified VCEM cost function with a Pauli noise channel after $\mathcal{U}_S$.}
    \label{fig:VCEM_circuit_noisy}
\end{figure}

By using Eq.~\eqref{eq:pauli_map}, each term $C_i$ in Eq.~\eqref{eq:cost_function_Si} gets modified as 
\begin{equation}\label{eq:senza_nome}
\begin{aligned}    &\tilde{C}_i^{ \scaleto{(\mathcal{P})}{6pt}}(\vec{\theta}) = -\Tr\Bigl(\hat{S}_i\mathcal{P}\hat{\rho}_S\Bigr)=-\chi_i\Tr\Bigl(\hat{S}_i\hat{\rho}_S\Bigr)=\chi_iC_i(\vec{\theta})\, ,
\end{aligned}
\end{equation}
where we used the ciclicity of the trace and Lemma \ref{lemma:action_pauli_noise_string} in Appendix \ref{appendix:notation_and_lemmas} such that $\mathcal{P}\hat{S}_i=\chi_i\hat{S}_i$ where $\chi_i=1-2\Gamma_i$ with $\Gamma_i = \sum_{j:\{\hat{P}_j,\hat{S_i}\} =0} p_j$.

The resulting modified cost function is given by
\begin{equation}\label{eq:noisy_cost_function}
    \tilde{C}^{ \scaleto{(\mathcal{P})}{6pt}}(\vec{\theta}) = \sum_{i=1}^{n}\chi_iC_i(\vec{\theta})\, .
\end{equation}
Usually, the Pauli map in Eq.~\eqref{eq:pauli_map} is such that the probability associated to the identity operator $\hat{P}_0 = \hat{\mathbb{1}}$ is $p_{0} = 1 - \sum_{j = 1}^{4^{n}-1}p_j$ with $p_0\ge 1/2$. This is a reasonable assumption otherwise the probability of having errors would be higher than that of the ideal evolution. Given this assumption, since $\Gamma_i$ does not contain $p_0$, then $\Gamma_i < 1/2$ which implies that in Eq.~\eqref{eq:noisy_cost_function}, all factors $\chi_i > 0$ for each $i$. Thus, if $\tilde{\mathbf{H}}^{ \scaleto{(\mathcal{P})}{6pt}}$ is the Hessian matrix of $\tilde{C}^{ \scaleto{(\mathcal{P})}{6pt}}$, namely $\tilde{\mathbf{H}}_{kl}^{ \scaleto{(\mathcal{P})}{6pt}} = \partial_{\theta_{k}}\partial_{\theta_l} \tilde{C}^{ \scaleto{(\mathcal{P})}{6pt}}(\vec{\theta}) |_{\vec{\theta}=-\vec{\epsilon}}$\, , then let $\mathbf{H}_i$ denote the Hessian matrices of $C_i$, we have that
\begin{equation}
    \tilde{\mathbf{H}}^{ \scaleto{(\mathcal{P})}{6pt}} = \sum_i \chi_i\mathbf{H}_i=\chi_{min} \mathbf{H}\, + \sum_i (\chi_i-\chi_{min})\mathbf{H}_i\, .
\end{equation}
Each term in the last summation is positive semi-definite, since by definition $(\chi_i - \chi_{\min}) \ge 0$ and $\mathbf{H}_i$ is positive semi-definite. Furthermore, since $\mathbf{H}$ is positive definite, as discussed in Sec.~\ref{sec:VCEM}, and $\chi_{\min} > 0$, it follows that $\tilde{\mathbf{H}}^{(\mathcal{P})}$ is also positive definite.
This ensures the convexity of the modified cost function in the region $\mathcal{B}(-\vec{\epsilon},\delta)$. Despite the fact that $C(\vec{\theta}_{opt})$ is a global minimum while $\tilde{C}^{ \scaleto{(\mathcal{P})}{6pt}}(\vec{\theta}_{opt})$ might become a local minimum, it is still possible to find $\vec{\theta}_{opt}$, starting the optimization with initial parameters $\vec{\theta} = 0\in \mathcal{B}$, since $|\vec{\epsilon}|$ is small. We show this scenario in Fig. \ref{fig:landscape}. A simple analytical example can be found in Appendix \ref{sec:analytic}.

\begin{figure}[H]
    \centering
    \includegraphics[width=\linewidth]{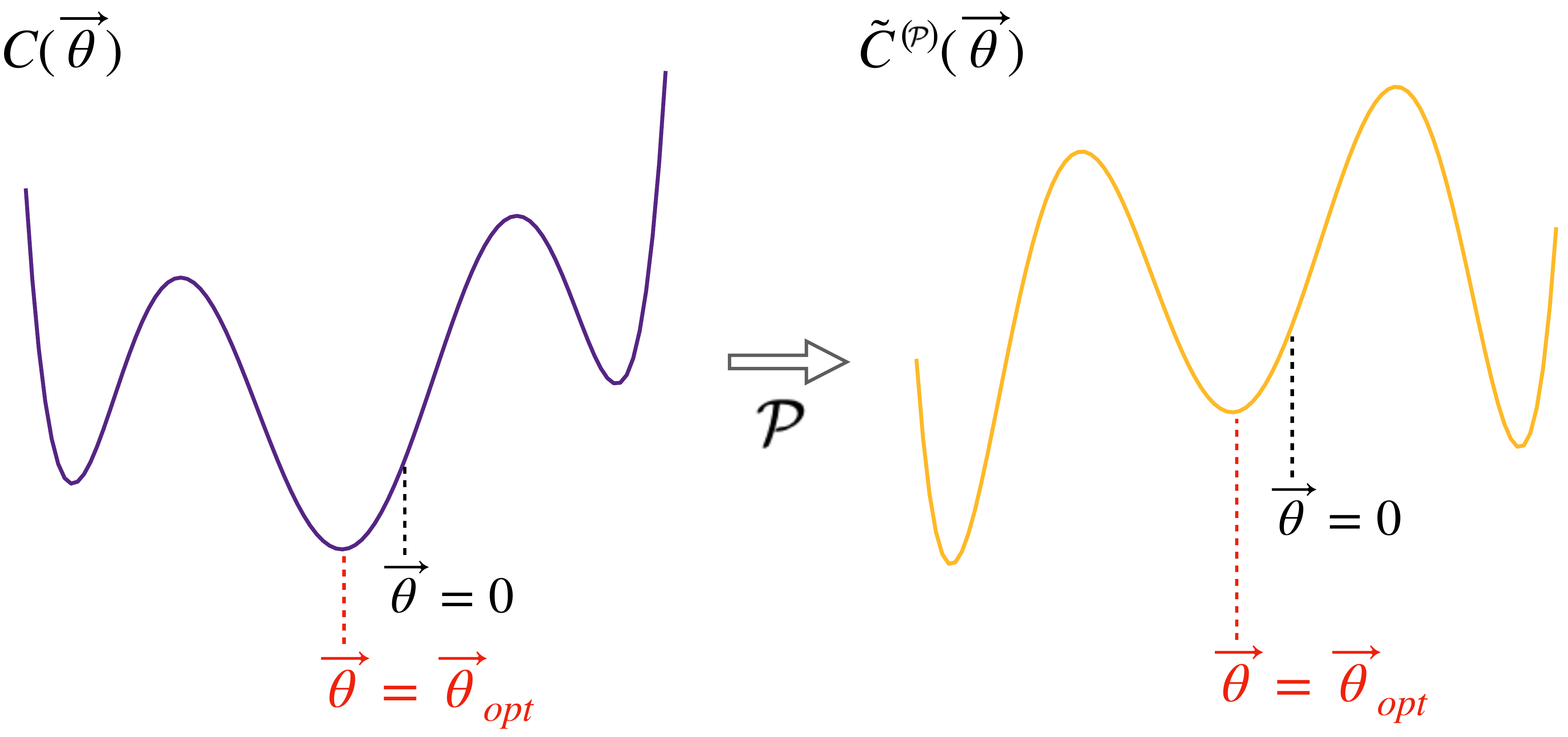}
    \caption{Sketch of how the landscape of the cost function could be modified by the Pauli incoherent noise $\mathcal{P}$ acting at the end of the circuit. The value of $\vec{\theta}$ in black is the initial choice of the parameters.}
    \label{fig:landscape}
\end{figure}

As a special case, one can consider the completely symmetric Pauli map, also called depolarizing map, for which $p_j = p/4^{n}$ for $j = 1,..\dots,4^{n}-1$ in Eq.~\eqref{eq:pauli_map}. Thus $\Gamma_i = \frac{p}{4^{n}} \sum_{j:\{\hat{P}_j,\hat{S_i}\} =0}1 = p/2$, leading to
\begin{equation}
\label{eq:symmetric_case}
\tilde{C}^{ \scaleto{(\mathcal{P})}{6pt}}(\vec{\theta}) = (1-p)C(\vec{\theta})\,.  
\end{equation}
In this case $\tilde{C}^{ \scaleto{(\mathcal{P})}{6pt}}(\vec{\theta}_{opt})$ remains a global minimum shifted by a constant offset $1-p$.

\subsection{Global depolarizing map after each circuit moment}
\label{subsec:global_dep_map}

Here we assume that after each circuit moment $\mathcal{U}_q$ acts a global depolarizing map, i.e. a global symmetric Pauli map, of the form (see Def. \ref{def:ddm} in Appendix \ref{appendix:notation_and_lemmas})
\begin{equation}\label{eq:global_dep}
\mathcal{D}_q\hat{\rho} = (1-p_q) \hat{\rho}+\frac{p_q}{2^n} \mathbb{1}\, ,
\end{equation}
where $\hat{\rho}$ is a generic density matrix. Consequently the noisy circuit can be written as
\begin{equation}\label{eq:circuit_rho_global_dep}
\mathcal{N}_{\mathcal{D}}\equiv\prod_{q=1}^{M}\mathcal{D}_q\mathcal{U}_q=\mathcal{D}_M\mathcal{U}_M\dots\mathcal{D}_2\mathcal{U}_2\mathcal{D}_1\mathcal{U}_1\,\ .
\end{equation}

By applying Eqs. \eqref{eq:global_dep} and \eqref{eq:circuit_rho_global_dep}, according to Lemma \ref{lemma:effective_dep} in Appendix \ref{appendix:notation_and_lemmas}, one gets:
\begin{equation}
\begin{aligned}
\label{eq:depolarizing}
\mathcal{N}_{\mathcal{D}}\hat{\rho}_0 &=(1-p')\mathcal{U}_S\hat{\rho}_0 + \frac{p'}{2^n}\mathbb{1}\\
& = \mathcal{D}\mathcal{U}_S\hat{\rho_0}=\mathcal{D}\hat{\rho}_S,
\end{aligned}
\end{equation}
where $\mathcal{D}$ is a global depolarizing map at the end of the circuit with error probability $p' = 1-\prod_{q=1}^{M}(1-p_q)$.

From Eq.~\eqref{eq:depolarizing} follows that global depolarizing maps acting after each circuit
moment can be seen as an effective global depolarizing map acting at the end of the entire circuit (see Fig. \ref{fig:VCEM_circuit_noise_dep}).
\begin{figure}[H]
    \centering
    \includegraphics[width=\linewidth]{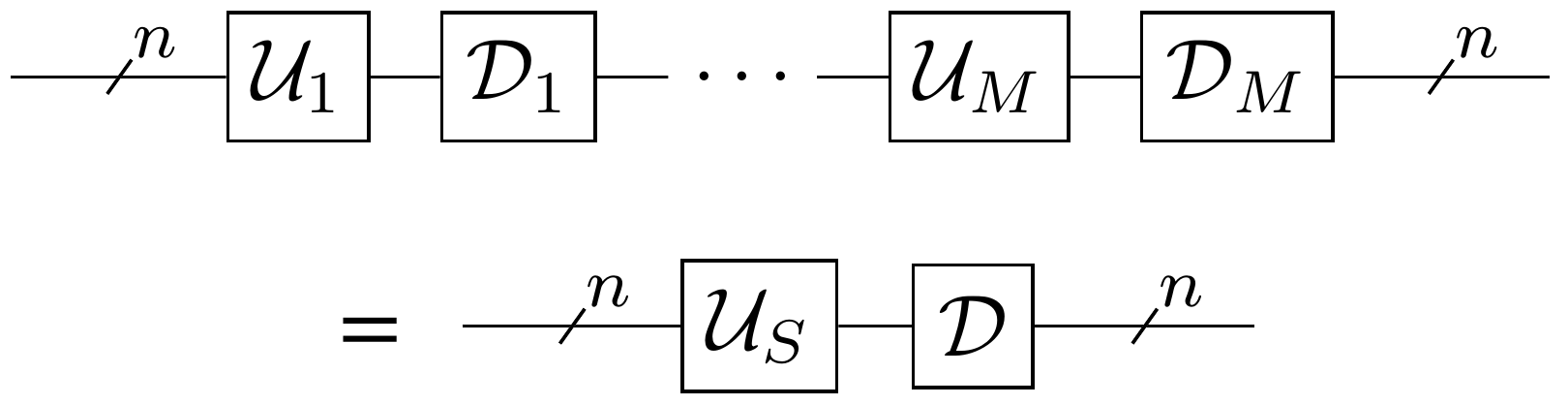}
    \caption{The action of global depolarizing maps after each momentum is equal to a single global depolarizing map after $\mathcal{U}_S$.}
    \label{fig:VCEM_circuit_noise_dep}
\end{figure}
Thus, the conclusion of the Sec. \ref{subsec:pauli_maps_end_circ} in the symmetric case is still true and the modified cost function reads
\begin{equation}\label{eq:cost_function_dep_momenta}
\tilde{C}^{ \scaleto{(\mathcal{D})}{6pt}}(\vec{\theta}) =(1-p')C(\vec{\theta})\,\ ,
\end{equation}
as in Eq.~\eqref{eq:symmetric_case}.
See appendix \ref{sec:analytic} for a simple analytical example.

\subsection{General Pauli map after each circuit moment}\label{sec:local_pauli_map}

Now we consider a more realistic scenario where after each circuit moment acts a general Pauli map
\begin{equation}
\begin{aligned}
\label{eq:m_local_pauli_map}
 \mathcal{P}_q\hat{\rho} = \sum_{j=0}^{4^{n}-1}p_{q,j}\hat{P}_{q,j}\hat{\rho}\hat{P}_{q,j}\,,
   \end{aligned}
\end{equation}
where $\hat{\rho}$ is a generic density matrix.

The corresponding noisy circuit reads

\begin{equation}\label{eq:circuit_rho_pauli}
\mathcal{N}\equiv\prod_{q=1}^{M}\mathcal{P}_q\mathcal{U}_q=\mathcal{P}_M\mathcal{U}_M\dots\mathcal{P}_2\mathcal{U}_2\mathcal{P}_1\mathcal{U}_1\,\ .
\end{equation}
In the following we assume that all the native gates $\mathcal{G}_w$ in the circuit are such that
\begin{equation}
\label{eq:gate_generators_main}
\begin{aligned}
 \mathcal{G}_w &= e^{-i(\phi_w+\theta_w+\epsilon_w)\mathcal{H}_w}\\ &=    e^{-i(\theta_w+\epsilon_w)\mathcal{H}_w}e^{-i\phi_w\mathcal{H}_w} \equiv \mathcal{G}_w^{\text{\tiny (N)}}\mathcal{G}_w^{\text{\tiny (C)}}
 \end{aligned}
 \end{equation}
where each $\theta_w+\epsilon_w$ gives rise to the non-Clifford transformation $\mathcal{G}_w^{\text{\tiny (N)}}$ associated to coherent noise and each $\phi_w$ is chosen such that $\mathcal{G}_w^{\text{\tiny (C)}}$ is Clifford. The latter is a reasonable assumption because it is satisfied by the totality of the native gates used in current quantum computers.

Accordingly we prove in Appendix \ref{appendix:remainder} that 
\begin{equation}\label{eq:send_noise_end_circuit_superoperator}
\begin{aligned}
\mathcal{N}\hat{\rho_0} =\mathcal{P}\mathcal{U}_S\hat{\rho}_0 + \mathcal{R}\hat{\rho}_0 =\mathcal{P}\hat{\rho}_S + \mathcal{R}\hat{\rho}_0\,,
\end{aligned}
\end{equation}
where $\mathcal{P}$ is an effective Pauli map as in Eq.~\eqref{eq:pauli_map} acting at the end of the circuit and $\mathcal{R}$ denotes the remainder term whose explicit expression can be found in Appendix \ref{appendix:remainder}. This is drawn in Fig. \ref{fig:VCEM_circuit_noise_momenta}.
\begin{figure}[H]
    \centering
    \includegraphics[width=\linewidth]{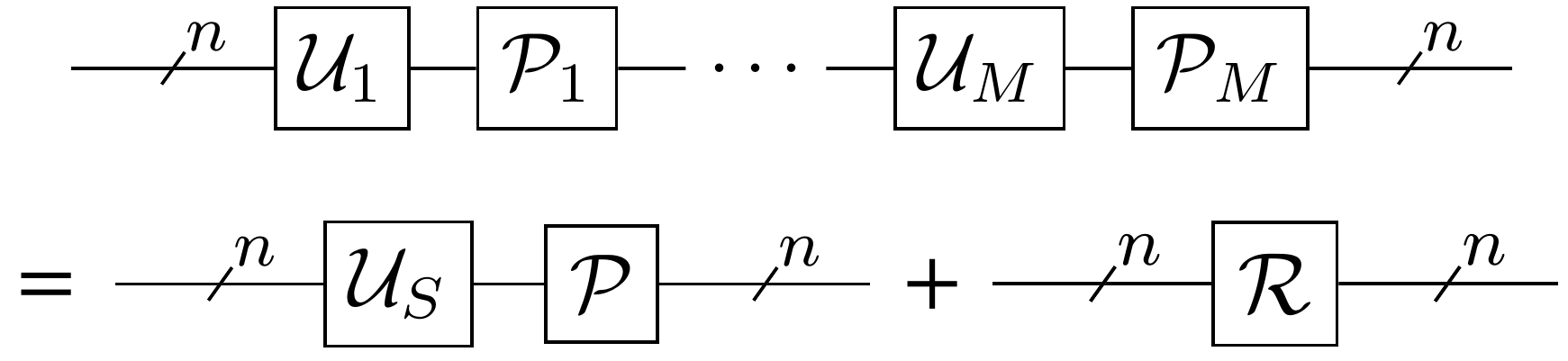}
    \caption{The action of Pauli maps after each moment is equivalent to an effective Pauli map after $\mathcal{U}_S$ plus a reminder $\mathcal{R}$.}
    \label{fig:VCEM_circuit_noise_momenta}
\end{figure}

By employing Eq.~\eqref{eq:send_noise_end_circuit_superoperator}, the modified cost function reads
\begin{equation}
\label{eq:cost_local_pauli}
\begin{aligned}
\tilde{C}(\vec{\theta}) &= \sum_{i=1}^n\tilde{C}_i(\vec{\theta}) = -\sum_{i=1}^{n}\Tr(\hat{S}_i \mathcal{N}\hat{\rho_0})\\
&=\tilde{C}^{\scaleto{(\mathcal{P})}{6pt}}(\vec{\theta})-\sum_{i=1}^{n}\Tr(\hat{S}_i\mathcal{R}\hat{\rho}_0)\,\ ,
\end{aligned}
\end{equation}
where $\tilde{C}^{\scaleto{(\mathcal{P})}{6pt}}$ is of the form of the cost function in Eq.~\eqref{eq:noisy_cost_function}.

The cost function $\tilde{C}$ differs from $\tilde{C}^{\scaleto{(\mathcal{P})}{6pt}}$ by the following term 
\begin{equation}\label{eq:delta_C_def}
   \Delta \tilde{C}(\vec{\theta}) \equiv\tilde{C}(\vec{\theta}) - \tilde{C}^{\scaleto{(\mathcal{P})}{6pt}}(\vec{\theta}) = -\sum_{i=1}^{n}\Tr(\hat{S}_i\mathcal{R}\hat{\rho}_0)\, . 
\end{equation} 
Given that the derivative of $\tilde{C}^{\scaleto{(\mathcal{P})}{6pt}}$ already vanishes at $\vec{\theta} = -\vec{\epsilon}$, one can prove the following theorem (the proof is presented in Appendix \ref{appendix:stationarity}) 

\begin{Theorem}[Stationarity of the solution]\label{theo:stationarity}

Let $\mathcal{N}$ be a noisy quantum circuit affected by Pauli maps after each circuit moment as defined in Eq.~\eqref{eq:circuit_rho_pauli}. Furthermore, assume that each native gate in $\mathcal{U}_q$ has a parameterization consistent with Eq.~\eqref{eq:gate_generators_main}. Then
\begin{equation}
\label{eq:zero_partial_derivative}
    \partial_{\theta_k} \Delta\tilde{C}(\vec{\theta}) \Bigl|_{\vec{\theta} = -\vec{\epsilon}} =0 \;\,\forall k\, .
\end{equation}
regardless of the noise strength.
\end{Theorem}
This means that the stationary point  $\vec{\theta} = -\vec{\epsilon}$ is shared betweeen the function $\tilde{C}$ in Eq.~\eqref{eq:cost_local_pauli}   $\tilde{C}^{\scaleto{(\mathcal{P})}{6pt}}$. Additionally, also the convexity of $\tilde{C}$ is guaranteed if the Pauli noise is weak enough as shown in Appendix \ref{appendix:convexity_C}.\\

Moreover, the following theorem holds true
\begin{Theorem}[Upper bound on $\Delta \tilde{C}$]
\label{theo:linear_upper_bound}
If the number of moments $M$ is fixed with respect to $n$, i.e. $M\le O(1)$, then $\Delta C(\vec{\theta})$ as defined in Eq.~\eqref{eq:delta_C_def} is upper bounded by
\begin{equation}
\label{eq:linear_upper_bound_2}
|\Delta C(\vec{\theta})| \leq O(\epsilon^2n)\, ,
\end{equation} 
where $\epsilon$ is the order of magnitude of coherent errors.  
\end{Theorem}
The proof is contained in Appendix \ref{appendix:upper_bound}.
As Eq.~\eqref{eq:linear_upper_bound_2} scales linearly with $n$ then $|\Delta C(\vec{\theta})/\tilde{C}(\vec{\theta})|$ does not depend on $n$.

We stress that the upper bound in Eq.~\eqref{eq:linear_upper_bound_2} depends only on the order of magnitude of coherent errors. This suggests the local convexity of $\tilde{C}$ as long as $\epsilon^2n$ is small beyond the weak Pauli noise limit, within the assumptions of Sec. \ref{subsec:pauli_maps_end_circ}. 

It follows from Theorems \ref{theo:stationarity} and \ref{theo:linear_upper_bound} that the case of general Pauli maps after each circuit moment can be essentially reduced to the one considered in Sec.~\ref{subsec:pauli_maps_end_circ}.

See appendix \ref{sec:analytic} for a simple analytical example.

\section{Realistic scenarios in which VCEM is effective}
\label{sec:realistic_scenarion}
Two of the main quantum algorithm subroutines based on the preparation of stabilizer states are the preparation of graph states \cite{hein2006entanglement,hein2004, shettell2020graph} and GHZ states \cite{contreras2019resource,d2004computational}. For example, 2D/3D graph states are a foundamental resource in the framework of Measurement Based Quantum Computation (MBQC) \cite{briegel2009measurement,raussendorf2003measurement}. For 2D/3D square/cubic graph states respectively, the optimal number of moments in which the circuit can be divided, without transpiling the latter into native gates, is $M' = 4$ and $M' = 6$. These are constant numbers independent from the number of nodes, i.e. of qubits, and this is a common property of many 2D/3D graph states and of GHZ states. When the circuit is transpiled into native gates the resulting number of moments $M$ might increase as $M = \nu M'$ where $\nu$ is a constant. These are practical situations where the scaling in Eq.~\eqref{eq:linear_upper_bound_2} holds. The corresponding states, whose structure matches the hardware connectivity, constitute realistic examples of stabilizer states suitable for performing VCEM.

In order to analyze situations in which the findings of Sec. \ref{sec:VCEM_incoherent} are useful, we make further assumptions on the structure of incoherent errors based on realistic features of quantum hardwares.

For real quantum devices, e.g. superconducting qubits \cite{ibm_quantum_exp,google_quantum_AI,rigetti}, photonic platforms \cite{quandela_cloud, Xanadu, PsiQuantum}, trapped ions \cite{IONQ, Quantinuum}, neutral atoms \cite{Pasqual,QuEra} and so on, can be expected that the dominant contributions of incoherent errors act after each circuit moment. This is the reason why we studied the scenario in Sec. \ref{sec:local_pauli_map} that resembles this situation. The main difference with respect to the latter is that, for a generic quantum device, the resulting noise map after each momentum is given by 
\begin{equation}
    \begin{aligned}
    \label{eq:m_local_map}
      &\mathcal{E}_q\hat{\rho} = \sum_{j=0}^{4^{n}-1}\hat{E}_{q,j}\hat{\rho}\hat{E}_{q,j}\, ,  
    \end{aligned}
\end{equation} for generic operators $\hat{E}_{q,j}$. In general $\mathcal{E}_q$ is not a Pauli map. However, this fact does not constitute a fundamental limitation for the VCEM technique. Indeed, by exploiting the Pauli twirling technique, \cite{Dur,wallman2016,hashim2021}  Eq.~\eqref{eq:m_local_map} can be reduced to Eq.~\eqref{eq:m_local_pauli_map} (see appendix \ref{sec:twirling}). Pauli twirling was already implemented efficiently on superconducting devices \cite{van2023probabilistic} in order to modify the noise behaviour, justifying our assumptions. Thus by combining VCEM with Pauli twirling, the conclusions drawn in Sec. \ref{sec:local_pauli_map} still apply.

\begin{figure*}[htp]
    \begin{minipage}{0.48\textwidth}
        \includegraphics[width=\textwidth]{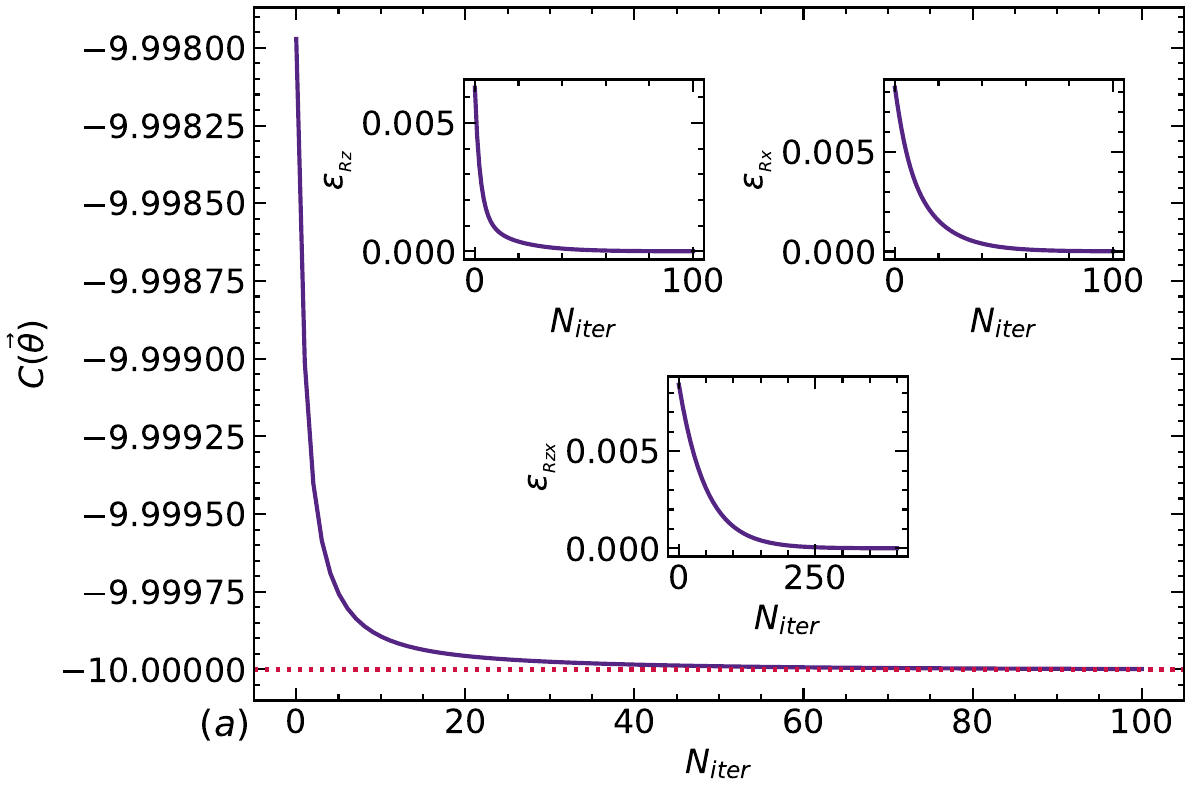}
    \end{minipage}
    \hfill%
    \begin{minipage}{0.48\textwidth}
        \includegraphics[width=\textwidth]{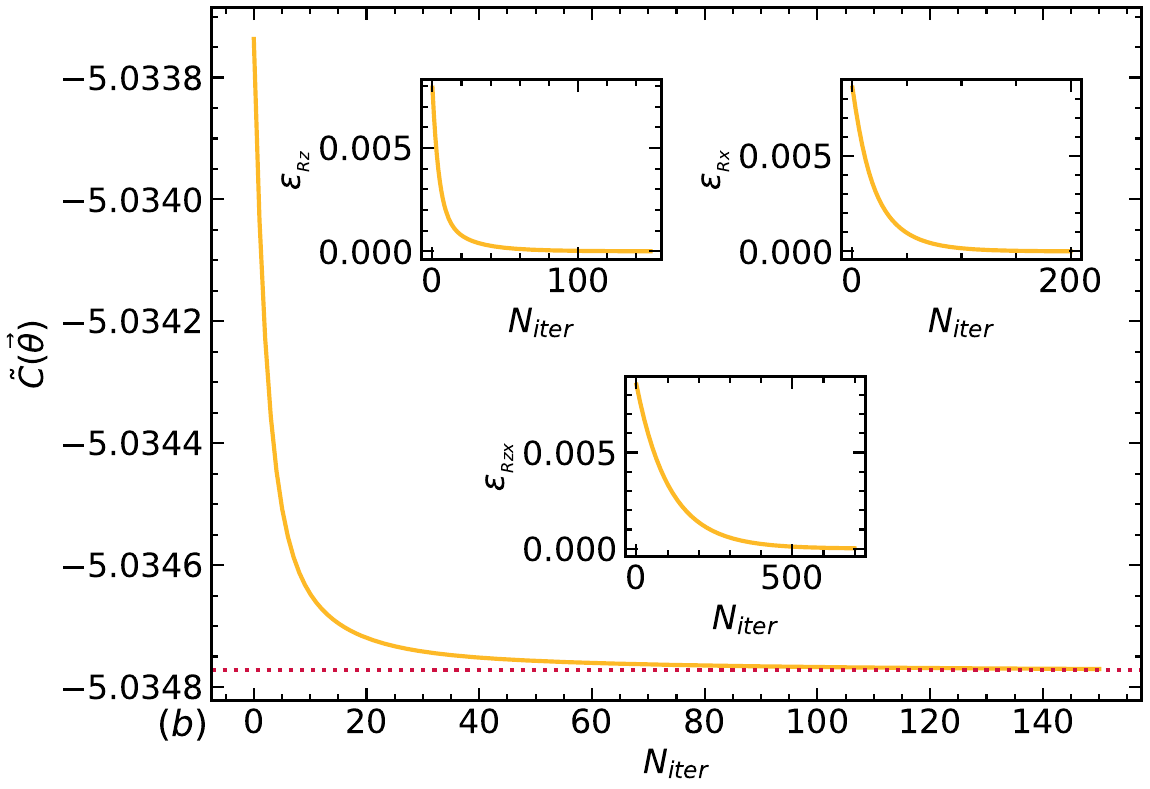}
    \end{minipage}
\caption{10 qubits 2D-rectangular graph state cost function optimization. Panel (a) shows the value of the cost function $C(\vec{\theta})$  with coherent errors only as a function of the number of optimization steps $N_{iter}$. The red dashed line is the minimal value $C(\vec{\theta}_{opt}) = -10$. The insets represent the behaviour of $\varepsilon_{_{R_z}}$, $\varepsilon_{_{R_x}}$ and $\varepsilon_{_{R_{zx}}}$ defined in the main text. Panel (b) shows the cost function $\tilde{C}(\vec{\theta})$ in presence of local Pauli incoherent errors after each circuit moment and the red dashed line is the minimal value $\tilde{C}(\vec{\theta}_{opt}) = -5.0347723$. The insets in panel (b) have the same meaning as for panel (a).}
\label{fig:cost_functions_numerical}
\end{figure*}

\begin{figure}[htp]

\includegraphics[width=0.45\textwidth]{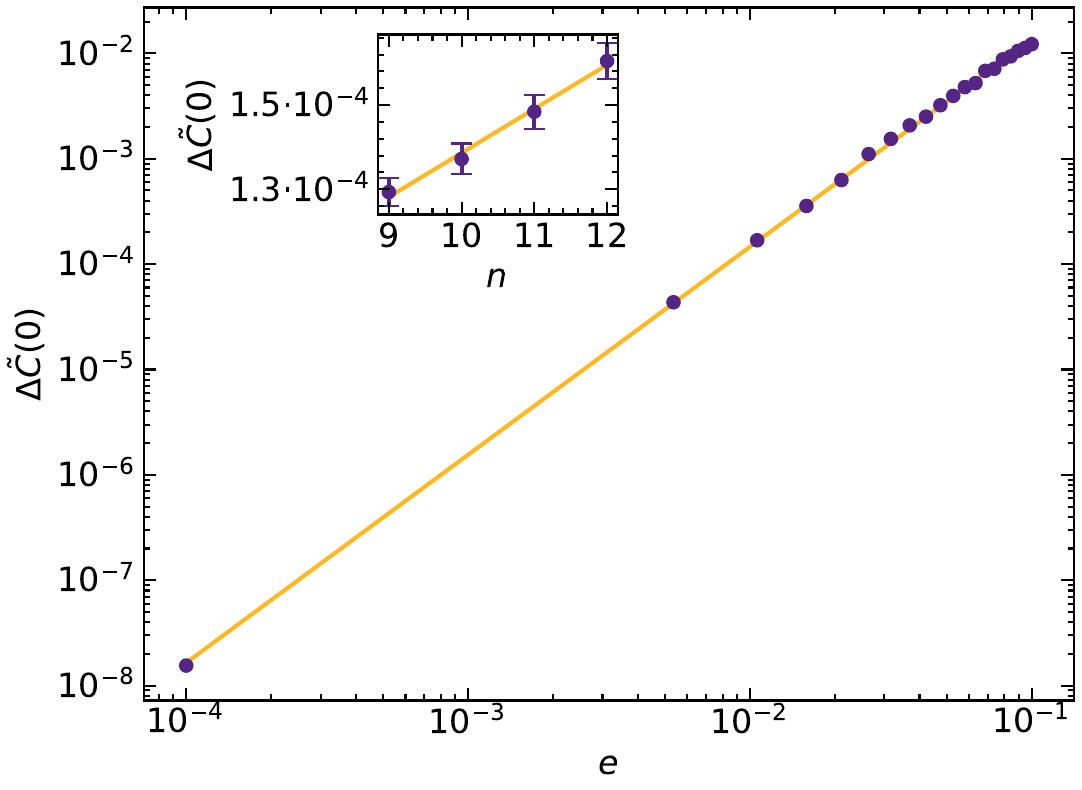}

\caption{$\Delta\tilde{C}(\vec{\theta})$ evaluated in $\vec{\theta}=0$ for 10 qubits 1D-linear graph state. The main plot shows the linear fit of $\Delta\tilde{C}(0)$ as a function of $\epsilon$ in log-log scale. The angular coefficient of the gold line is $a = 1.97\simeq 2$ confirming the quadratic scaling of $\Delta\tilde{C}$ in terms of $\epsilon$. The inset shows the linear fit of $\Delta\tilde{C}(0)$ as a function of the number of qubits $n$. The linear behaviour is again satisfied.} 
\label{fig:linear_e_n}
\end{figure}

\section{Numerical results}
\label{sec:numerics}

In the following we show numerically VCEM effectiveness in realistic scenarios as those presented in Sec. \ref{sec:realistic_scenarion}. All the simulations are performed with the software library Pennylane \cite{pennylane2022}. We focus in particular on graph states preparation.  

Given a graph with $n$ nodes and by calling $n(j)$ for each node $j = 1,\dots,n$, the set of all nodes connected to node $j$, the graph state stabilizers are given by 
\begin{equation}
\label{eq:graph_stabilizers}
 \hat{G}_i = \hat{X}_i \otimes \bigg(\otimes_{j\in n(i)}\hat{Z}_j\bigg)\,\, i = 1,\dots,n \,\ . 
\end{equation}
The stabilizer state associated to Eq.~\eqref{eq:graph_stabilizers}, which is named graph state, is
$|\psi_G\rangle = \hat{U}_G|0\rangle^{\otimes n}$ with the Clifford circuit $\hat{U}_G$ given by
\begin{equation}
    \hat{U}_G = \biggl(\bigotimes_{<i,j>\in E}\hat{CZ}_{ij}\biggr)\hat{\text{H}}^{\otimes n}\,\ ,
\end{equation}
where $E$ is the set of edges of the graph, $\hat{CZ}$ and $\hat{\text{H}}$ are respectively controlled-Z and Hadamard gates.

Here we consider a 2D rectangular graph with $n = 10$ and sides of length 2 and 5. We choose the following native gates basis $\hat{R}_x(\phi) = e^{-i\frac{\phi}{2}\hat{X}},\hat{R}_z(\phi)=e^{-i\frac{\phi}{2}\hat{Z}},\hat{R}_{zx}(\phi)= e^{-i\frac{\phi}{2}\hat{Z}\otimes\hat{X}}$ where $\hat{X}$, $\hat{Z}$ are respectively the x, z Pauli gates. In the following simulations we assume that given a qubit, the same kind of single-qubit native gate acting on it has the same random coherent error and correspondingly the same parameter. Analogously, given a couple of qubits, each $\hat{R}_{zx}$ gate acting on it has the same coherent error and parameter. Fig. \ref{fig:transpile} of Appendix \ref{sec:transpilation} shows respectively the transpilation of Hadamard and the CZ gates with coherent errors and parameters. 

First, we optimize the cost function in the case of coherent errors only. We sample randomly coherent errors in the interval $[-0.01,0.01]$. In Fig. \ref{fig:cost_functions_numerical} panel (a), we show the value of $C(\vec{\theta})$ as a function of the number of optimization steps $N_{iter}$. The red dashed line is the minimal value $C(\vec{\theta}_{opt}) = -10$. The insets represent the behaviour of $\varepsilon_l = \abs{\vec{\theta}_l + \vec{\epsilon}_l}$ with $l \in \{R_z, R_x, R_{zx}\}$. Here, $\vec{\theta}_l$ denotes the vector containing all the parameters associated with the gate $l$, while $\vec{\epsilon}_l$ collects the corresponding coherent errors. For the single-qubit rotations $\hat{R}_z$ and $\hat{R}_x$, these vectors include parameters and errors defined per qubit, whereas for the two-qubit gate $\hat{R}_{zx}$ they are defined over pairs of qubits. As shown, each $\varepsilon_l$ converges to zero with increasing number of optimization steps thus $\lim_{N_{iter}\to\infty}\vec{\theta}_l =-\vec{\epsilon}_l$.

Secondly, we study the modification introduced by incoherent errors by applying $m$-local Pauli maps after each circuit moment with $m=1,2$. Analogously, we sample randomly incoherent errors in the interval $[-0.01,0.01]$, while coherent ones are the same of the previous simulation. In panel (b) of Fig. \ref{fig:cost_functions_numerical} we show the behaviour of $\tilde{C}(\vec{\theta})$ and the red dashed line is the minimal value $\tilde{C}(\vec{\theta}_{opt}) = -5.0347723$. The insets in panel (b) have the same meaning as for panel (a).

The results show that VCEM is effective in realistic scenarios.

Moreover, our simulations confirm the scaling of the upper bound in Theorem \ref{theo:linear_upper_bound} as we show in Fig. \ref{fig:linear_e_n}. We simulated a 1D-linear graph state with 10 qubits and computed $\Delta\tilde{C}(0)$ by varying the order of magnitude $\epsilon$ of coherent errors and by keeping fixed incoherent errors in the same interval chosen above. Then, we performed a linear fit using a log-log scale which resulted in a slope $a = 1.97 \approx 2$. This indicates that $\Delta\tilde{C}$ grows quadratically with $\epsilon$ as expected. 
Additionally we computed $\Delta\tilde{C}(0)$ for the same graph state by increasing the number of qubits $n$ from 9 to 12. The data follow a linear trend, confirming the expected linear dependence on $n$.


\section{Conclusions and outlook}
This work introduces the Variational Coherent Error Mitigation (VCEM) framework, which effectively suppresses coherent errors through stabilizer-based variational optimization. VCEM operates directly on native gate parameters and remains robust even under realistic incoherent noise. Numerical 
results confirm its scalability and accuracy. 
As an outlook, we can consider a further scenario with respect to the one in Sec. \ref{sec:realistic_scenarion} in which it could be possible to engineer general incoherent errors in order to be modeled as a global depolarizing map after each circuit
moment. In that case the results in Sec. \ref{subsec:global_dep_map} would be applied, thus VCEM would work in presence of incoherent errors without approximations. A technique that in principle goes in that direction is Clifford twirling \cite{wallman2016, Cai_2019, olivia_dimatteo}, for which a general incoherent error map can be reduced to Eq.~\eqref{eq:global_dep}  (see appendix \ref{sec:twirling}).

However, at present this
scenario is considered less feasible in practice because circuit implementation of Clifford twirling can have a significant depth \cite{Bravyi_2021} and could itself introduce an unacceptable level of noise.

Looking ahead, further exploration of these kind of techniques and their practical implementation may open promising avenues to improve VCEM effectiveness enabling robust coherent error mitigation even with large-scale noisy stabilizer states.

Finally, the techinical points raised in this work are also very interesting, as they suggest that, for certain classes of parametrized quantum circuits, the location of the minimum $\vec{\theta}_{opt}$ rather than the optimal value of the cost function $C(\vec{\theta}_{opt})$ could be better exploited, as it is an intrinsically noise-resistant property of the model.

\section*{Acknowledgments}
G.C., G.D.B. A.B. and M.V. acknowledge support from University of Trieste and INFN. G.D.B. acknowledge support from a research fellowship funded by the Autonomous Region of Friuli Venezia Giulia under the Regional Programme FSE+ 2021/2027 (Operation code 2024/1424, CUP J93C23001490008). M.V. acknowledge support from a research fellowship funded by the Autonomous Region of Friuli Venezia Giulia under the Regional Programme FSE+ 2021/2027 (Operation code 2023/2143/1, CUP J93C23001490008). A.B. acknowledges support from the PNRR MUR project
PE0000023 - NQSTI and the EU EIC Pathfinder project QuCoM (GA 101046973).

The numerical simulations have been
performed within the PennyLane library~\cite{pennylane2022}.

\appendix

\onecolumngrid

\section{Useful properties and intermediate results}
\label{appendix:notation_and_lemmas}

We start the section by introducing the concept of Pauli map, and later we show some of their useful properties.

\begin{Definition}[Pauli map] 
\label{def:dpm}
A quantum channel $\mathcal{P}$ acting on $n$ qubits is a Pauli map if it can be decomposed as
\begin{equation}
    \mathcal{P}\hat{\rho} = \sum_{j=0}^{4^n-1} p_j\hat{P}_j\hat{\rho} \hat{P}_j,
\end{equation}
where $p_j \geq 0$, $\sum_jp_j=1$; $\hat{P}_j$ are Pauli strings and $\hat{\rho}$ is a generic density matrix of $n$ qubits.
\end{Definition}

It can be shown that any generic noise channel $\mathcal{E}$ can be reduced to a  Pauli map $\mathcal{P}$ by a procedure called Pauli twirling (see Appendix \ref{sec:twirling}). One of the main theoretical advantages of Pauli maps, is that they are self-adjoint, i.e. $\mathcal{P}=\mathcal{P}^\dagger$. This can easily be shown by noting that each term $\hat{P}_j$ is itself self-adjoint, i.e. $\hat{P}_j=\hat{P}_j^\dagger$. 
Additionally, they act on Pauli strings simply by rescaling them. This idea is formalized in the following Lemma.

\begin{Lemma}[Action of Pauli maps on Pauli strings]\label{lemma:action_pauli_noise_string}
    Given a Pauli map $\mathcal{P}$ and a Pauli string $\hat{S}$, the action $\mathcal{P}\hat{S}$ is given by
    \begin{equation}
       \mathcal{P}\hat{S} = \chi \hat{S},
    \end{equation}
    for some $\chi\in \mathbb{R}$, $|\chi| \leq 1$. In particular, $\chi = 1-2\sum_{j:\{S, P_j\} = 0} p_j$.
\end{Lemma}

\begin{proof}
    The proof is based on the observation that any pair of Pauli strings either commutes or anti-commutes with each other. More specifically, given two Pauli strings $\hat{P}_j$ and $\hat{S}$, then if $[\hat{P}_j,\hat{S}]\neq 0 \Rightarrow \{\hat{P}_j,\hat{S}\}=0$ and vice versa, as it can be directly verified using the commutation relations of Pauli matrices. Consequently it follows that
    \begin{equation}
        \hat{P}_j\hat{S}\hat{P}_j = \begin{cases}
            \hat{S}\quad \text{if} \quad[\hat{S}, \hat{P}_j]=0\\
            -\hat{S} \quad\text{if}\quad\{\hat{S}, \hat{P}_j\}=0
        \end{cases}
    \end{equation}
By following Def. \ref{def:dpm}, we have
    \begin{equation}
        \mathcal{P}\hat{S} = \sum_{j=0}^{4^n-1}p_j\hat{P}_j\hat{S}\hat{P}_j = \left(\sum_{j:[\hat{S},\hat{P}_j]=0} p_j-\sum_{j:\{\hat{S},\hat{P}_j\}=0} p_j \right)\hat{S} = \left(1-2\sum_{j:\{\hat{S},\hat{P}_j\}=0} p_j\right)\hat{S} = \chi \hat{S},
    \end{equation}
    where the condition $|\chi| \leq 1$ is ensured by the relation
    \begin{equation}
        0\leq\sum_{i:\{S,P_j\}=0} p_j\leq 1.
    \end{equation}
\end{proof}

Another useful property of Pauli maps, is that they preserve their structure under conjugation with Clifford unitaries. More specifically, we have the following Lemma.

\begin{Lemma}[Pauli and unitary Clifford maps]
\label{lemma:commute_pauli_clifford}
Given a Pauli map $\mathcal{P}$ and a Clifford unitary $\mathcal{U}$, then we have
\begin{equation}
    \mathcal{U}\mathcal{P}\mathcal{U}^{\dagger} = \mathcal{P}', \;\; \text{or equivalently} \;\; \mathcal{U}\mathcal{P} = \mathcal{P}'\mathcal{U},
\end{equation}
where the new map $\mathcal{P}'$ is still of the form of a Pauli map.
\end{Lemma}

\begin{proof}

We prove the statement by showing that $\mathcal{U}\mathcal{P}\mathcal{U}^{\dagger}\hat{\rho} = \mathcal{P}'\hat{\rho}$ $\forall \hat{\rho}$. In particular, we use the definitions of both $\mathcal{P}$ and $\mathcal{U}$ to achieve the following chain of equalities

\begin{equation}
\label{eq:clifford_swap_dim}
\begin{aligned}
    \mathcal{U}\mathcal{P}\mathcal{U}^{\dagger}\hat{\rho} &=\hat{U}\left(\sum_{j=0}^{4^n-1} p_j \hat{P}_j\left(\hat{U}^{\dagger} \hat{\rho}\hat{U}\right)\hat{P_j}\right) \hat{U}^{\dagger} = \sum_{j=0}^{4^n-1} p_j \left(\hat{U}\hat{P}_j\hat{U}^{\dagger}\right) \hat{\rho}\left(\hat{U}\hat{P_j}\hat{U}^{\dagger}\right)=
    \sum_{j=0}^{4^n-1} p_j \hat{P}'_j\hat{\rho}\hat{P}'_j  =
    \sum_{j=0}^{4^n-1} p'_{j} \hat{P}_j\hat{\rho}\hat{P}_j =
    \mathcal{P}'\rho,
\end{aligned}
\end{equation}
where we used the property that, for any Clifford unitary $\hat{U}$ and Pauli string $\hat{P}_j$, $\hat{U}\hat{P}_j\hat{U}^\dagger = \hat{P}_j'$ is still a Pauli string.
\end{proof}

A similar property also holds for the generators of unitary superoperators.

\begin{Lemma}
\label{lemma:commute_h_unitary}
    Given a generic unitary superoperator $\mathcal{U}$ and a generic generator $\mathcal{H}$ defined by the Hamiltonian $\hat{H}$, as $\mathcal{H}\hat{\rho} = [\hat{H}, \hat{\rho}]$, we have
    \begin{equation}
        \mathcal{U}\mathcal{H}\mathcal{U}^{\dagger} = \mathcal{H}'\;\; \text{or equivalently} \;\; \mathcal{U}\mathcal{H} = \mathcal{H}'\mathcal{U},
    \end{equation}
    where $\mathcal{H}'$ is defined as $\mathcal{H}'\hat{\rho} = [\mathcal{U}\hat{H}, \hat{\rho}]$.
\end{Lemma}

\begin{proof}
This Lemma follows immediately from the definitons of $\mathcal{U}$ and $\mathcal{H}$ by the following chain of equalities:

\begin{equation}
    \begin{aligned}
        \mathcal{U}\mathcal{H}\mathcal{U}^{\dagger}\hat{\rho} &= \hat{U}[\hat{H},\hat{U}^\dagger\hat{\rho} \hat{U}]\hat{U}^\dagger=\hat{U}\hat{H}\hat{U}^\dagger\hat{\rho} \hat{U}\hat{U}^\dagger-\hat{U}\hat{U}^\dagger\hat{\rho} \hat{U}\hat{H}\hat{U}^\dagger\\
        &=\hat{U}\hat{H}\hat{U}^\dagger \hat{\rho} -\hat{\rho} \hat{U}\hat{H}\hat{U}^\dagger = [\hat{U}\hat{H}\hat{U}^\dagger, \hat{\rho}] = [\mathcal{U}\hat{H}, \hat{\rho}] = \mathcal{H}'\hat{\rho}   
    \end{aligned}
\end{equation}

where $\mathcal{H}'$ is defined as $\mathcal{H}'\hat{\rho} = [\mathcal{U}\hat{H}, \hat{\rho}]$.
\end{proof}

We now show some useful upper bounds to compute derivatives in later sections.

\begin{Lemma} \label{lemma:first_derivative_upperbound}
Given a quantum state $\hat{\rho}$, a Pauli string $\hat{S}$, a Pauli map $\mathcal{P}$ and a superoperator $\mathcal{H}$, defined by the Hamiltonian $\hat{H}$, as $\mathcal{H}\hat{\rho} = [\hat{H}, \hat{\rho}]$, we have the following upper bound

\begin{equation}
    \left|\Tr(\hat{S}\mathcal{H}\mathcal{P\hat{\rho}})\right| \leq \|\hat{H}\|_2\|[\hat{S}, \hat{\rho}]\|_2\, ,
\end{equation}
where $\|\cdot\|_2$ is the operator 2-norm \cite{nielsen2000quantum}. In particular, if $[\hat{S},\hat{\rho}] = 0$, then $\Tr(\hat{S}\mathcal{H}\mathcal{P}\hat{\rho}) = 0$.
\end{Lemma}

\begin{proof}
    We begin our proof by expanding the Hamiltonian $\hat{H}$ in the Pauli basis, namely $\hat{H} = \sum_k h_k\hat{P}_k$. Using this decomposition and since $\Tr(\hat{S}\mathcal{H}\,\cdot\,)=\Tr(\mathcal{H}^\dagger\hat{S}\, \cdot\,)=\Tr( [\hat{S}, \hat{H}]\,\cdot\,)$ we get

    \begin{equation}
    \Tr([\hat{S}, \hat{H}]\mathcal{P\hat{\rho}}) = \sum_k h_k \Tr([\hat{S}, \hat{P}_k]\mathcal{P}\hat{\rho}).
    \end{equation}

    Note that, since both $\hat{S}$ and $\hat{P}_k$ are Pauli strings, also their commutator has the same structure. This can be easiliy derived from the defining relations of Pauli matrices. By Lemma \ref{lemma:action_pauli_noise_string}, we have that $\mathcal{P}[\hat{S}, \hat{P}_k] = \chi_k [\hat{S}, \hat{P}_k]$. In particular, this gives

    \begin{equation}
        \sum_k h_k \Tr([\hat{S}, \hat{P}_k]\mathcal{P}\hat{\rho})  = \sum_k \chi_kh_k \Tr([\hat{S}, \hat{P}_k]\hat{\rho}) =   \sum_k \chi_kh_k \Tr(\hat{P}_k[\hat{S},\hat{\rho}]) = \Tr{\left(\sum_k \chi_kh_k \hat{P}_k\right)[\hat{S},\hat{\rho}]} = \Tr(\hat{H}'[\hat{S},\hat{\rho}]),
    \end{equation}
    where we used cyclicity and linearity properties of the trace and used the shorthand notation $\hat{H}' = \sum_k h'_k \hat{P}_k$, with $h'_k = \chi_kh_k$.
    Recall that, by definiton of the $2$-norm, we have the handy relation $\|\hat{H}\|^2_2 = d\sum_k h_k^2$, where $d$ is the dimension of the system on which $\hat{H}$ acts on. Together with Lemma \ref{lemma:action_pauli_noise_string}, which ensures that $\chi_k^2\leq 1$, we can show that $h_k'^2 = h_k^2\chi_k^2 \leq h_k^2$, which implies $\|\hat{H}'\|^2_2 \leq \|\hat{H}\|^2_2$. Finally, invoking Cauchy-Schwarz inequality we have
    \begin{equation}
        \left|\Tr(\hat{H}'[\hat{S},\hat{\rho}])\right| < \|\hat{H}'\|_2\|[\hat{S},\hat{\rho}]\|_2 \leq \|\hat{H}\|_2\|[\hat{S},\hat{\rho}]\|_2,
    \end{equation}
    which concludes the proof.
\end{proof}

We can easily extend Lemma \ref{lemma:first_derivative_upperbound} to a more generic superoperator $\alpha = \sum_{l=1}^L \mathcal{H}_l$, getting the immediate upper bound \begin{equation}
\label{eq:gradient}
    \left|\Tr(\hat{S}\alpha\mathcal{P}\hat{\rho})\right| \leq L \max_l \left(\|\hat{H}_l\|_2\right)\|[\hat{S},\hat{\rho}]\|_2.
\end{equation}

Another useful upper bound involving superoperators of the form of $\alpha$ is given in the following Lemma.

\begin{Lemma}
\label{lemma:second_derivative_upperbound}
Given a quantum state $\hat{\rho}$, a Pauli string $\hat{S}$, two Pauli maps $\mathcal{P}_1$ and $\mathcal{P}_2$, let's define two superoperators $\mathcal{\alpha}_1$ and $\mathcal{\alpha}_2$ by

\begin{equation}
    \mathcal{\alpha}_q = \sum_{l=1}^{L_q}\mathcal{H}_{ql}, \;\;\text{where} \;\; \mathcal{H}_{ql} \hat{\rho} = [\hat{H}_{ql}, \hat{\rho}]\, .
\end{equation}
Further we assume that each $\hat{S}$ and $\hat{H}_{ql}$ acts non-trivially on $m$ qubits, i.e. they are $m$-local, for both $q=1,2$, each $\hat{H}_{ql}$ acts on a different subset of qubits when varying $l$. Furthermore, we assume that $\|\hat{H}_{ql}\|_2^2 \le O(d)$, where $d=2^m$. Then we have the following upper bound

\begin{equation}
    \left|\Tr(\hat{S}\alpha_1\mathcal{P}_1\alpha_2\mathcal{P}_2\hat{\rho})\right| \leq \Upsilon N_c\le O(4^{2m-1})\, .
\end{equation}
where $\Upsilon = \max_{ql} \|\hat{H}_{ql}\|^2_2/d$ is a constant, and $N_c \le O(4^{2m-1})$ is an integer number.
\end{Lemma}

\begin{proof}
    We start by noting that, decomposing each Hamiltonian $\hat{H}_{ql}$ into its Pauli decomposition $\hat{H}_{ql} = \sum_k h_{qlk} \hat{P}_k$ and collecting terms which share the same $\hat{P}_k$, we can obtain the following Pauli decomposition for the action of $\alpha_q$, i.e.
    $\alpha_q\hat{\rho}= \sum_k \left(\sum_{l}h_{qlk}\right) [\hat{P}_k, \hat{\rho}] = \sum_k\left(\sum_{l}h_{qlk}\right)\mathcal{H}^{\text{\tiny (P)}}_k,$
    where $\mathcal{H}^{\text{\tiny (P)}}_k$ is defined by the corresponding $\hat{P}_k$. Since we assume that each $\hat{H}_{ql}$ acts non-trivially on distinct qubits, then all terms in $\sum_{l}h_{qlk}$ must vanish except when $k$ matches the support of $\hat{H}_{ql}$, which can only happen for one $l$. As there is no ambiguity in the choice of $\hat{H}_{ql}$, we hence can drop the index $l$ entirely, and simply write $h_{qk}$ instead.
    
    Using the more explicit form $\Tr(\hat{S}\mathcal{H}\,\cdot\,)=\Tr(\mathcal{H}^\dagger\hat{S}\, \cdot\,)=\Tr( [\hat{S}, \hat{H}]\,\cdot\,)$ we get

    \begin{equation}
    \begin{aligned}
    \left|\Tr(\hat{S}\alpha_1\mathcal{P}_1\alpha_2\mathcal{P}_2\hat{\rho})\right| &= \left|\sum_{kk'} h_{1k}h_{2k'} \Tr(\hat{S}\mathcal{H}^{\text{\tiny (P)}}_k\mathcal{P}_1\mathcal{H}^{\text{\tiny (P)}}_{k'}\mathcal{P}_2\hat{\rho})\right| \leq \sum_{kk'} |h_{1k}| \cdot|h_{2k'}| \cdot\left|\Tr([\hat{S},\hat{P}_k]\mathcal{P}_1\mathcal{H}^{\text{\tiny (P)}}_{k'}\mathcal{P}_2\hat{\rho})\right|\\
    &\leq\sum_{kk'} |h_{1k}|\cdot |h_{2k'}| \cdot\left|\Tr([[\hat{S},\hat{P}_k], \hat{P}_{k'}]\mathcal{P}_2\hat{\rho})\right| \leq \sum_{kk'} |h_{1k}|\cdot |h_{2k'}| \cdot\left|\Tr([[\hat{S},\hat{P}_k], \hat{P}_{k'}]\hat{\rho})\right|\\
    &\leq \max_k(|h_{1k}|)\cdot\max_{k'}(|h_{2k'}|)\cdot \sum_{kk'} \left|\Tr([[\hat{S},\hat{P}_k], P_{k'}]\hat{\rho})\right|\, ,
    \end{aligned}
    \end{equation}

    where in the third term of the first line and in the first term of the second line we used the ciclicity of the trace and Lemma \ref{lemma:action_pauli_noise_string} for $\mathcal{P}_1[\hat{S},\hat{P}_k]$ and $\mathcal{P}_2[[\hat{S},\hat{P}_k], \hat{P}_{k'}]$ by upper bounding the corresponding constants to 1. 
    
    Recalling the relation $\|\hat{H}_{ql}\|^2_2 = d\sum_k h^2_{qlk}$, where $d$ denotes the dimension of the system on which $\hat{H}_{ql}$ acts, we can note that, given a $k$, $h_{qk}^2 \leq \|\hat{H}_{ql}\|^2_2$. As a consequence, it is immediate to see that $d\max_k |h_{qk}| \leq \max_l \|\hat{H}_{ql}\|_2$, $\forall q$. Thanks to this observation, we can finally write the bound for the product
    \begin{equation}
        \max_k(|h_{1k}|)\cdot\max_{k'}(|h_{2k'}|) \leq \max_{ql}\left(\frac{\|\hat{H}_{ql}\|^2_2}{d}\right) \equiv \Upsilon \le O(1)\, .
    \end{equation}
    
    Furthermore, $\hat{S}$, $\hat{P}_k$ and $\hat{P}_{k'}$ are $m$-local by virtue of the locality of $\hat{H}_{ql}$. This implies that only a constant number $N_c$ of nested commutators 
$[[ \hat{S}, \hat{P}_k ], \hat{P}_{k'}]$ are non-vanishing. 
Indeed, each operator $\hat{S}_i$ commutes with all Pauli strings that act trivially on the same qubits. For all the others Pauli strings $\hat{P}_k$, since they act non-trivially on at most a constant number $m$ of qubits, the overlap between the supports of $\hat{S}_i$ and $\hat{P}_k$ is bounded by $4^{2m-1}$, yielding only a constant number of terms in the sum.

Using H\"older inequality, we have $\left|\Tr([[\hat{S},\hat{P}_k], \hat{P}_{k'}]\hat{\rho})\right| \leq \|[[\hat{S},\hat{P}_k], \hat{P}_{k'}]\|_\infty\|\hat{\rho}\|_1 = 1$, giving the final upper bound

    \begin{equation}
         \left|\Tr(\hat{S}\alpha_1\mathcal{P}_1\alpha_2\mathcal{P}_2\hat{\rho})\right| \leq  \Upsilon N_c \le O(4^{2m-1})\, .
    \end{equation}
\end{proof}

\begin{Definition}[Global depolarizing map] 
\label{def:ddm}
A quantum channel $\mathcal{D}$ acting on $n$ qubits is a global depolarizing map if it has the form 
\begin{equation}\label{eq:global_dep_general}
\mathcal{D}\hat{\rho} = (1-p) \hat{\rho}+\frac{p}{2^n} \mathbb{1}\, ,
\end{equation}
where $p$ is the probability of getting a totally depolarized state $\mathbb{1}/2^n$ and $\hat{\rho}$ is a generic density matrix of $n$ qubits. We notice that a global symmetric Pauli map, i.e Definition \ref{def:dpm} with $p_j = p $ $\forall j$, reduces to Eq.~\eqref{eq:global_dep_general}.
\end{Definition}

\begin{Lemma}(Effective depolarizing map at the end of the circuit).
\label{lemma:effective_dep}
Given a noisy quantum circuit $\mathcal{N}_{\mathcal{D}}$ as in Eq.~\eqref{eq:circuit_rho_global_dep}, where a global depolarizing map $\mathcal{D}_q$ as in Eq.~\eqref{eq:global_dep} acts after each circuit moment $\mathcal{U}_q$ of the circuit $\mathcal{U}_S=\prod_{q=1}^{M}\mathcal{U}_q$, then we show that

\begin{equation}
\mathcal{N}_{\mathcal{D}}\hat{\rho_0}=\mathcal{D}\hat{\rho}_S\, ,
\end{equation}
where $\hat{\rho}_S = \mathcal{U}_S\hat{\rho}_0$, $\hat{\rho}_0= (\ket{0}\bra{0})^{\otimes n}$ and $\mathcal{D}$ is an effective global depolarizing map as in Eq.~\eqref{eq:global_dep_general} acting at the end of the circuit $\mathcal{U}_S$ with probability $p' = 1-\prod_{q=1}^{M}(1-p_q)$.
\end{Lemma}

\begin{proof}
For each $q$ in $\mathcal{N}_{\mathcal{D}}$ we apply consecutively $\mathcal{D}_q$ to the right giving rise to a new depolarizing map
\begin{equation}
\begin{aligned}
\mathcal{N}_{\mathcal{D}}\hat{\rho}_0 &=\biggl(\prod_{q=2}^{M}\mathcal{D}_q\mathcal{U}_q\biggr)\hat{\rho}_0=(1-p_1)\biggl(\prod_{q=2}^{M}\mathcal{D}_q\mathcal{U}_q\biggr)\mathcal{U}_1\hat{\rho}_0 + \frac{p_1}{2^n}\mathbb{1}\\
&= (1-p_1)(1-p_2)\biggl(\prod_{q=3}^{M}\mathcal{D}_q\mathcal{U}_q\biggr)\mathcal{U}_2\mathcal{U}_1\hat{\rho}_0+ \frac{1-(1-p_1)(1-p_2)}{2^n}\mathbb{1}\\
&=\prod_{q=1}^{M}(1-p_q)\mathcal{U}_S\hat{\rho}_0 + \frac{1-\prod_{q=1}^{M}(1-p_q)}{2^n}\mathbb{1}\\
& = \mathcal{D}\mathcal{U}_S\hat{\rho_0}=\mathcal{D}\hat{\rho}_S\, ,
\end{aligned}
\end{equation}
where $\mathcal{D}$ is a global depolarizing map with probability $p'$.
\end{proof}

\section{Simple analytical examples of the findings in sections \ref{sec:VCEM} and \ref{sec:VCEM_incoherent}}
\label{sec:analytic}
Here we show the modification of the landscape of the cost function according to the cases presented in the main text, by considering simple stabilizer states. We focus on a two qubit GHZ state as showed in Fig. \ref{fig:lakers}.

\begin{figure}[H]
    \centering
    \includegraphics[width=0.5\linewidth]{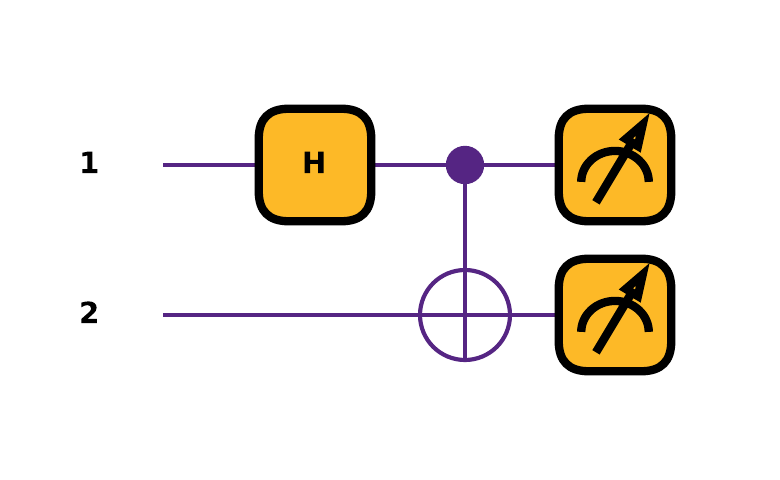}
    \caption{Two qubits GHZ circuit.}
    \label{fig:lakers}
\end{figure}

In order to easily visualize the cost function here we work with a single parameter. Thus, we choose the following simplified parametrization: the Hadamard gate $\hat{\text{H}} = \frac{1}{\sqrt{2}}\begin{pmatrix}1 &1\\
1 & -1\end{pmatrix}$ is assumed to be without coherent errors and the CNOT gate $\hat{CX}$ is parametrized as
\begin{equation}
\label{eq:UCX}
    \begin{aligned}
        \hat{U}_{CX}(\pi+\theta +\epsilon)= \begin{pmatrix}\mathbb{1} &0\\
0 & i\hat{R}_x(\pi+\theta +\epsilon)\end{pmatrix} 
    \end{aligned}\,\, ,
\end{equation}
where $\hat{U}_{CX}(\pi) = \hat{CX}$, $\hat{R}_x(\phi) = e^{-i\frac{\phi}{2}\hat{X}} = \begin{pmatrix} \cos(\frac{\phi}{2}) & -i\sin(\frac{\phi}{2})\\
-i\sin(\frac{\phi}{2}) & \cos(\frac{\phi}{2})\end{pmatrix}$ 
and $\hat{X}$, $\hat{Z}$ are respectively the x, z Pauli gates. The initial state is $|00\rangle$, this leads to
\begin{equation}
\begin{aligned}
\label{eq:ghz_state}
|\psi_{GHZ}(\theta +\epsilon)\rangle &= \hat{U}_{CX}(\pi+\theta +\epsilon)\hat{\text{H}}\mathbb{1}|00\rangle  = \frac{1}{\sqrt{2}}\biggl[|00\rangle+i\cos\Bigl(\frac{\pi+\theta+\epsilon}{2}\Bigr)|10\rangle+\sin\Bigl(\frac{\pi+\theta+\epsilon}{2}\Bigr)|11\rangle\biggr]\,\ .
\end{aligned}
\end{equation} 
The corresponding stabilizers operators are $\hat{S}_1 = \hat{X}\hat{X}$ and $\hat{S}_2 = \hat{Z}\hat{Z}$.

\subsection{Absence of incoherent errors}

By using Eq.~\eqref{eq:ghz_state} and Eq.~\eqref{eq:cost_function} we get
\begin{equation}
\label{eq:ghz_cost_func}
    C(\theta) = -\sin\Bigl(\frac{\pi+\theta +\epsilon}{2}\Bigr) -\sin^2\Bigl(\frac{\pi+\theta +\epsilon}{2}\Bigr) \,\ .
\end{equation}
We plot the behavior of Eq.~\eqref{eq:ghz_cost_func} in panel (a) of Fig. \ref{fig:cost_functions_GHZ}. The cost function displays a global minimum in $\theta = -\epsilon$.
\subsection{Pauli map at the end of the circuit}
At this point we assume the presence of incoherent errors as in Sec. \ref{subsec:pauli_maps_end_circ}. As an example we apply the following Pauli map

\begin{equation}\label{eq:pauli_map_XZ}
    \mathcal{P}\hat{\rho} =p_0 \hat{\rho}+p_1 \hat{X}\hat{X} \hat{\rho}\hat{X}\hat{X}+p_2 \hat{Z}\hat{Z} \hat{\rho}\hat{Z}\hat{Z}+p_3 \hat{Z}\hat{\mathbb{1}} \hat{\rho}\hat{Z}\mathbb{1}+p_4 \hat{\mathbb{1}}\hat{Z} \hat{\rho}\mathbb{1}\hat{Z}+p_5 \hat{X}\hat{\mathbb{1}} \hat{\rho}\hat{X}\hat{\mathbb{1}}+p_6 \hat{\mathbb{1}}\hat{X} \hat{\rho}\hat{\mathbb{1}}\hat{X}\,,
\end{equation}
where $p_1 =p_5=p_6=0.01$, $p_2 = 0.02$, $p_3 =p_4 = 0.2$ and $p_0 =1-\sum_{j=1}^{6}p_j = 0.55 $. 

By using Eq.~\eqref{eq:noisy_cost_function} we obtain
\begin{equation}
    \tilde{C}^{ \scaleto{(\mathcal{P})}{6pt}}(\vec{\theta}) = -(1-2\Gamma_1)\sin\Bigl(\frac{\pi+\theta +\epsilon}{2}\Bigr) -(1-2\Gamma_2)\sin^2\Bigl(\frac{\pi+\theta +\epsilon}{2}\Bigr) \,\ ,
\end{equation}
where $\Gamma_1 = p_3 + p_4$ and $\Gamma_2 = p_5+p_6$. The behaviour of this modified cost function is shown in panel (b) of Fig. \ref{fig:cost_functions_GHZ}. We notice that in the latter the local minima of the modified cost function $ \tilde{C}^{ \scaleto{(\mathcal{P})}{6pt}}$ approach the global minimum for increasing $\Gamma_1$. Despite this fact by starting the optimization in $\theta = 0$ we are sure to be close to the minimum corresponding to the target optimal parameter $\theta = -\epsilon$.

\subsection{Global depolarizing after each circuit moment}
We assume incoherent errors as in section \ref{subsec:global_dep_map} with $p_q=p$ $\forall q$. The circuit preparing the GHZ state in Fig. \ref{fig:lakers} has $M = 2$ moments and thus according to Eq.~\eqref{eq:cost_function_dep_momenta} the cost function is modified to
\begin{equation}
    \tilde{C}^{ \scaleto{(\mathcal{D})}{6pt}}(\theta) =(1-p)^2C(\theta)\,\ .
\end{equation}
where $p = 0.2$. This is drawn in panel (c) of Fig. \ref{fig:cost_functions_GHZ}. We notice that in this case the global minimum of the modified cost function remains global and in general cannot become local because the cost function is just reshaped by a constant.
\begin{figure*}[htp]
    \begin{minipage}{0.329\textwidth}
        \includegraphics[width=\textwidth]{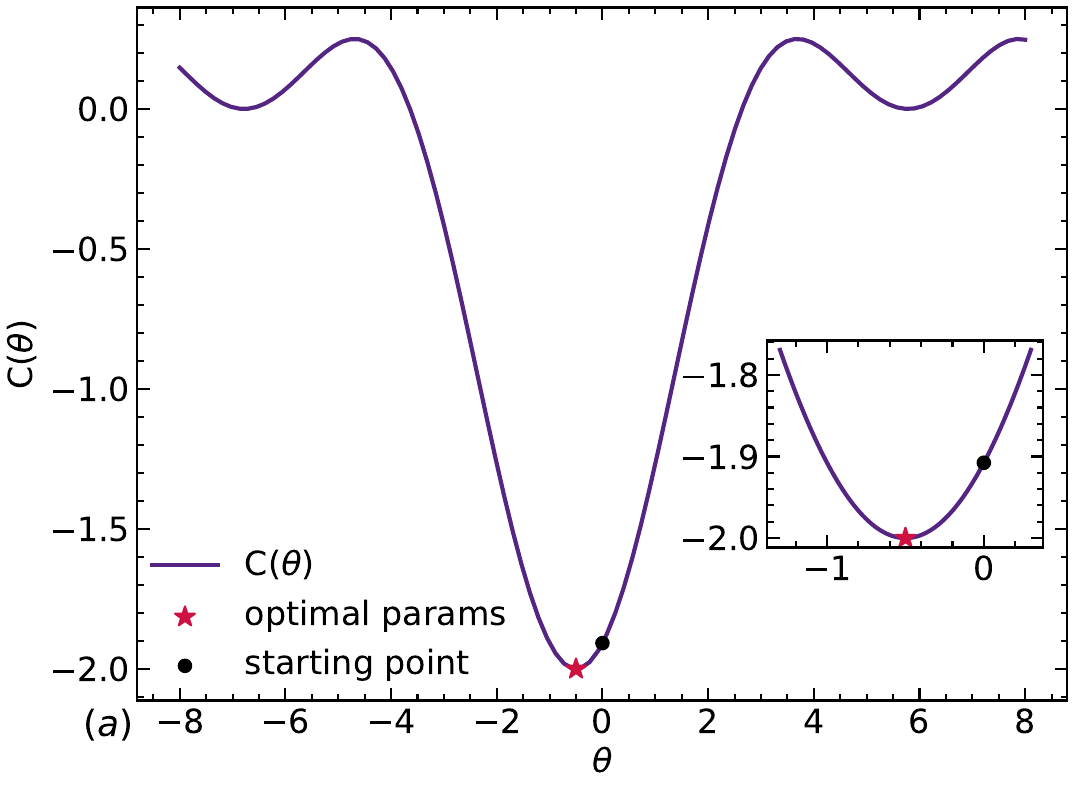}
    \end{minipage}
    \hfill%
    \begin{minipage}{0.329\textwidth}
        \includegraphics[width=\textwidth]{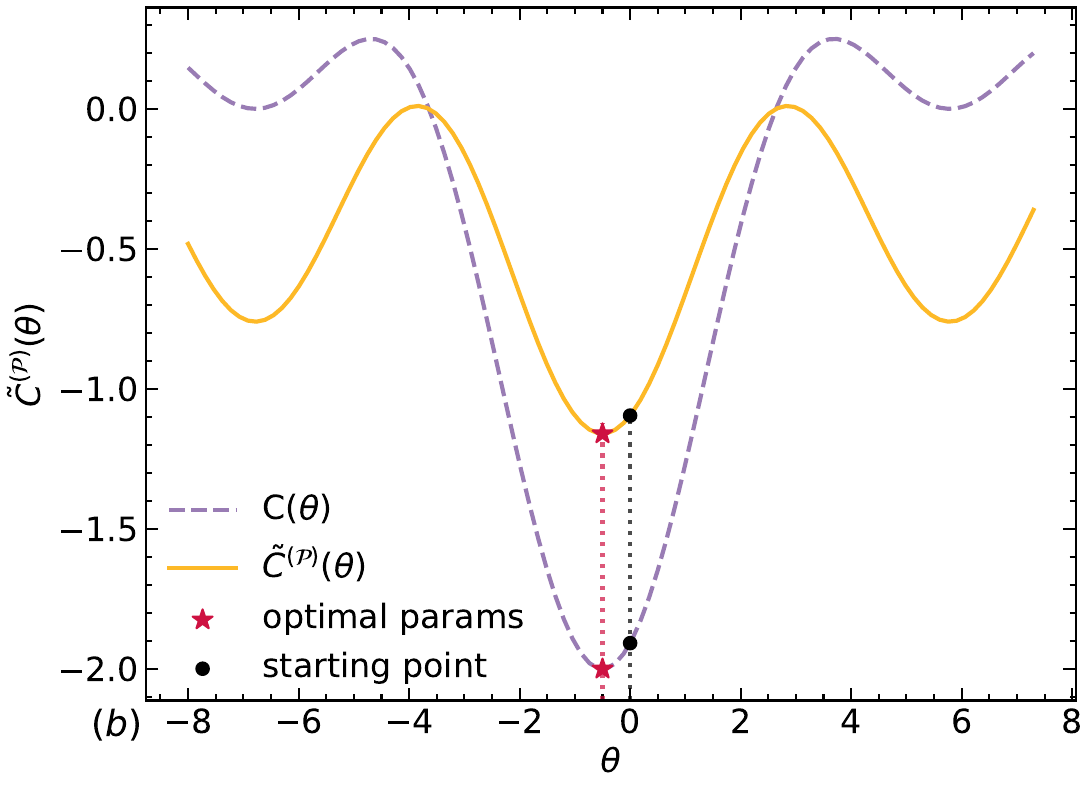}
    \end{minipage}
    \begin{minipage}{0.329\textwidth}
        \includegraphics[width=\textwidth]{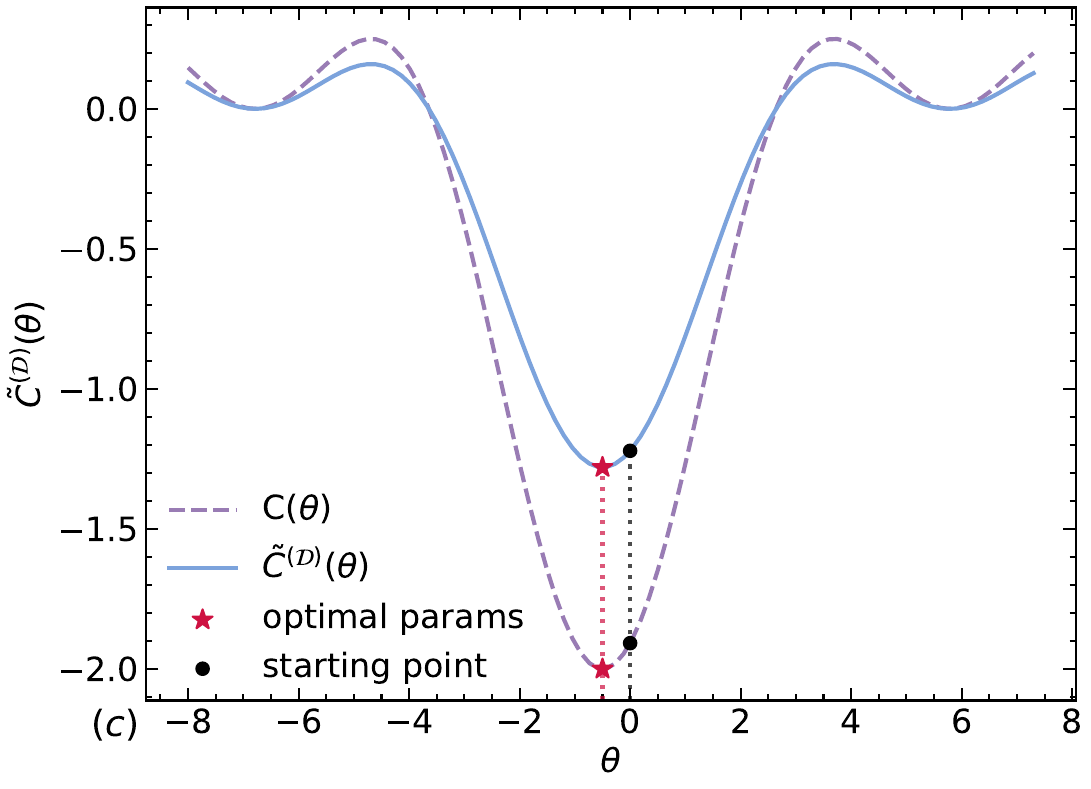}
    \end{minipage}
\caption{GHZ state cost function landscape for different noise scenarios. In panel (a) we show in purple the cost function in absence of incoherent noises. We choose the value of coherent error $\epsilon = 0.5$ for better visualization. The red star marks the optimal parameter $\theta=-\epsilon$ corresponding to the minimum and the black dot marks the initial parameter $\theta=0$ of the optimization in order to be sure to start inside a convex region $\mathcal{B}$ around the minimum. A zoom of $\mathcal{B}$ is shown in the inset. In panel (b) we show in gold the modified cost function in presence of Pauli errors at the end of the circuit. The purple dashed curve is the same of panel (a) and we highlight the translation of the minimum and of the starting point with respectively vertical red and black dotted lines. In panel (c) we show in blue the modified cost function in presence of global depolarizing errors after each circuit moment. The meaning of the remaining symbols and lines is the same of panel (b).}
\label{fig:cost_functions_GHZ}
\end{figure*}

\subsection{Pauli maps after each circuit momentum}

Finally we assume incoherent errors as in section \ref{sec:local_pauli_map}. Here we change the parametrization of the CNOT gate because the assumption that each native gate in the circuit should be able to be expressed as $\hat{U}(\phi) = e^{-i\phi\hat{H}}$ is necessary and Eq.~\eqref{eq:UCX} does not satisfy it. In the previous examples we used Eq.~\eqref{eq:UCX} because in these cases the aforementioned assumption is not necessary and moreover the former choice leads to a better visualization of the differences between the plots in Fig. \ref{fig:cost_functions_GHZ}. Thus now we choose the following $\hat{U}_{CX}$ gate

\begin{equation}
\label{eq:UCX_2}
    \begin{aligned}
        \hat{U}_{CX}(\pi+\theta +\epsilon)= e^{-i(\pi + \theta + \epsilon)\hat{H}_{CX}}= \begin{pmatrix}\mathbb{1} &0\\
0 & e^{\frac{i}{2}(\pi + \theta + \epsilon)}\hat{R}_x(\pi+\theta +\epsilon)\end{pmatrix} 
    \end{aligned}\,\, ,
\end{equation}
where $\hat{H}_{CX} = \begin{pmatrix}0 &0\\
0 & \frac{1}{2}(\hat{X}-\mathbb{1})\end{pmatrix}$ and $\hat{U}_{CX}(\pi+\theta +\epsilon)=\hat{U}_{CX}(\theta +\epsilon)\hat{U}_{CX}(\pi)$.

We apply the following 2-local and 1-local Pauli maps
\begin{equation}
\begin{aligned}
\label{eq:local_pauli_maps_XYZ}
&\mathcal{P}_1\hat{\rho} =p_0^{(1)} \hat{\rho}+p^{(1)}_1 \hat{Z}\hat{\mathbb{1}} \hat{\rho}\hat{Z}\mathbb{1}+p_2^{(1)} \hat{X}\hat{\mathbb{1}} \hat{\rho}\hat{X}\hat{\mathbb{1}}+p_3^{(1)} \hat{Y}\hat{\mathbb{1}} \hat{\rho}\hat{Y}\hat{\mathbb{1}}\,,\\
&\mathcal{P}_2\hat{\rho} =p_0^{(2)} \hat{\rho}+p_1^{(2)} \hat{Z}\hat{\mathbb{1}} \hat{\rho}\hat{Z}\mathbb{1}+p_2^{(2)} \hat{\mathbb{1}}\hat{Z} \hat{\rho}\mathbb{1}\hat{Z}+p_3^{(2)} \hat{X}\hat{\mathbb{1}} \hat{\rho}\hat{X}\hat{\mathbb{1}}+p_4^{(2)} \hat{\mathbb{1}}\hat{X} \hat{\rho}\hat{\mathbb{1}}\hat{X}\,,\\
\end{aligned}
\end{equation}
where $p_1^{(1)} =0.01$, $p_2^{(1)} = 0.02$, $p_3^{(1)} = 0.17$, $p_0^{(1)} =1-\sum_{j=1}^{3}p_j^{(1)} = 0.80$, $p_1^{(2)} =0.01$, $p_2^{(2)} = 0.1$, $p_3^{(2)} = 0.18$, $p_4^{(2)} = 0.01$, $p_0^{(2)} =1-\sum_{j=1}^{4}p_j^{(2)} = 0.70$ and $\hat{Y}$ is the Pauli y-gate.

By applying the procedure of Appendix \ref{appendix:remainder} as in Eq.~\eqref{eq:send_noise_end_circuit_superoperator} we find
\begin{equation}
\begin{aligned}
&\mathcal{P}_2\Bigl(\hat{U}_{CX}(\pi +\theta +\epsilon)\mathcal{P}_1(\hat{\text{H}}\mathbb{1}\hat{\rho}_0\hat{\text{H}}\mathbb{1})\hat{U}^{\dagger}_{CX}(\pi +\theta +\epsilon)\Bigr)\\
&=\mathcal{P}\Bigl(\hat{U}_{CX}(\pi +\theta +\epsilon)\hat{\text{H}}\mathbb{1}\hat{\rho}_0\hat{\text{H}}\mathbb{1}\hat{U}^{\dagger}_{CX}(\pi +\theta +\epsilon)\Bigr) + \mathcal{R}\hat{\rho}_0=\\
&=\mathcal{P}\hat{\rho}_{GHZ}(\theta+\epsilon) + \mathcal{R}\hat{\rho}_0\, ,
\end{aligned}
\end{equation}
where $\hat{\rho}_0 =|00\rangle\langle 00|$, $\hat{\rho}_{GHZ}(\theta+\epsilon) = |\psi_{GHZ}(\theta +\epsilon)\rangle\langle\psi_{GHZ}(\theta +\epsilon)|$ with $|\psi_{GHZ}(\theta +\epsilon )\rangle$ defined in Eq.~\eqref{eq:ghz_state}, and
\begin{equation}
\begin{aligned}
\mathcal{P}\hat{\rho} = \mathcal{P}_2\mathcal{P}_1'\hat{\rho}&= p_0\hat{\rho}+ p_1\hat{Z}\hat{\mathbb{1}}\hat{\rho}\hat{Z}\hat{\mathbb{1}}+p_2 \hat{\mathbb{1}}\hat{Z}\hat{\rho}\hat{\mathbb{1}}\hat{Z}+ p_3 \hat{X}\hat{\mathbb{1}}\hat{\rho}\hat{X}\hat{\mathbb{1}}+ p_4 \hat{\mathbb{1}}\hat{X}\hat{\rho}\hat{\mathbb{1}}\hat{X}+p_5 \hat{Y}\hat{\mathbb{1}}\hat{\rho}\hat{Y}\hat{\mathbb{1}}+p_6 \hat{Z}\hat{Z}\hat{\rho}\hat{Z}\hat{Z}+p_7 \hat{X}\hat{X}\hat{\rho}\hat{X}\hat{X}\\
&\quad +p_8 \hat{Y}\hat{Y}\hat{\rho}\hat{Y}\hat{Y}+p_9\hat{Z}\hat{X}\hat{\rho}\hat{Z}\hat{X}+p_{10}\hat{X}\hat{Y}\hat{\rho}\hat{X}\hat{Y}+p_{11}\hat{Y}\hat{X}\hat{\rho}\hat{Y}\hat{X}\, ,
\end{aligned}
\end{equation}
where $\mathcal{P}_1'\hat{\rho}= p_0^{(1)} \hat{\rho}+p^{(1)}_1 \hat{Z}\hat{\mathbb{1}} \hat{\rho}\hat{Z}\mathbb{1}+p_2^{(1)} \hat{X}\hat{X} \hat{\rho}\hat{X}\hat{X}+p_3^{(1)} \hat{Y}\hat{X} \hat{\rho}\hat{Y}\hat{X}$
and $p_0 = p_0^{(1)}p_0^{(2)} + p_1^{(1)}p_1^{(2)}$, $p_1 = p_1^{(1)}p_0^{(2)} + p_0^{(1)}p_1^{(2)}$, $p_2 = p_0^{(1)}p_2^{(2)}$, $p_3 = p_0^{(1)}p_3^{(2)} + p_2^{(1)}p_4^{(2)}$, $p_4 = p_2^{(1)}p_3^{(2)}+p_0^{(1)}p_4^{(2)}$, $p_5 =p_1^{(1)}p_3^{(2)}+p_3^{(1)}p_4^{(2)}$, $p_6 = p_1^{(1)}p_2^{(2)}$, $p_7= p_2^{(1)}p_0^{(2)}+p_3^{(1)}p_1^{(2)}$, $p_8 = p_3^{(1)}p_2^{(2)}$, $p_9 = p_3^{(1)}p_3^{(2)}+p_1^{(1)}p_4^{(2)}$, $p_{10} = p_2^{(1)}p_2^{(2)}$, $p_{11} = p_3^{(1)}p_0^{(2)}+p_2^{(1)}p_1^{(2)}$. 

Moreover, by defining $\hat{\rho}_{GHZ} = \hat{U}_{CX}(\pi)\hat{\text{H}}\mathbb{1}\hat{\rho}_0\hat{\text{H}}\mathbb{1}\hat{U}_{CX}^{\dagger}(\pi)$  we have
\begin{equation}
\begin{aligned}
&\mathcal{R}\hat{\rho}_0 = \mathcal{P}_2\Bigl(\hat{U}_{CX}(\theta +\epsilon)\mathcal{P}'_1(\hat{\rho}_{GHZ})\hat{U}_{CX}^{\dagger}(\theta +\epsilon)\Bigr)-\mathcal{P}\Bigl(\hat{U}_{CX}(\theta +\epsilon)\hat{\rho}_{GHZ}\hat{U}_{CX}^{\dagger}(\theta +\epsilon)\Bigr)\, .
\end{aligned}
\end{equation}
Now given that the group $\mathcal{\bf{S}}_2$, generated by $\hat{S}_1$, $\hat{S}_2$, is the subgroup of the 2-qubits Pauli group $
\mathcal{\bf{P}}_2$, we define $\hat{F}_i(\hat{S}_i)\equiv \hat{U}_{CX}^{\dagger}(\theta +\epsilon)\hat{S}_i\hat{U}_{CX}(\theta +\epsilon) = \nu_{S_i} \hat{S}_i+\sum_{w: \hat{P}_{i,w}\notin\mathcal{\bf{S}}_2} \nu_{i,w} \hat{P}_{i,w}$ then
\begin{equation}
\begin{aligned}
\Tr\bigl(\hat{S}_i\mathcal{R}\hat{\rho}_0\bigr)&= \Tr\Bigl(\hat{U}_{CX}^{\dagger}(\theta +\epsilon)\mathcal{P}_2(\hat{S}_i)\hat{U}_{CX}(\theta +\epsilon)\mathcal{P}'_1\hat{\rho}_{GHZ}\Bigr)-\Tr\Bigl(\hat{U}_{CX}^{\dagger}(\theta +\epsilon)\mathcal{P}(\hat{S}_i)\hat{U}_{CX}(\theta +\epsilon)\hat{\rho}_{GHZ}\Bigr)\\
&=(1-2\Gamma_i^{(\mathcal{P}_2)})\Tr\bigl(\hat{F}_i(\hat{S}_i)\mathcal{P}'_1\hat{\rho}_{GHZ}\bigr)-(1-2\Gamma_i^{(\mathcal{P})})\Tr\bigl(\hat{F}_i(\hat{S}_i)\hat{\rho}_{GHZ}\bigr)\\
&=\sum_{w: \hat{P}_{i,w}\notin\mathcal{\bf{S}}_2} \nu_{i,w}\Bigl[(1-2\Gamma_i^{(\mathcal{P}_2)})(1-2\Gamma_{i,w}^{(\mathcal{P}'_1)})-(1-2\Gamma_i^{(\mathcal{P})})\Bigr]\Tr\bigl(\hat{P}_{i,w}\hat{\rho}_{GHZ}\bigr) = 0\,\ ,
\end{aligned}
\end{equation}
where $\Gamma_i^{(\mathcal{P}_2)}$, $\Gamma_{i,w}^{(\mathcal{P}'_1)}$ and , $\Gamma_i^{(\mathcal{P})}$ are calculated as those in Eq.~\eqref{eq:senza_nome}.
Since for each $i$, $\hat{P}_{i,w}$ is not a product of the stabilizer operators $\hat{S}_i$, then it is orthogonal to $\hat{\rho}_{GHZ}$ that is a linear combination that contains only products of stabilizer operators. This fact is similar to what happens in general in Theorem \ref{theo:linear_upper_bound}, however in this simple example $\Delta \tilde{C}(\vec{\theta}) = -\Tr\bigl(\hat{S}_i\mathcal{R}\hat{\rho}_0\bigr)$ is exactly zero and not only up to a certain order. 

Finally, the modified cost function reads
\begin{equation}
\label{eq:nomi_finiti}
    \tilde{C}(\theta) = (1-\Gamma_1^{(\mathcal{P})}-\Gamma_2^{(\mathcal{P})})C(\theta) \,\ ,
\end{equation}
where by using Eq.~\eqref{eq:UCX_2} we have $C(\theta)=-2\sin^2\Bigl(\frac{\pi+\theta +\epsilon}{2}\Bigr)$, $\Gamma_1^{(\mathcal{P})} = p_1 + p_2 + p_5 + p_9 + p_{10} + p_{11}$ and  $\Gamma_2^{(\mathcal{P})} = p_3 + p_4 + p_5 + p_9$.

The behaviour of Eq.~\eqref{eq:nomi_finiti} is shown in Fig. \ref{fig:cost_functions_GHZ_local_pauli_maps}.

\begin{figure*}[htp]
\includegraphics[width=0.4\textwidth]{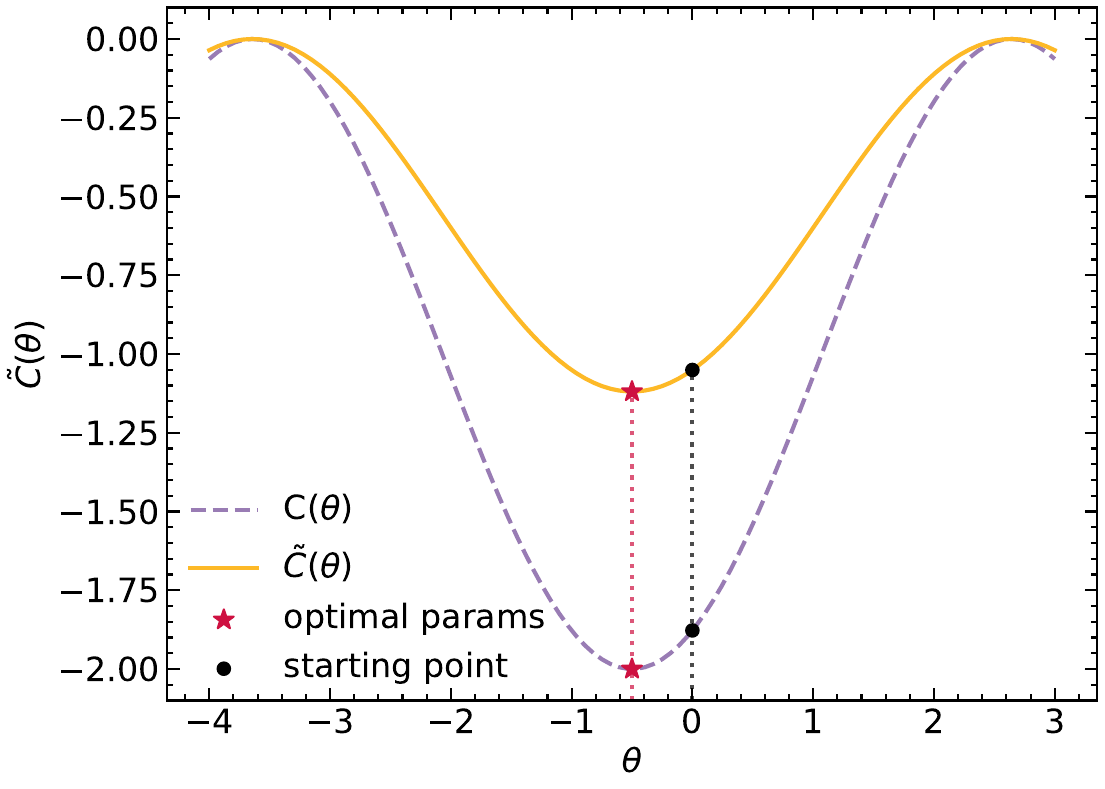}
\caption{GHZ state cost function landscape for local Pauli maps after each circuit moment. We show in dashed purple the cost function in absence of incoherent noises. We choose the value of coherent error $\epsilon = 0.5$ for better visualization. The red star marks the optimal parameter $\theta=-\epsilon$ corresponding to the minimum and the black dot marks the initial parameter $\theta=0$ of the optimization in order to be sure to start inside a convex region $\mathcal{B}$ around the minimum. The gold curve represents the modified cost function in presence of local Pauli errors after each circuit moment. We highlight the translation of the minimum and of the starting point with respectively vertical red and black dotted lines.}
\label{fig:cost_functions_GHZ_local_pauli_maps}
\end{figure*}

\section{Proof of Eq.~\eqref{eq:send_noise_end_circuit_superoperator}}
\label{appendix:remainder}

We can start showing that, for each moment of a given imperfect stabilizer circuit, we can separate the Clifford and non-Clifford contributions. In order to do that, we assume that all the native gates $\mathcal{G}_w$ in the circuit are such that 
\begin{equation}
    \mathcal{G}_w = e^{-i(\phi_w+\theta_w+\epsilon_w)\mathcal{H}_w} =    e^{-i(\theta_w+\epsilon_w)\mathcal{H}_w}e^{-i\phi_w\mathcal{H}_w} \equiv \mathcal{G}_w^{\text{\tiny (N)}}\mathcal{G}_w^{\text{\tiny (C)}}
\end{equation}
where each $\theta_w+\epsilon_w$ gives rise to the non-Clifford transformation $\mathcal{G}_w^{\text{\tiny (N)}}$ associated to coherent noise and each $\phi_w$ is chosen such that $\mathcal{G}_w^{\text{\tiny (C)}}$ is Clifford.

Moreover, given a moment $\mathcal{U}_q=\bigotimes_{l=1}^{L_q}\mathcal{G}_{ql}$ since $\mathcal{G}_{ql}$ act on different qubits, they commute with each other inside the same moment. This allows us to write

\begin{equation}
    \mathcal{U}_q=\bigotimes_{l=1}^{L_q}\mathcal{G}_{ql}^{\text{\tiny (N)}}\mathcal{G}_{ql}^{\text{\tiny (C)}}= \left(\bigotimes_{l=1}^{L_q}\mathcal{G}_{ql}^{\text{\tiny (N)}}\right)\left(\bigotimes_{l=1}^{L_q}\mathcal{G}_{ql}^{\text{\tiny (C)}}\right)\equiv\mathcal{U}_q^{\text{\tiny (N)}}\mathcal{U}_q^{\text{\tiny (C)}},
\end{equation}
where $\mathcal{U}_q^{\text{\tiny (N)}}$ and $\mathcal{U}_q^{\text{\tiny (C)}}$ are respectively the non-Clifford and Clifford parts. Such operations, are still of the form of exponentials, except now they are generated by the more complex superoperator 
\begin{equation}
    \label{eq:alpha_def}
    \bar{\alpha}_q = \sum_{l=1}^{L_q}(\theta_{ql} +\epsilon_{ql})\mathcal{H}_{ql}, \;\; \text{i.e.}\;\; \mathcal{U}_q^{\text{\tiny (N)}}=e^{-i\bar{\alpha}_q}.
\end{equation}
A similar structure also holds for $\mathcal{U}_q^{\text{\tiny (C)}}$, but we do not need it.

\begin{Lemma}(Circuit decomposition and shape of the remainder).
\label{lemma:circ_decomposition}
    The noisy quantum circuit $\mathcal{N}$ in Eq.~\eqref{eq:send_noise_end_circuit_superoperator} can be decomposed as
    \begin{equation}
    \label{eq:circuit_factorized}
         \mathcal{N}= \mathcal{P}\mathcal{U}_S +\mathcal{R}\, ,
    \end{equation}
    where $\mathcal{P}\equiv\left(\prod_{q=1}^M \mathcal{P}'_q\right)$ is an effective Pauli map and the remainder $\mathcal{R}$ is expressed as
    \begin{equation}
    \label{eq:form_of_the_remainder}
        \mathcal{R} = \sum_{p=1}^{M-1}\sum_{q>p}^{M}\mathcal{T}_{qp},\;\; \mathcal{T}_{qp} = \left(\prod_{s=p+1}^{M}\mathcal{P}'_{s} \right)\left(\prod_{r=q+1}^{M}\mathcal{U}_{r}^{\text{\tiny (N)}}\mathcal{U}_{r}^{\text{\tiny (C)}} \right)\Delta_{qp} \mathcal{U}_q^{\text{\tiny (C)}} \left(\prod_{r'=p}^{q-1}\mathcal{U}_{r'}^{\text{\tiny (N)}}\mathcal{U}_{r'}^{\text{\tiny (C)}} \right)\left(\prod_{t=1}^{p-1} \mathcal{P}_t\mathcal{U}_t^{\text{\tiny (N)}}\mathcal{U}_t^{\text{\tiny (C)}}\right)\, ,
    \end{equation}
    where $\Delta_{qp} \equiv [\mathcal{U}_q^{\text{\tiny (N)}}, \mathcal{P}'_{pq}]$ is the commutator between coherent and incoherent noise maps, and $\mathcal{P}'_{qp}$ are diagonal Pauli maps obtained by conjugation of $\mathcal{P}_p$ with the  Clifford moments between $p+1$ and $q$ (see Lemma \ref{lemma:commute_pauli_clifford}), namely
    \begin{equation}
    \label{eq:P_qp}
        \mathcal{P}'_{qp} = \left(\prod_{v=p+1}^q \mathcal{U}_v^{\text{\tiny (C)}}\right) \mathcal{P}_p \left(\prod_{v=p+1}^q \mathcal{U}_v^{\text{\tiny (C)}}\right)^{\dagger}.
    \end{equation}
    Finally, for the left most product, we use the shorthand notation $\mathcal{P}'_{s} = \mathcal{P}'_{Ms}$.
\end{Lemma}

\begin{proof}
The Lemma is proven giving a constructive approach. In particular, we will commute each $\mathcal{P}_q$ to the left of the expression on one step at the time, keeping track of the remainders. Each term of the double sum in Eq.~\eqref{eq:form_of_the_remainder} represents exactly this. We now give the blueprint on how to explicitly build this construction, using as an example the first few $p$ and $q$. Let's start from $\mathcal{N}$ where we highlight the last moments, namely
\begin{equation}
    \mathcal{N} = \mathcal{P}_M\mathcal{U}_M^{\text{\tiny (N)}}\mathcal{U}_M^{\text{\tiny (C)}}\mathcal{P}_{M-1}\mathcal{U}_{M-1}^{\text{\tiny (N)}}\mathcal{U}_{M-1}^{\text{\tiny (C)}}\prod_{q=1}^{M-2} \mathcal{P}_q\mathcal{U}_q^{\text{\tiny (N)}}\mathcal{U}_q^{\text{\tiny (C)}}\, .
\end{equation}
According to Lemma \ref{lemma:commute_pauli_clifford}, we can swap $\mathcal{U}_M^{\text{\tiny (C)}}$ and $\mathcal{P}_{M-1}$ at the price of having a new Pauli map $\mathcal{P}'_{M-1}$, obtained by conjugation, i.e. $\mathcal{U}_M^{\text{\tiny (C)}}\mathcal{P}_{M-1}= \mathcal{P}'_{M-1}\mathcal{U}_M^{\text{\tiny (C)}}$. However we cannot do this for the non-Clifford moment $\mathcal{U}_{M}^{\text{\tiny (N)}}$. In this case we swap them and keep track of the error $\Delta_{M,M-1}$, i.e. we use the relation
\begin{equation}
    \mathcal{U}_{M}^{\text{\tiny (N)}}\mathcal{P}'_{M-1} = \mathcal{P}'_{M-1}\mathcal{U}_{M}^{\text{\tiny (N)}} + [\mathcal{U}_{M}^{\text{\tiny (N)}}, \mathcal{P}'_{M-1}] = \mathcal{P}_{M-1}\mathcal{U}_{M}^{\text{\tiny (N)}} + \Delta_{M,M-1}.
\end{equation}
In this way, we get the first decomposition
\begin{equation}
\begin{aligned}
    \mathcal{N}&= \mathcal{P}_M \mathcal{P}'_{M-1}\mathcal{U}_M^{\text{\tiny (N)}}\mathcal{U}_M^{\text{\tiny (C)}}\mathcal{U}_{M-1}^{\text{\tiny (N)}}\mathcal{U}_{M-1}^{\text{\tiny (C)}}\prod_{q=1}^{M-2} \mathcal{P}_q\mathcal{U}_q^{\text{\tiny (N)}}\mathcal{U}_q^{\text{\tiny (C)}}\\
    &+ \mathcal{P}_M \Delta_{M,M-1} \mathcal{U}_M^{\text{\tiny (C)}}\mathcal{U}_{M-1}^{\text{\tiny (N)}}\mathcal{U}_{M-1}^{\text{\tiny (C)}}\prod_{q=1}^{M-2} \mathcal{P}_q\mathcal{U}_q^{\text{\tiny (N)}}\mathcal{U}_q^{\text{\tiny (C)}} \, .
\end{aligned}
\end{equation}
Note that the first term is one step closer to look like the first term in Eq.~\eqref{eq:circuit_factorized}, while the second term corresponds exactly to $\mathcal{T}_{M\,M-1}$.
By iterating this process on the first term, we can obtain all subsequent $\mathcal{T}_{qp}$ terms.
\end{proof}

Beyond the technicalities involved in keeping track of all the conjugations appearing in $\mathcal{P}'_{qp}$, the real power of Lemma \ref{lemma:circ_decomposition} is to provide the general structure of the remainder $\mathcal{R}$ in terms a sequence of Clifford, non-Clifford and Pauli maps, which will be key for the subsequent proofs.

\section{Stationarity of the solution in the presence of Pauli noise after each circuit moment}
\label{appendix:stationarity}

In this section we show that, despite the presence of Pauli noise in the system, the optimal set of parameters that we aim to find, namely $\vec{\theta} = - \vec{\epsilon}$ remains a stationary point. In order to do so, we prove Theorem \ref{theo:stationarity} of the main text.

\begin{proof}
For simplicity now we define $\gamma_k = \theta_k+\epsilon_k$ in the native gates inside the non-Clifford moments $\mathcal{U}_q^{\text{\tiny (N)}}$. In this manner, it is sufficient to compute the derivative $\partial_{\gamma_k}\mathcal{U}_q^{\text{\tiny (N)}}$ and evaluate it in $\vec{\gamma} = \vec{\theta}+\vec{\epsilon}=0$ to achieve the desired result of Eq.~\eqref{eq:zero_partial_derivative}. 

By employing the expression of $\mathcal{R}$ in Eq.~\eqref{eq:form_of_the_remainder} of Lemma \ref{lemma:circ_decomposition} we have
\begin{equation}
\label{eq:zero_partial_derivative_2}
    \partial_{\gamma_k} \Delta\tilde{C}(\vec{\gamma}-\vec{\epsilon}) |_{\vec{\gamma}=0} = -\sum_{i=1}^{n}  \Tr(\hat{S}\partial_{\gamma_k}\mathcal{R}|_{\vec{\gamma}=0} \hat{\rho}_0)=-\sum_{i=1}^{n}\sum_{p=1}^{M-1}\sum_{q>p}^{M}\Tr(\hat{S}\partial_{\gamma_k}\mathcal{T}_{qp}|_{\vec{\gamma}=0}\hat{\rho}_0)\, .
\end{equation}
At this point we focus on two distinct cases.

\begin{Obs}[Selection of the remainder terms] Given a parameter $\gamma_k$, for each moment $\mathcal{U}^{\text{\tiny (N)}}_q$ that does not depend on $\gamma_k$, then 
\begin{equation}
    \partial_{\gamma_k} \mathcal{T}_{qp}|_{\vec{\gamma}=0} \equiv 0\, .
\end{equation}
\end{Obs}

This follows from the definition of $\mathcal{T}_{qp}$, and in particular on the structure of $\Delta_{qp} = [\mathcal{U}_q^{\text{\tiny (N)}}, \mathcal{P}'_{qp}]$ arising from Lemma \ref{lemma:circ_decomposition}. Indeed, we assume that $\gamma_k$ is not contained in a moment $\mathcal{U}^{\text{\tiny (N)}}_q$, and without loss of generality and to ease the notation, we assume that instead, it is contained in the moment $\mathcal{U}^{\text{\tiny (N)}}_1$. Then we have

\begin{equation}
    \begin{aligned}
        \partial_{\gamma_k} \mathcal{T}_{qp}|_{\vec{\gamma}=0} &= \left(\prod_{s=p+1}^{M}\mathcal{P}'_{s} \right)\left(\prod_{r=q+1}^{M}\mathcal{U}_{r}^{\text{\tiny (N)}}\mathcal{U}_{r}^{\text{\tiny (C)}} \right)\Delta_{qp} \mathcal{U}_q^{\text{\tiny (C)}} \left(\prod_{r'=p}^{q-1}\mathcal{U}_{r'}^{\text{\tiny (N)}}\mathcal{U}_{r'}^{\text{\tiny (C)}} \right)\left(\prod_{t=2}^{p-1} \mathcal{P}_t\mathcal{U}_t^{\text{\tiny (N)}}\mathcal{U}_t^{\text{\tiny (C)}}\right) \mathcal{P}_1\left(\partial_{\gamma_k}\mathcal{U}_1^{\text{\tiny (N)}}\right)\mathcal{U}_1^{\text{\tiny (C)}}\Bigg|_{\vec{\gamma}=0}\\
        &= -i\left(\prod_{s=p+1}^{M}\mathcal{P}'_{s} \right)\left(\prod_{r=q+1}^{M}\mathcal{U}_{r}^{\text{\tiny (N)}}\mathcal{U}_{r}^{\text{\tiny (C)}} \right)\Delta_{qp} \mathcal{U}_q^{\text{\tiny (C)}} \left(\prod_{r'=p}^{q-1}\mathcal{U}_{r'}^{\text{\tiny (N)}}\mathcal{U}_{r'}^{\text{\tiny (C)}} \right)\left(\prod_{t=2}^{p-1} \mathcal{P}_t\mathcal{U}_t^{\text{\tiny (N)}}\mathcal{U}_t^{\text{\tiny (C)}}\right) \mathcal{P}_1\left(\mathcal{H}_k\mathcal{U}_1^{\text{\tiny (N)}}\right)\mathcal{U}_1^{\text{\tiny (C)}}\Bigg|_{\vec{\gamma}=0}\,,
    \end{aligned}
\end{equation}
where we used the definition $\mathcal{U}_q^{\text{\tiny (N)}} = e^{-i\bar{\alpha}_q}$, and $\bar{\alpha}_q = \sum_k \gamma_k\mathcal{H}_k$. Then, noting that, according to the same definition, $\bar{\alpha}_q = 0$ when evaluated at $\vec{\gamma}=0$, we get the simple result that 
$\mathcal{U}_q^{\text{\tiny (N)}} = e^{-i\bar{\alpha}_q} = \mathbb{1}\;\forall q$.

In particular this implies that, at $\vec{\gamma}=0$, $\Delta_{qp} \equiv 0$, since in that case $\Delta_{qp} = [\mathcal{U}_q^{\text{\tiny (N)}}, \mathcal{P}'_{qp}] = [\mathbb{1}, \mathcal{P}'_{qp}] \equiv0$, killing the whole product. Clearly from this reasoning, all other terms where the moment does not depend on $\gamma_k$ will similarly vanish.

The same however cannot be said when $\mathcal{U}^{\text{\tiny (N)}}_q$ depends on $\gamma_k$. Indeed, in that case the $\gamma_k$ dependence is included in $\Delta_{qp}$, and hence we need to compute $\partial_{\gamma_k}\Delta_{qp}$. In particular, we get

\begin{equation}\label{eq:derivative_of_delta}
\partial_{\gamma_k}\Delta_{qp}|_{\vec{\gamma}=0} = \partial_{\gamma_k} [\mathcal{U}_q^{\text{\tiny (N)}}, \mathcal{P}'_{qp}]\Big|_{\vec{\gamma}=0} = [\partial_{\gamma_k}\mathcal{U}_q^{\text{\tiny (N)}}, \mathcal{P}'_{qp}]\Big|_{\vec{\gamma}=0} = -i[\mathcal{H}_k \mathcal{U}_q^{\text{\tiny (N)}}, \mathcal{P}'_{qp}]\Big|_{\vec{\gamma}=0} = -i[\mathcal{H}_k, \mathcal{P}'_{qp}]\, ,
\end{equation}
which is in general non-vanishing. However, we will now prove that, when evaluated in the trace, such contributions do indeed vanish. As a first step, we make the following simple observation.

\begin{Obs}[Structure of the derivative]
    For any non-vanishing $\partial_{\gamma_k}\mathcal{T}_{qp}$, we can write its gradient contribution as

    \begin{equation}
        \label{eq:derivative_Tqp}
\Tr(\hat{S}_i\partial_{\gamma_k}\mathcal{T}_{qp}|_{\vec{\gamma}=0}\hat{\rho}_0) = -i\chi'_i\Tr{\hat{S}_i \left(\prod_{r=q+1}^M\mathcal{U}^{\text{\tiny (C)}}_r\right) [\mathcal{H}_k, \mathcal{P}'_{qp}] \mathcal{P}''_{qp}\left(\prod_{r'=1}^{q}\mathcal{U}^{\text{\tiny (C)}}_{r'}\right)\hat{\rho}_0}\, ,
    \end{equation}
    where $\mathcal{P}''_{p}$ is a Pauli map including the contribution of the first $p$ Pauli noise maps.
\end{Obs}

To show this, we recall the definition of $\mathcal{T}_{qp}$ in Eq.~\eqref{eq:form_of_the_remainder} and the previous result that when $\vec{\gamma}=0$, all non-Clifford moments become the identity. Using this we get

\begin{equation}
\begin{aligned}
    \partial_{\gamma_k} \mathcal{T}_{qp}|_{\vec{\gamma}=0} &= \left(\prod_{s=p+1}^{M}\mathcal{P}'_{s} \right)\left(\prod_{r=q+1}^{M}\mathcal{U}^{\text{\tiny (N)}}_{r}\mathcal{U}^{\text{\tiny (C)}}_{r} \right)\partial_{\gamma_k}\Delta_{qp} \mathcal{U}^{\text{\tiny (C)}}_q \left(\prod_{r'=p}^{q-1}\mathcal{U}^{\text{\tiny (N)}}_{r'}\mathcal{U}^{\text{\tiny (C)}}_{r'} \right)\left(\prod_{t=1}^{p-1} \mathcal{P}_t\mathcal{U}^{\text{\tiny (N)}}_t\mathcal{U}^{\text{\tiny (C)}}_t\right)\Bigg|_{\vec{\gamma}=0}\\
    &= \left(\prod_{s=p+1}^{M}\mathcal{P}'_{s} \right)\left(\prod_{r=q+1}^{M}\mathcal{U}^{\text{\tiny (C)}}_{r} \right)\partial_{\gamma_k}\Delta_{qp}\big|_{\vec{\gamma}=0} \mathcal{U}^{\text{\tiny (C)}}_q \left(\prod_{r'=p}^{q-1}\mathcal{U}^{\text{\tiny (C)}}_{r'} \right)\left(\prod_{t=1}^{p-1} \mathcal{P}_t\mathcal{U}^{\text{\tiny (C)}}_t\right).
\end{aligned}
\end{equation}
Inserting the explicit form of $\Delta_{qp}$, using Lemma \ref{lemma:action_pauli_noise_string} to the leftmost Pauli maps and using Eq.~\eqref{eq:derivative_of_delta}, we can obtain 
\begin{equation}
\Tr(\hat{S}_i\partial_{\gamma_k}\mathcal{T}_{qp}|_{\vec{\gamma}=0}\hat{\rho}_0) = -i\chi'_i\Tr{\hat{S}_i \left(\prod_{r=q+1}^M\mathcal{U}^{\text{\tiny (C)}}_r\right) [\mathcal{H}_k, \mathcal{P}'_{qp}] \,\mathcal{U}^{\text{\tiny (C)}}_q\left(\prod_{r'=p}^{q-1}\mathcal{U}^{\text{\tiny (C)}}_{r'} \right)\left(\prod_{t=1}^{p-1} \mathcal{P}_t\mathcal{U}^{\text{\tiny (C)}}_t\right)\hat{\rho}_0}.
\end{equation}

Finally, the last step is achieved by noting that, since we only have Clifford moments, we can swap all Pauli maps next to the commutator, at the price of conjugation with Clifford moments as in Lemma \ref{lemma:commute_pauli_clifford}.

While it would be possible to keep track exactly of the form of those Pauli maps, it suffices for our purposes to note that they will remain Pauli maps. With this in mind, we are ready to conclude the proof by computing for simplicity the modulus of the derivative (without using the modulus the derivation is more involved). In particular, we have

\begin{equation}
\begin{aligned}
    &\left|\Tr(\hat{S}_i\partial_{\gamma_k}\mathcal{T}_{qp}|_{\vec{\gamma}=0}\hat{\rho}_0)\right| =\left|\chi'_i\Tr{\hat{S}_i \left(\prod_{r=q+1}^M\mathcal{U}^{\text{\tiny (C)}}_r\right) [\mathcal{H}_k, \mathcal{P}'_{qp}] \mathcal{P}''_{qp}\left(\prod_{r'=1}^{q}\mathcal{U}^{\text{\tiny (C)}}_{r'}\right)\hat{\rho}_0}\right|\\
    &\leq \left|\Tr{\hat{S}_i \left(\prod_{r=q+1}^M\mathcal{U}^{\text{\tiny (C)}}_r\right) \mathcal{H}_k \mathcal{P}'_{qp} \mathcal{P}''_{qp}\left(\prod_{r'=1}^{q}\mathcal{U}^{\text{\tiny (C)}}_{r'}\right)\hat{\rho}_0}\right| + \left|\Tr{S_i \left(\prod_{r=q+1}^M\mathcal{U}^{\text{\tiny (C)}}_r\right) \mathcal{P}'_{qp}\mathcal{H}_k \mathcal{P}''_{qp}\left(\prod_{r'=1}^{q}\mathcal{U}^{\text{\tiny (C)}}_{r'}\right)\hat{\rho}_0}\right|,
\end{aligned}
\end{equation}
where we used $|\chi_i|\leq1 \,\forall i$ coming from Lemma \ref{lemma:action_pauli_noise_string} and the triangle inequality to expand the commutator outside the modulus. If we now introduce the notations 
\begin{equation}
    \hat{S}^{(q)}_i \equiv \bigg(\prod_{r=q+1}^M\mathcal{U}^{\text{\tiny (C)}}_r\bigg)^{\dagger}\hat{S}_i \;\; \text{and} \;\; \hat{\rho}_{q} \equiv \bigg(\prod_{r'=1}^{q}\mathcal{U}^{\text{\tiny (C)}}_{r'} \bigg)\hat{\rho}_0,
\end{equation}
using the ciclicity of the trace we have the final bound
\begin{equation}
\begin{aligned}
    \left|\Tr(\hat{S}_i\partial_{\gamma_k}\mathcal{T}_{qp}|_{\vec{\gamma}=0}\hat{\rho}_0)\Big|_{\vec{\gamma}=0}\right| &\leq \left|\Tr(\hat{S}_i^{(q)} \mathcal{H}_k \mathcal{P}'_{qp} \mathcal{P}''_{qp}\hat{\rho}_{q})\right| + \left|\Tr(\hat{S}_i^{(q)} \mathcal{P}'_{qp}\mathcal{H}_k \mathcal{P}''_{qp}\hat{\rho}_{q})\right|\\
    &\leq \left|\Tr(\hat{S}_i^{(q)} \mathcal{H}_k \mathcal{P}'_{qp} \mathcal{P}''_{qp}\hat{\rho}_{q})\right| + \left|\Tr(\hat{S}_i^{(q)} \mathcal{H}_k \mathcal{P}''_{qp}\hat{\rho}_{q})\right|\,,
\end{aligned}
\end{equation}
where again we used Lemma \ref{lemma:action_pauli_noise_string} applied to $\mathcal{P}'_{qp}\hat{S}_i^{(q)}$ in the second term of the first line. By Lemma \ref{lemma:first_derivative_upperbound}, we can now write
\begin{equation}
\left|\Tr(\hat{S}_i\partial_{\gamma_k}\mathcal{T}_{qp}|_{\vec{\gamma}=0}\hat{\rho}_0)\Big|_{\vec{\gamma}=0}\right| \leq 2\|\hat{H}_k\|_2\|[\hat{S}_i^{(q)}, \hat{\rho}_{q}]\|_2 =0,
\end{equation}
where the equality comes from the fact that $\hat{S}_i^{(q)}$ is a stabilizer for $\hat{\rho}_{q}$ by construction, and hence they commute.
\end{proof}

\section{Local convexity of the cost function under weak Pauli noise}
\label{appendix:convexity_C}

Another interesting property to estimate in this setting, are second derivatives of $\Delta\tilde{C}(\vec{\theta})$ in Eq.~\eqref{eq:delta_C_def}, which are related to the convexity of the cost function $\tilde{C}(\vec{\theta})$. In particular, in this section we give some estimates of the Hessian matrix $\tilde{\mathbf{H}}^{\scaleto{(\mathcal{R})}{6pt}}$ computed at the stationary point $\vec{\theta} =-\vec{\epsilon}$, namely

\begin{equation}
\label{eq:hessian_definition}
    \tilde{\mathbf{H}}^{\scaleto{(\mathcal{R})}{6pt}}_{kl} = \partial_{\theta_{k}}\partial_{\theta_l} \Delta\tilde{C}_\theta |_{\vec{\theta}=-\vec{\epsilon}} =- \sum_{i = 1}^n  \Tr(\hat{S}_i\partial_{\theta_k} \partial_{\theta_l}\mathcal{R}\Big|_{\vec{\theta} = -\vec{\epsilon}} \,\hat{\rho}_0).
\end{equation}

Proving positive semi-definiteness of $\tilde{\mathbf{H}}^{\scaleto{(\mathcal{R})}{6pt}}$ in general is a complicated task. Instead of going through this path, we focus on weak Pauli noise, introducing the parameter $\eta \ll 1$ to measure its strength. In particular, we can define the parameter $\eta$ by writing all Pauli maps inside $\mathcal{R}$ as
\begin{equation}
    \mathcal{P}_q = e^{\eta \mathcal{L}_q}, \;\; \text{where} \;\; \mathcal{L}_q\hat{\rho} = \sum_k \mu_{qk} (\hat{P}_{qk}\hat{\rho} \hat{P}_{qk} -\hat{\rho})\,,
\end{equation}
where $\mu_{qk}$ are adimensional numbers and $\mathcal{L}_q$ is the generator of the Pauli map. From this construction, it is clear that, as a function of $\eta$, the Hessian $\tilde{\mathbf{H}}^{\scaleto{(\mathcal{R})}{6pt}}$ is a continuous function, and as result, so are its eigenvalues. By defining $\mathbf{H}$ and $\tilde{\mathbf{H}}$ respectively the Hessian of the cost function $C$ in Eq.~\eqref{eq:cost_function} and of $\tilde{C}$ in Eq.~\eqref{eq:cost_local_pauli} and by recalling that $\tilde{\mathbf{H}}^{\scaleto{(\mathcal{P})}{6pt}}$ is the Hessian of $\tilde{C}^{\scaleto{(\mathcal{P})}{6pt}}$ in Eq.~\eqref{eq:noisy_cost_function}, we can now make the following observation.

\begin{Obs}[Zero-noise limit] In the limit of zero noise $\mathcal{\eta} \to 0$ we have that $\tilde{\mathbf{H}}^{\scaleto{(\mathcal{P})}{6pt}}\to\mathbf{H}$, $\Delta\tilde{C} \to 0$ and hence $\tilde{\mathbf{H}}^{\scaleto{(\mathcal{R})}{6pt}}\to 0$, thus given that $\tilde{\mathbf{H}} = \tilde{\mathbf{H}}^{\scaleto{(\mathcal{P})}{6pt}} + \tilde{\mathbf{H}}^{\scaleto{(\mathcal{R})}{6pt}}$ then $\tilde{\mathbf{H}}\to\mathbf{H}$ ensuring his positive definiteness.
\end{Obs}
This is because at $\eta = 0$, all noise maps are trivial, i.e. $\mathcal{P}_q=\mathbb{1}$, then $\tilde{C}^{\scaleto{(\mathcal{P})}{6pt}} = C$ and $\Delta \tilde{C}=0$. At this point, note that it is safe to assume that all eigenvalues $\lambda_k(\eta = 0)$ of $\mathbf{H}$ are strictly larger than zero. This is a consequence of the fact that, at $\vec{\theta}=-\vec{\epsilon}$, the circuit is Clifford, and produces the unique state stabilized by all $\hat{S}_i$. If we assume that the circuit used is simple enough to produce a different state for every parameter $\vec{\theta}$\footnote{or at least for all $\vec{\theta}$ in a neighborhood of $-\vec{\epsilon}$}, then $C$ is locally convex in $-\vec{\epsilon}$, giving the condition. The argument is concluded by noting that, by continuity, eigenvalues $\lambda_k(\eta)$ of $\tilde{\mathbf{H}}$ are also strictly positive for small enough $\eta$. In other words, if the Pauli noise is weak enough, $\tilde{C}$ is locally convex near its stationary point.

\section{Upper bound on $\Delta \tilde{C}(\vec{\theta})$}
\label{appendix:upper_bound}

In this section we prove Theorem \ref{theo:linear_upper_bound} by computing an analytical upper bound on the error $\Delta \tilde{C}(\vec{\theta})$ that we make when approximating the circuit $\mathcal{N}$ in Eq.~\eqref{eq:circuit_rho_pauli} with the circuit $\mathcal{P}\mathcal{U}_S$, where all noise maps have been commuted to the end giving rise to an effective Pauli map $\mathcal{P}$ as in Eq.~\eqref{eq:circuit_factorized}. 
\begin{proof}
We evaluate the expression
\begin{equation}
    \abs{\Delta \tilde{C}(\vec{\theta})} = \left|\sum_{i=1}^n \Tr(\hat{S}_i \mathcal{R} \hat{\rho}_0)\right|\leq \sum_{i=1}^n \sum_{p=1}^{M-1}\sum_{q>p}^{M} \left|\Tr(\hat{S}_i\mathcal{T}_{qp}\hat{\rho}_0)\right|\, ,
\end{equation}
where as shown in the main text $\mathcal{R}=\mathcal{N}-\mathcal{P}\mathcal{U}_S$ and we inserted the explicit expression of $\mathcal{R}$ in Eq.~\eqref{eq:form_of_the_remainder}.

As next steps, we assume that the coherent errors in each native gate have the same order of magnitude $\epsilon\ll1$. This allows to find the general scaling of $\mathcal{T}_{qp}$ in terms of $\epsilon$, which will give the final upper bound. To do so, lets consider the scalings of $\Delta_{qp}$ and $\mathcal{U}_q^{\text{\tiny (N)}}$ in Eq.~\eqref{eq:form_of_the_remainder}, namely

\begin{equation}
    \label{eq:Uq_expanded}\mathcal{U}_q^{\text{\tiny (N)}} \sim \mathbb{1}-i\epsilon\alpha_q -\frac{\epsilon^2}{2}\alpha_q^2+O(\epsilon^3)\,,
\end{equation}
where we used the shorthand notation $\alpha_q = \sum_{l=1}^{L_q} \mathcal{H}_{ql}$ which represents the sum over generators for all gates in the same moment $q$, and
\begin{equation}
\label{eq:Delta_qp_new}
    \Delta_{qp} = [\mathcal{U}^{\text{\tiny (N)}}_q, \mathcal{P}'_{qp}] \sim -i \epsilon [\alpha_q, \mathcal{P}'_{qp}] -\frac{\epsilon^2}{2} [\alpha^2_q, \mathcal{P}'_{qp}] + O(\epsilon^3),
\end{equation}
where $\mathcal{P}'_{qp}$ is defined in Eq.~\eqref{eq:P_qp}. Substituting Eqs. \eqref{eq:Uq_expanded} and \eqref{eq:Delta_qp_new} into the definition of $\mathcal{T}_{qp}$, and keeping terms up to order $\epsilon^2$, we get

\begin{equation}
\label{eq:second_order_expansion}
\begin{aligned}
    \mathcal{T}_{qp} = -i&\epsilon \left(\prod_{s=p+1}^{M}\mathcal{P}'_{s} \right)\left(\prod_{r=q+1}^{M}\mathcal{U}^{\text{\tiny (C)}}_{r} \right)[\alpha_q, \mathcal{P}'_{qp}]\, \left(\prod_{r'=p}^{q}\mathcal{U}^{\text{\tiny (C)}}_{r'} \right)\left(\prod_{t=1}^{p-1} \mathcal{P}_t\mathcal{U}^{\text{\tiny (C)}}_t\right)\\
    -&\epsilon^2 \sum_{q<u} \left(\prod_{s=p+1}^{M}\mathcal{P}'_{s} \right)\left(\prod_{b=u+1}^{M}\mathcal{U}^{\text{\tiny (C)}}_{b} \right)\alpha_u\left(\prod_{b'=q+1}^{u}\mathcal{U}^{\text{\tiny (C)}}_{b'} \right)[\alpha_q, \mathcal{P}'_{qp}]\, \left(\prod_{r'=p}^{q}\mathcal{U}^{\text{\tiny (C)}}_{r'} \right)\left(\prod_{t=1}^{p-1} \mathcal{P}_t\mathcal{U}^{\text{\tiny (C)}}_t\right)\\
    -&\epsilon^2 \sum_{p \leq u<q} \left(\prod_{s=p+1}^{M}\mathcal{P}'_{s} \right)\left(\prod_{r=q+1}^{M}\mathcal{U}^{\text{\tiny (C)}}_{r} \right)[\alpha_q, \mathcal{P}'_{qp}]\, \left(\prod_{b=u+1}^{q}\mathcal{U}^{\text{\tiny (C)}}_{b} \right)\alpha_u\left(\prod_{b'=p}^{u}\mathcal{U}^{\text{\tiny (C)}}_{b'} \right)\left(\prod_{t=1}^{p-1} \mathcal{P}_t\mathcal{U}^{\text{\tiny (C)}}_t\right)\\
    -&\epsilon^2 \sum_{u<p} \left(\prod_{s=p+1}^{M}\mathcal{P}'_{s} \right)\left(\prod_{r=q+1}^{M}\mathcal{U}^{\text{\tiny (C)}}_{r} \right)[\alpha_q, \mathcal{P}'_{qp}]\, \left(\prod_{r'=p}^{q}\mathcal{U}^{\text{\tiny (C)}}_{r'} \right)\left(\prod_{b=u+1}^{p-1} \mathcal{P}_b\mathcal{U}^{\text{\tiny (C)}}_b\right)\mathcal{P}_u \alpha_u\left(\prod_{b'=1}^{u} \mathcal{P}_{b'}\mathcal{U}_{b'}^{\text{\tiny (C)}}\right)\\
    -&\frac{\epsilon^2}{2} \left(\prod_{s=p+1}^{M}\mathcal{P}'_{s} \right)\left(\prod_{r=q+1}^{M}\mathcal{U}^{\text{\tiny (C)}}_{r} \right)[\alpha_q^2, \mathcal{P}'_{qp}]\, \left(\prod_{r'=p}^{q}\mathcal{U}^{\text{\tiny (C)}}_{r'} \right)\left(\prod_{t=1}^{p-1} \mathcal{P}_t\mathcal{U}^{\text{\tiny (C)}}_t\right).
\end{aligned}
\end{equation}

This decomposition has a lot of structure, which can be exploited to greatly simplify our calculations. In particular, we can make a few observations.

\begin{Obs}[No linear terms in $\epsilon$] The term proportional to $\epsilon$ vanish in the expectation value with $\hat{S}_i$ and $\hat{\rho}_0$, i.e.
\begin{equation}
    \Tr(\hat{S}_i \mathcal{T}_{qp}\hat{\rho}_0) \le O(\epsilon^2)\, .
\end{equation}
\end{Obs}
Indeed one can notice that each of the terms arising from $\alpha_q = \sum_{l=1}^{L_q} \mathcal{H}_{ql}$ in the first line of Eq.~\eqref{eq:second_order_expansion} is proportional to the expression $\partial_{\gamma_k} \Tr(\hat{S}_i \mathcal{T}_{qp}\hat{\rho}_0)|_{\vec{\gamma} = 0}$ in Eq.~\eqref{eq:derivative_Tqp}, which was already proven to vanish in Section \ref{appendix:stationarity}.

\begin{Obs}[All non-vanishing terms have the same structure] If we expand the commutators $[\alpha_q, \mathcal{P}'_{qp}] = \alpha_q \mathcal{P'}_{qp} -\mathcal{P}'_{qp}\alpha_q$ and $[\alpha_q^2, \mathcal{P}'_{qp}] = \alpha_q^2 \mathcal{P'}_{qp} -\mathcal{P}'_{qp}\alpha_q^2$, the terms proportional to $\epsilon^2$ in Eq.~\eqref{eq:second_order_expansion} have the same underlying structure, namely

\begin{equation}
\label{eq:second_order_common_structure}
    \epsilon^2\; \mathcal{P}_3\mathcal{U}_3\alpha_2\mathcal{P}_2\alpha_1\mathcal{P}_1\mathcal{U}_0\, ,
\end{equation}
where the explicit expression of the maps in Eq.~\eqref{eq:second_order_common_structure} depends on the specific term considered in Eq.~\eqref{eq:second_order_expansion}.
\end{Obs}

Consider as an example the second term in Eq.~\eqref{eq:second_order_expansion} for some specific $u>q$, then we set
\begin{equation}
\mathcal{P}_3 = \prod_{s=p+1}^{M}\mathcal{P}'_{s}\, .
\end{equation}
Concerning $\mathcal{U}_3$ and $\alpha_2$, we can iteratively apply Lemma \ref{lemma:commute_h_unitary} to move all unitary operations indexed by $b$ and $b'$ next to each other collecting them in a single unitary. By naming $\alpha'_u$ the result of this conjugation, we are left with
\begin{equation}
\mathcal{U}_3 =\left(\prod_{b=u+1}^{M}\mathcal{U}^{\text{\tiny (C)}}_{b} \right)\left(\prod_{b'=q+1}^{u}\mathcal{U}^{\text{\tiny (C)}}_{b'} \right)=\prod_{r=q+1}^{M}\mathcal{U}^{\text{\tiny (C)}}_{r}\;\; \text{and} \;\; \alpha_2 = \alpha_u' = \left(\prod_{b'=q+1}^{u}\mathcal{U}^{\text{\tiny (C)}}_{b'} \right)^{\dagger} \alpha_u \left(\prod_{b'=q+1}^{u}\mathcal{U}^{\text{\tiny (C)}}_{b'} \right)\, .
\end{equation}
Continuing, $\mathcal{P}_2$ is equal either to the identity or to $\mathcal{P}'_{qp}$ depending on which branch of the commutator we are considering, $\alpha_1 = \alpha_q$, and finally, $\mathcal{P}_1$ and $\mathcal{U}_0$ can be obtained by Lemma \ref{lemma:commute_pauli_clifford} moving all remaining Pauli maps next to the right of $\alpha_q$, getting
\begin{equation}
 \mathcal{P}_1 = \prod_{t=1}^{p-1} \mathcal{P}'_{qt} \;\; \text{and}\;\; \mathcal{U}_0 = \prod_{t=1}^{q}\mathcal{U}^{\text{\tiny (C)}}_t \, ,  
\end{equation}
where $\mathcal{P}'_{qt}$ has the structure in Eq.~\eqref{eq:P_qp}.

With similar arguments, we can rewrite all terms arising from Eq.~\eqref{eq:second_order_expansion} in this form. It then suffices to upper bound the equation

\begin{equation}
\left|\Tr(\hat{S}_i\mathcal{P}_3\mathcal{U}_3\alpha_2\mathcal{P}_2\alpha_1\mathcal{P}_1\mathcal{U}_0 \hat{\rho}_0)\right| \leq \left|\Tr(\hat{S}_i\mathcal{U}_3\alpha_2\mathcal{P}_2\alpha_1\mathcal{P}_1\mathcal{U}_0 \hat{\rho}_0)\right| = \left|\Tr(\hat{S}'_i\alpha_2\mathcal{P}_2\alpha_1\mathcal{P}_1\hat{\rho}')\right|\,,
\end{equation}
where we used Lemma \ref{lemma:action_pauli_noise_string} on $\hat{S}_i\mathcal{P}_3$ and we defined $\hat{S}'_i \equiv \mathcal{U}^{\dagger}_3\hat{S}_i$ and $\hat{\rho}' \equiv \mathcal{U}_0 \hat{\rho}_0$.
Assuming that $\hat{S}'_i$ and all generators in $\alpha_2, \;\alpha_1$ are $m$-local, we can apply Lemma \ref{lemma:second_derivative_upperbound}, and get
\begin{equation}
\left|\Tr(\hat{S}'_i\alpha_2\mathcal{P}_2\alpha_1\mathcal{P}_1\hat{\rho}')\right| \leq \Upsilon N_c \le O(4^{2m-1})\,.
\end{equation}

Counting all the terms of this form appearing in Eq.~\eqref{eq:second_order_expansion}, we get a total of $M$ contributions of this kind for each $\Tr(\hat{S}_i\mathcal{T}_{qp}\hat{\rho}_0)$, giving the final bound
\begin{equation}
    |\Delta C(\vec{\theta})| \leq \epsilon^2n M^2(M-1)\Upsilon N_c\,.
\end{equation}

We remark that, the bound obtained scales with $\Upsilon N_c \le O(4^{2m-1})$, which is constant only if $m$ does not scale with $n$. This is indeed satisfied for locally-connected graph state circuits, where the stabilizers $\hat{S}_i$ can be chosen to be local, and the depth $M$ is also constant, meaning that $m \leq 2^M$.
\end{proof}

\section{Twirling quantum maps}
\label{sec:twirling}
Let $\mathbb{U}_d$ denote the continuous set of $d\times d$ unitary matrices $\hat{U}\in  \mathbb{U}_d$ where $d = 2^n$ and $\mathcal{E}\hat{\rho} = \sum_{j=0}^{4^{n}-1}\hat{E}_{j}\hat{\rho}\hat{E}_{j}$ denote a general quantum map, then we can define
\begin{equation}
\mathcal{E}_{\text{\tiny T}}^{\text{\tiny (C)}}\hat{\rho}=\int_{\mathbb{U}_d} \D \hat{U}\,\ \hat{U}^{\dagger}\mathcal{E}(\hat{U}\hat{\rho}\hat{U}^{\dagger})\hat{U}\,\ ,
\end{equation} 
the continuous twirling of the map $\mathcal{E}$ \cite{olivia_dimatteo}. One can also define the twirling on a discrete distribution of unitaries $\mathbb{D} = \{\hat{U}_0\dots\hat{U}_r\}$ such that
\begin{equation}
\mathcal{E}_{\text{\tiny T}}^{\text{\tiny (D)}}\hat{\rho}=\frac{1}{r}\sum_{j=0}^{r}\hat{U}_j^{\dagger}\mathcal{E}(\hat{U}_j\hat{\rho}\hat{U}_j^{\dagger})\hat{U}_j\, .
\end{equation}
\subsection{Pauli twirling}
Pauli twirling is an example of a discrete twirling where $\mathbb{D}=\mathcal{\bf{P}}_n$ is the Pauli group, thus
\begin{equation}
\mathcal{E}_{\text{\tiny T}}^{\text{\tiny (Pauli)}}\hat{\rho}=\frac{1}{4^n-1}\sum_{j=0}^{4^n-1}\hat{P}_j\mathcal{E}(\hat{P}_j\hat{\rho}\hat{P}_j)\hat{P}_j\, ,
\end{equation}
with $\hat{P}_j \in \mathcal{\bf{P}}_n$. By expressing $\hat{E}_j=\sum_{k=0}^{4^n-1}c_{j,k}\hat{P}_k$  with $c_{j,k}\in\mathbb{C}$, then it is straightforward to show that
\begin{equation}
\mathcal{E}_{\text{\tiny T}}^{\text{\tiny (Pauli)}}\hat{\rho}= \sum_{k=0}^{4^n-1}p_k\hat{P}_k\hat{\rho}\hat{P}_k\,\ ,
\end{equation}
where $p_k = \sum_{j=0}^{4^n-1}\abs{c_{j,k}}^2$. Thus, Pauli twirling transforms a general quantum map into a Pauli map.

\subsection{Clifford twirling}
Let $\mathcal{\bf{C}}_n$ denote the Clifford group, namely the group of unitaries $\hat{C}_j\in \mathcal{\bf{C}}_n$ such that $\hat{C}_j\hat{P}_i\hat{C}_j^{\dagger}=\hat{P}_k$ where $j=1,\dots,r$ and $\hat{P}_i,\hat{P}_k\in\mathcal{\bf{P}}_n$. Then, one can define the Clifford twirling as
\begin{equation}
\mathcal{E}_{\text{\tiny T}}^{\text{\tiny (Clifford)}}\hat{\rho}=\frac{1}{r}\sum_{j=0}^{r}\hat{C}_j\mathcal{E}(\hat{C}_j\hat{\rho}\hat{C}_j)\hat{C}_j\, .
\end{equation}
Moreover, it is proved that $\mathcal{E}_{\text{\tiny T}}^{\text{\tiny (Clifford)}}\hat{\rho}=\mathcal{E}_{\text{\tiny T}}^{\text{\tiny (C)}}\hat{\rho}$ \cite{Gross_2007} and $\mathcal{E}_{\text{\tiny T}}^{\text{\tiny (C)}}\hat{\rho}=(1-p) \hat{\rho}+\frac{p}{2^n} \mathbb{1}$ with $p=\frac{4^n-\Tr(\mathcal{E}_{\text{\tiny T}}^{\text{\tiny (C)}}(\hat{\rho}))}{4^n-1}$ \cite{Emerson_2005,D_r_2005}, then
\begin{equation}
\mathcal{E}_{\text{\tiny T}}^{\text{\tiny (Clifford)}}\hat{\rho}= (1-p) \hat{\rho}+\frac{p}{2^n} \mathbb{1}\,\ .
\end{equation}
Thus, Clifford twirling transforms a general quantum map into a global depolarizing map.

\section{Transpilation of Hadamard and CZ gates}
\label{sec:transpilation}

\begin{figure}[H]
    \centering
    \includegraphics[width=\linewidth]{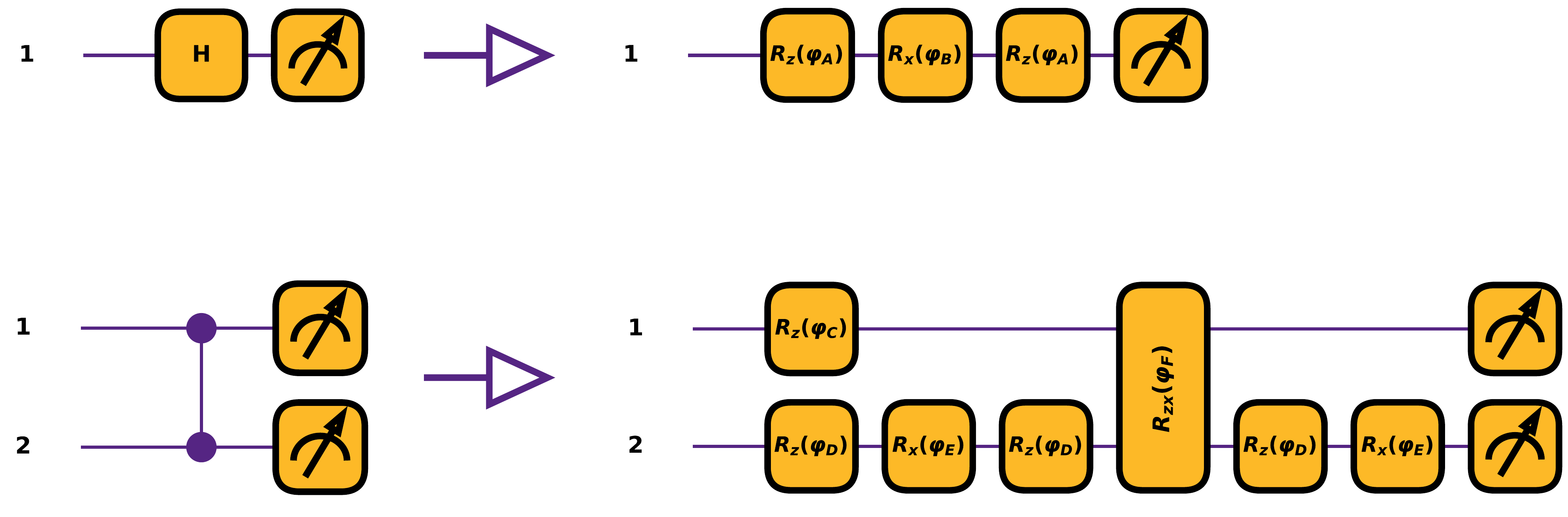}
    \caption{Transpilation of Hadamard and CZ gates into Rz, Rx and Rzx native gates. The angles are defined in the following way: $\varphi_A = \frac{\pi}{2}+\theta_1 + \epsilon_1$,  $\varphi_B = \frac{\pi}{2}+\theta_2 + \epsilon_2$, $\varphi_C = -\frac{3}{2}\pi+\theta_1 + \epsilon_1$,  $\varphi_D = -\frac{\pi}{2}+\theta_3 + \epsilon_3$, $\varphi_E = \frac{\pi}{2}+\theta_4 + \epsilon_4$, $\varphi_F = \frac{\pi}{2}+\theta_5 + \epsilon_5$.}, 
    \label{fig:transpile}
\end{figure}

Here in Fig. \ref{fig:transpile} we show the transpilation of Hadamard and CZ gates into native gates used in section \ref{sec:numerics}.

\bibliography{biblio}

\begin{thebibliography}{88}%
\makeatletter
\providecommand \@ifxundefined [1]{%
 \@ifx{#1\undefined}
}%
\providecommand \@ifnum [1]{%
 \ifnum #1\expandafter \@firstoftwo
 \else \expandafter \@secondoftwo
 \fi
}%
\providecommand \@ifx [1]{%
 \ifx #1\expandafter \@firstoftwo
 \else \expandafter \@secondoftwo
 \fi
}%
\providecommand \natexlab [1]{#1}%
\providecommand \enquote  [1]{``#1''}%
\providecommand \bibnamefont  [1]{#1}%
\providecommand \bibfnamefont [1]{#1}%
\providecommand \citenamefont [1]{#1}%
\providecommand \href@noop [0]{\@secondoftwo}%
\providecommand \href [0]{\begingroup \@sanitize@url \@href}%
\providecommand \@href[1]{\@@startlink{#1}\@@href}%
\providecommand \@@href[1]{\endgroup#1\@@endlink}%
\providecommand \@sanitize@url [0]{\catcode `\\12\catcode `\$12\catcode `\&12\catcode `\#12\catcode `\^12\catcode `\_12\catcode `\%12\relax}%
\providecommand \@@startlink[1]{}%
\providecommand \@@endlink[0]{}%
\providecommand \url  [0]{\begingroup\@sanitize@url \@url }%
\providecommand \@url [1]{\endgroup\@href {#1}{\urlprefix }}%
\providecommand \urlprefix  [0]{URL }%
\providecommand \Eprint [0]{\href }%
\providecommand \doibase [0]{https://doi.org/}%
\providecommand \selectlanguage [0]{\@gobble}%
\providecommand \bibinfo  [0]{\@secondoftwo}%
\providecommand \bibfield  [0]{\@secondoftwo}%
\providecommand \translation [1]{[#1]}%
\providecommand \BibitemOpen [0]{}%
\providecommand \bibitemStop [0]{}%
\providecommand \bibitemNoStop [0]{.\EOS\space}%
\providecommand \EOS [0]{\spacefactor3000\relax}%
\providecommand \BibitemShut  [1]{\csname bibitem#1\endcsname}%
\let\auto@bib@innerbib\@empty
\bibitem [{\citenamefont {Englbrecht}\ and\ \citenamefont {Kraus}(2020)}]{Englbrecht_symmetries}%
  \BibitemOpen
  \bibfield  {author} {\bibinfo {author} {\bibfnamefont {M.}~\bibnamefont {Englbrecht}}\ and\ \bibinfo {author} {\bibfnamefont {B.}~\bibnamefont {Kraus}},\ }\bibfield  {title} {\bibinfo {title} {Symmetries and entanglement of stabilizer states},\ }\href {https://doi.org/10.1103/PhysRevA.101.062302} {\bibfield  {journal} {\bibinfo  {journal} {Phys. Rev. A}\ }\textbf {\bibinfo {volume} {101}},\ \bibinfo {pages} {062302} (\bibinfo {year} {2020})}\BibitemShut {NoStop}%
\bibitem [{\citenamefont {Veitch}\ \emph {et~al.}(2014)\citenamefont {Veitch}, \citenamefont {Mousavian}, \citenamefont {Gottesman},\ and\ \citenamefont {Emerson}}]{veitch2014resource}%
  \BibitemOpen
  \bibfield  {author} {\bibinfo {author} {\bibfnamefont {V.}~\bibnamefont {Veitch}}, \bibinfo {author} {\bibfnamefont {S.~H.}\ \bibnamefont {Mousavian}}, \bibinfo {author} {\bibfnamefont {D.}~\bibnamefont {Gottesman}},\ and\ \bibinfo {author} {\bibfnamefont {J.}~\bibnamefont {Emerson}},\ }\bibfield  {title} {\bibinfo {title} {The resource theory of stabilizer quantum computation},\ }\href {https://doi.org/10.1088/1367-2630/16/1/013009} {\bibfield  {journal} {\bibinfo  {journal} {New Journal of Physics}\ }\textbf {\bibinfo {volume} {16}},\ \bibinfo {pages} {013009} (\bibinfo {year} {2014})}\BibitemShut {NoStop}%
\bibitem [{\citenamefont {Smith}\ and\ \citenamefont {Leung}(2006)}]{graeme_typical}%
  \BibitemOpen
  \bibfield  {author} {\bibinfo {author} {\bibfnamefont {G.}~\bibnamefont {Smith}}\ and\ \bibinfo {author} {\bibfnamefont {D.}~\bibnamefont {Leung}},\ }\bibfield  {title} {\bibinfo {title} {Typical entanglement of stabilizer states},\ }\href {https://doi.org/10.1103/PhysRevA.74.062314} {\bibfield  {journal} {\bibinfo  {journal} {Phys. Rev. A}\ }\textbf {\bibinfo {volume} {74}},\ \bibinfo {pages} {062314} (\bibinfo {year} {2006})}\BibitemShut {NoStop}%
\bibitem [{\citenamefont {Shor}(1995)}]{shor_QEC}%
  \BibitemOpen
  \bibfield  {author} {\bibinfo {author} {\bibfnamefont {P.~W.}\ \bibnamefont {Shor}},\ }\bibfield  {title} {\bibinfo {title} {Scheme for reducing decoherence in quantum computer memory},\ }\href {https://doi.org/10.1103/PhysRevA.52.R2493} {\bibfield  {journal} {\bibinfo  {journal} {Phys. Rev. A}\ }\textbf {\bibinfo {volume} {52}},\ \bibinfo {pages} {R2493} (\bibinfo {year} {1995})}\BibitemShut {NoStop}%
\bibitem [{\citenamefont {Knill}\ \emph {et~al.}(2000)\citenamefont {Knill}, \citenamefont {Laflamme},\ and\ \citenamefont {Viola}}]{knill2000theory}%
  \BibitemOpen
  \bibfield  {author} {\bibinfo {author} {\bibfnamefont {E.}~\bibnamefont {Knill}}, \bibinfo {author} {\bibfnamefont {R.}~\bibnamefont {Laflamme}},\ and\ \bibinfo {author} {\bibfnamefont {L.}~\bibnamefont {Viola}},\ }\bibfield  {title} {\bibinfo {title} {Theory of quantum error correction for general noise},\ }\href@noop {} {\bibfield  {journal} {\bibinfo  {journal} {Physical Review Letters}\ }\textbf {\bibinfo {volume} {84}},\ \bibinfo {pages} {2525} (\bibinfo {year} {2000})}\BibitemShut {NoStop}%
\bibitem [{\citenamefont {Gottesman}(2009)}]{gottesman2009introduction}%
  \BibitemOpen
  \bibfield  {author} {\bibinfo {author} {\bibfnamefont {D.}~\bibnamefont {Gottesman}},\ }\bibfield  {title} {\bibinfo {title} {An introduction to quantum error correction and fault-tolerant quantum computation},\ }\href@noop {} {\bibfield  {journal} {\bibinfo  {journal} {arXiv preprint arXiv:0904.2557}\ } (\bibinfo {year} {2009})}\BibitemShut {NoStop}%
\bibitem [{\citenamefont {Terhal}(2015)}]{terhal_QEC}%
  \BibitemOpen
  \bibfield  {author} {\bibinfo {author} {\bibfnamefont {B.~M.}\ \bibnamefont {Terhal}},\ }\bibfield  {title} {\bibinfo {title} {Quantum error correction for quantum memories},\ }\href {https://doi.org/10.1103/RevModPhys.87.307} {\bibfield  {journal} {\bibinfo  {journal} {Rev. Mod. Phys.}\ }\textbf {\bibinfo {volume} {87}},\ \bibinfo {pages} {307} (\bibinfo {year} {2015})}\BibitemShut {NoStop}%
\bibitem [{\citenamefont {Cory}\ \emph {et~al.}(1998)\citenamefont {Cory}, \citenamefont {Price}, \citenamefont {Maas}, \citenamefont {Knill}, \citenamefont {Laflamme}, \citenamefont {Zurek}, \citenamefont {Havel},\ and\ \citenamefont {Somaroo}}]{cory_QEC}%
  \BibitemOpen
  \bibfield  {author} {\bibinfo {author} {\bibfnamefont {D.~G.}\ \bibnamefont {Cory}}, \bibinfo {author} {\bibfnamefont {M.~D.}\ \bibnamefont {Price}}, \bibinfo {author} {\bibfnamefont {W.}~\bibnamefont {Maas}}, \bibinfo {author} {\bibfnamefont {E.}~\bibnamefont {Knill}}, \bibinfo {author} {\bibfnamefont {R.}~\bibnamefont {Laflamme}}, \bibinfo {author} {\bibfnamefont {W.~H.}\ \bibnamefont {Zurek}}, \bibinfo {author} {\bibfnamefont {T.~F.}\ \bibnamefont {Havel}},\ and\ \bibinfo {author} {\bibfnamefont {S.~S.}\ \bibnamefont {Somaroo}},\ }\bibfield  {title} {\bibinfo {title} {Experimental quantum error correction},\ }\href {https://doi.org/10.1103/PhysRevLett.81.2152} {\bibfield  {journal} {\bibinfo  {journal} {Phys. Rev. Lett.}\ }\textbf {\bibinfo {volume} {81}},\ \bibinfo {pages} {2152} (\bibinfo {year} {1998})}\BibitemShut {NoStop}%
\bibitem [{\citenamefont {Acharya~et al.}(2024)}]{2024_google_QEC}%
  \BibitemOpen
  \bibfield  {author} {\bibinfo {author} {\bibfnamefont {R.}~\bibnamefont {Acharya~et al.}},\ }\bibfield  {title} {\bibinfo {title} {Quantum error correction below the surface code threshold},\ }\href {https://doi.org/10.1038/s41586-024-08449-y} {\bibfield  {journal} {\bibinfo  {journal} {Nature}\ }\textbf {\bibinfo {volume} {638}},\ \bibinfo {pages} {920–926} (\bibinfo {year} {2024})}\BibitemShut {NoStop}%
\bibitem [{\citenamefont {Raussendorf}\ and\ \citenamefont {Briegel}(2001)}]{one_way_raussendorf}%
  \BibitemOpen
  \bibfield  {author} {\bibinfo {author} {\bibfnamefont {R.}~\bibnamefont {Raussendorf}}\ and\ \bibinfo {author} {\bibfnamefont {H.~J.}\ \bibnamefont {Briegel}},\ }\bibfield  {title} {\bibinfo {title} {A one-way quantum computer},\ }\href {https://doi.org/10.1103/PhysRevLett.86.5188} {\bibfield  {journal} {\bibinfo  {journal} {Phys. Rev. Lett.}\ }\textbf {\bibinfo {volume} {86}},\ \bibinfo {pages} {5188} (\bibinfo {year} {2001})}\BibitemShut {NoStop}%
\bibitem [{\citenamefont {Briegel}\ \emph {et~al.}(2009)\citenamefont {Briegel}, \citenamefont {Browne}, \citenamefont {D{\"u}r}, \citenamefont {Raussendorf},\ and\ \citenamefont {Van~den Nest}}]{briegel2009measurement}%
  \BibitemOpen
  \bibfield  {author} {\bibinfo {author} {\bibfnamefont {H.~J.}\ \bibnamefont {Briegel}}, \bibinfo {author} {\bibfnamefont {D.~E.}\ \bibnamefont {Browne}}, \bibinfo {author} {\bibfnamefont {W.}~\bibnamefont {D{\"u}r}}, \bibinfo {author} {\bibfnamefont {R.}~\bibnamefont {Raussendorf}},\ and\ \bibinfo {author} {\bibfnamefont {M.}~\bibnamefont {Van~den Nest}},\ }\bibfield  {title} {\bibinfo {title} {Measurement-based quantum computation},\ }\href@noop {} {\bibfield  {journal} {\bibinfo  {journal} {Nature Physics}\ }\textbf {\bibinfo {volume} {5}},\ \bibinfo {pages} {19} (\bibinfo {year} {2009})}\BibitemShut {NoStop}%
\bibitem [{\citenamefont {Raussendorf}\ \emph {et~al.}(2003)\citenamefont {Raussendorf}, \citenamefont {Browne},\ and\ \citenamefont {Briegel}}]{raussendorf2003measurement}%
  \BibitemOpen
  \bibfield  {author} {\bibinfo {author} {\bibfnamefont {R.}~\bibnamefont {Raussendorf}}, \bibinfo {author} {\bibfnamefont {D.~E.}\ \bibnamefont {Browne}},\ and\ \bibinfo {author} {\bibfnamefont {H.~J.}\ \bibnamefont {Briegel}},\ }\bibfield  {title} {\bibinfo {title} {Measurement-based quantum computation on cluster states},\ }\href@noop {} {\bibfield  {journal} {\bibinfo  {journal} {Physical review A}\ }\textbf {\bibinfo {volume} {68}},\ \bibinfo {pages} {022312} (\bibinfo {year} {2003})}\BibitemShut {NoStop}%
\bibitem [{\citenamefont {Wei}(2021)}]{Wei_2021}%
  \BibitemOpen
  \bibfield  {author} {\bibinfo {author} {\bibfnamefont {T.-C.}\ \bibnamefont {Wei}},\ }\href {https://doi.org/10.1093/acrefore/9780190871994.013.31} {\bibinfo {title} {Measurement-based quantum computation}} (\bibinfo {year} {2021})\BibitemShut {NoStop}%
\bibitem [{\citenamefont {Walther}\ \emph {et~al.}(2005)\citenamefont {Walther}, \citenamefont {Resch}, \citenamefont {Rudolph}, \citenamefont {Schenck}, \citenamefont {Weinfurter}, \citenamefont {Vedral}, \citenamefont {Aspelmeyer},\ and\ \citenamefont {Zeilinger}}]{walther2005experimental}%
  \BibitemOpen
  \bibfield  {author} {\bibinfo {author} {\bibfnamefont {P.}~\bibnamefont {Walther}}, \bibinfo {author} {\bibfnamefont {K.~J.}\ \bibnamefont {Resch}}, \bibinfo {author} {\bibfnamefont {T.}~\bibnamefont {Rudolph}}, \bibinfo {author} {\bibfnamefont {E.}~\bibnamefont {Schenck}}, \bibinfo {author} {\bibfnamefont {H.}~\bibnamefont {Weinfurter}}, \bibinfo {author} {\bibfnamefont {V.}~\bibnamefont {Vedral}}, \bibinfo {author} {\bibfnamefont {M.}~\bibnamefont {Aspelmeyer}},\ and\ \bibinfo {author} {\bibfnamefont {A.}~\bibnamefont {Zeilinger}},\ }\bibfield  {title} {\bibinfo {title} {Experimental one-way quantum computing},\ }\href {https://doi.org/https://doi.org/10.1038/nature03347} {\bibfield  {journal} {\bibinfo  {journal} {Nature}\ }\textbf {\bibinfo {volume} {434}},\ \bibinfo {pages} {169} (\bibinfo {year} {2005})}\BibitemShut {NoStop}%
\bibitem [{\citenamefont {Ferguson}\ \emph {et~al.}(2021)\citenamefont {Ferguson}, \citenamefont {Dellantonio}, \citenamefont {Balushi}, \citenamefont {Jansen}, \citenamefont {D\"ur},\ and\ \citenamefont {Muschik}}]{ferguson_MBQC}%
  \BibitemOpen
  \bibfield  {author} {\bibinfo {author} {\bibfnamefont {R.~R.}\ \bibnamefont {Ferguson}}, \bibinfo {author} {\bibfnamefont {L.}~\bibnamefont {Dellantonio}}, \bibinfo {author} {\bibfnamefont {A.~A.}\ \bibnamefont {Balushi}}, \bibinfo {author} {\bibfnamefont {K.}~\bibnamefont {Jansen}}, \bibinfo {author} {\bibfnamefont {W.}~\bibnamefont {D\"ur}},\ and\ \bibinfo {author} {\bibfnamefont {C.~A.}\ \bibnamefont {Muschik}},\ }\bibfield  {title} {\bibinfo {title} {Measurement-based variational quantum eigensolver},\ }\href {https://doi.org/10.1103/PhysRevLett.126.220501} {\bibfield  {journal} {\bibinfo  {journal} {Phys. Rev. Lett.}\ }\textbf {\bibinfo {volume} {126}},\ \bibinfo {pages} {220501} (\bibinfo {year} {2021})}\BibitemShut {NoStop}%
\bibitem [{\citenamefont {Shor}\ and\ \citenamefont {Preskill}(2000)}]{shor_QKD}%
  \BibitemOpen
  \bibfield  {author} {\bibinfo {author} {\bibfnamefont {P.~W.}\ \bibnamefont {Shor}}\ and\ \bibinfo {author} {\bibfnamefont {J.}~\bibnamefont {Preskill}},\ }\bibfield  {title} {\bibinfo {title} {Simple proof of security of the bb84 quantum key distribution protocol},\ }\href {https://doi.org/10.1103/PhysRevLett.85.441} {\bibfield  {journal} {\bibinfo  {journal} {Phys. Rev. Lett.}\ }\textbf {\bibinfo {volume} {85}},\ \bibinfo {pages} {441} (\bibinfo {year} {2000})}\BibitemShut {NoStop}%
\bibitem [{\citenamefont {Gottesman}\ and\ \citenamefont {Preskill}(2001)}]{gottesman_secure}%
  \BibitemOpen
  \bibfield  {author} {\bibinfo {author} {\bibfnamefont {D.}~\bibnamefont {Gottesman}}\ and\ \bibinfo {author} {\bibfnamefont {J.}~\bibnamefont {Preskill}},\ }\bibfield  {title} {\bibinfo {title} {Secure quantum key distribution using squeezed states},\ }\href {https://doi.org/10.1103/PhysRevA.63.022309} {\bibfield  {journal} {\bibinfo  {journal} {Phys. Rev. A}\ }\textbf {\bibinfo {volume} {63}},\ \bibinfo {pages} {022309} (\bibinfo {year} {2001})}\BibitemShut {NoStop}%
\bibitem [{\citenamefont {Renner}(2008)}]{renner2008security}%
  \BibitemOpen
  \bibfield  {author} {\bibinfo {author} {\bibfnamefont {R.}~\bibnamefont {Renner}},\ }\bibfield  {title} {\bibinfo {title} {Security of quantum key distribution},\ }\href {https://doi.org/https://doi.org/10.1142/S0219749908003256} {\bibfield  {journal} {\bibinfo  {journal} {International Journal of Quantum Information}\ }\textbf {\bibinfo {volume} {6}},\ \bibinfo {pages} {1} (\bibinfo {year} {2008})}\BibitemShut {NoStop}%
\bibitem [{\citenamefont {Le~Roy-Deloison}\ \emph {et~al.}(2025)\citenamefont {Le~Roy-Deloison}, \citenamefont {Lobo}, \citenamefont {Pauwels},\ and\ \citenamefont {Pironio}}]{leroy_QKD}%
  \BibitemOpen
  \bibfield  {author} {\bibinfo {author} {\bibfnamefont {T.}~\bibnamefont {Le~Roy-Deloison}}, \bibinfo {author} {\bibfnamefont {E.~P.}\ \bibnamefont {Lobo}}, \bibinfo {author} {\bibfnamefont {J.}~\bibnamefont {Pauwels}},\ and\ \bibinfo {author} {\bibfnamefont {S.}~\bibnamefont {Pironio}},\ }\bibfield  {title} {\bibinfo {title} {Device-independent quantum key distribution based on routed bell tests},\ }\href {https://doi.org/10.1103/PRXQuantum.6.020311} {\bibfield  {journal} {\bibinfo  {journal} {PRX Quantum}\ }\textbf {\bibinfo {volume} {6}},\ \bibinfo {pages} {020311} (\bibinfo {year} {2025})}\BibitemShut {NoStop}%
\bibitem [{\citenamefont {Shi}\ \emph {et~al.}(2025)\citenamefont {Shi}, \citenamefont {Patil},\ and\ \citenamefont {Guha}}]{Yu_QKD}%
  \BibitemOpen
  \bibfield  {author} {\bibinfo {author} {\bibfnamefont {Y.}~\bibnamefont {Shi}}, \bibinfo {author} {\bibfnamefont {A.}~\bibnamefont {Patil}},\ and\ \bibinfo {author} {\bibfnamefont {S.}~\bibnamefont {Guha}},\ }\bibfield  {title} {\bibinfo {title} {Stabilizer entanglement distillation and efficient fault-tolerant encoders},\ }\href {https://doi.org/10.1103/PRXQuantum.6.010339} {\bibfield  {journal} {\bibinfo  {journal} {PRX Quantum}\ }\textbf {\bibinfo {volume} {6}},\ \bibinfo {pages} {010339} (\bibinfo {year} {2025})}\BibitemShut {NoStop}%
\bibitem [{\citenamefont {Cirac}\ \emph {et~al.}(1999)\citenamefont {Cirac}, \citenamefont {Ekert}, \citenamefont {Huelga},\ and\ \citenamefont {Macchiavello}}]{eisert_DQC}%
  \BibitemOpen
  \bibfield  {author} {\bibinfo {author} {\bibfnamefont {J.~I.}\ \bibnamefont {Cirac}}, \bibinfo {author} {\bibfnamefont {A.~K.}\ \bibnamefont {Ekert}}, \bibinfo {author} {\bibfnamefont {S.~F.}\ \bibnamefont {Huelga}},\ and\ \bibinfo {author} {\bibfnamefont {C.}~\bibnamefont {Macchiavello}},\ }\bibfield  {title} {\bibinfo {title} {Distributed quantum computation over noisy channels},\ }\href {https://doi.org/10.1103/PhysRevA.59.4249} {\bibfield  {journal} {\bibinfo  {journal} {Phys. Rev. A}\ }\textbf {\bibinfo {volume} {59}},\ \bibinfo {pages} {4249} (\bibinfo {year} {1999})}\BibitemShut {NoStop}%
\bibitem [{\citenamefont {Caleffi}\ \emph {et~al.}(2024)\citenamefont {Caleffi}, \citenamefont {Amoretti}, \citenamefont {Ferrari}, \citenamefont {Illiano}, \citenamefont {Manzalini},\ and\ \citenamefont {Cacciapuoti}}]{Caleffi_2024}%
  \BibitemOpen
  \bibfield  {author} {\bibinfo {author} {\bibfnamefont {M.}~\bibnamefont {Caleffi}}, \bibinfo {author} {\bibfnamefont {M.}~\bibnamefont {Amoretti}}, \bibinfo {author} {\bibfnamefont {D.}~\bibnamefont {Ferrari}}, \bibinfo {author} {\bibfnamefont {J.}~\bibnamefont {Illiano}}, \bibinfo {author} {\bibfnamefont {A.}~\bibnamefont {Manzalini}},\ and\ \bibinfo {author} {\bibfnamefont {A.~S.}\ \bibnamefont {Cacciapuoti}},\ }\bibfield  {title} {\bibinfo {title} {Distributed quantum computing: A survey},\ }\href {https://doi.org/10.1016/j.comnet.2024.110672} {\bibfield  {journal} {\bibinfo  {journal} {Computer Networks}\ }\textbf {\bibinfo {volume} {254}},\ \bibinfo {pages} {110672} (\bibinfo {year} {2024})}\BibitemShut {NoStop}%
\bibitem [{\citenamefont {Beals}\ \emph {et~al.}(2013)\citenamefont {Beals}, \citenamefont {Brierley}, \citenamefont {Gray}, \citenamefont {Harrow}, \citenamefont {Kutin}, \citenamefont {Linden}, \citenamefont {Shepherd},\ and\ \citenamefont {Stather}}]{beals2013efficient}%
  \BibitemOpen
  \bibfield  {author} {\bibinfo {author} {\bibfnamefont {R.}~\bibnamefont {Beals}}, \bibinfo {author} {\bibfnamefont {S.}~\bibnamefont {Brierley}}, \bibinfo {author} {\bibfnamefont {O.}~\bibnamefont {Gray}}, \bibinfo {author} {\bibfnamefont {A.~W.}\ \bibnamefont {Harrow}}, \bibinfo {author} {\bibfnamefont {S.}~\bibnamefont {Kutin}}, \bibinfo {author} {\bibfnamefont {N.}~\bibnamefont {Linden}}, \bibinfo {author} {\bibfnamefont {D.}~\bibnamefont {Shepherd}},\ and\ \bibinfo {author} {\bibfnamefont {M.}~\bibnamefont {Stather}},\ }\bibfield  {title} {\bibinfo {title} {Efficient distributed quantum computing},\ }\href {https://doi.org/https://doi.org/10.1098/rspa.2012.0686} {\bibfield  {journal} {\bibinfo  {journal} {Proceedings of the Royal Society A: Mathematical, Physical and Engineering Sciences}\ }\textbf {\bibinfo {volume} {469}},\ \bibinfo {pages} {20120686} (\bibinfo {year} {2013})}\BibitemShut {NoStop}%
\bibitem [{\citenamefont {Wu}\ \emph {et~al.}(2023)\citenamefont {Wu}, \citenamefont {Matsui}, \citenamefont {Forrer}, \citenamefont {Soeda}, \citenamefont {Andrés-Martínez}, \citenamefont {Mills}, \citenamefont {Henaut},\ and\ \citenamefont {Murao}}]{Wu_2023}%
  \BibitemOpen
  \bibfield  {author} {\bibinfo {author} {\bibfnamefont {J.-Y.}\ \bibnamefont {Wu}}, \bibinfo {author} {\bibfnamefont {K.}~\bibnamefont {Matsui}}, \bibinfo {author} {\bibfnamefont {T.}~\bibnamefont {Forrer}}, \bibinfo {author} {\bibfnamefont {A.}~\bibnamefont {Soeda}}, \bibinfo {author} {\bibfnamefont {P.}~\bibnamefont {Andrés-Martínez}}, \bibinfo {author} {\bibfnamefont {D.}~\bibnamefont {Mills}}, \bibinfo {author} {\bibfnamefont {L.}~\bibnamefont {Henaut}},\ and\ \bibinfo {author} {\bibfnamefont {M.}~\bibnamefont {Murao}},\ }\bibfield  {title} {\bibinfo {title} {Entanglement-efficient bipartite-distributed quantum computing},\ }\href {https://doi.org/10.22331/q-2023-12-05-1196} {\bibfield  {journal} {\bibinfo  {journal} {Quantum}\ }\textbf {\bibinfo {volume} {7}},\ \bibinfo {pages} {1196} (\bibinfo {year} {2023})}\BibitemShut {NoStop}%
\bibitem [{\citenamefont {Main}\ \emph {et~al.}(2025)\citenamefont {Main}, \citenamefont {Drmota}, \citenamefont {Nadlinger}, \citenamefont {Ainley}, \citenamefont {Agrawal}, \citenamefont {Nichol}, \citenamefont {Srinivas}, \citenamefont {Araneda},\ and\ \citenamefont {Lucas}}]{main2025distributed}%
  \BibitemOpen
  \bibfield  {author} {\bibinfo {author} {\bibfnamefont {D.}~\bibnamefont {Main}}, \bibinfo {author} {\bibfnamefont {P.}~\bibnamefont {Drmota}}, \bibinfo {author} {\bibfnamefont {D.}~\bibnamefont {Nadlinger}}, \bibinfo {author} {\bibfnamefont {E.}~\bibnamefont {Ainley}}, \bibinfo {author} {\bibfnamefont {A.}~\bibnamefont {Agrawal}}, \bibinfo {author} {\bibfnamefont {B.}~\bibnamefont {Nichol}}, \bibinfo {author} {\bibfnamefont {R.}~\bibnamefont {Srinivas}}, \bibinfo {author} {\bibfnamefont {G.}~\bibnamefont {Araneda}},\ and\ \bibinfo {author} {\bibfnamefont {D.}~\bibnamefont {Lucas}},\ }\bibfield  {title} {\bibinfo {title} {Distributed quantum computing across an optical network link},\ }\href {https://doi.org/https://doi.org/10.1038/s41586-024-08404-x} {\bibfield  {journal} {\bibinfo  {journal} {Nature}\ ,\ \bibinfo {pages} {1}} (\bibinfo {year} {2025})}\BibitemShut {NoStop}%
\bibitem [{\citenamefont {Kimble}(2008)}]{kimble2008quantum}%
  \BibitemOpen
  \bibfield  {author} {\bibinfo {author} {\bibfnamefont {H.~J.}\ \bibnamefont {Kimble}},\ }\bibfield  {title} {\bibinfo {title} {The quantum internet},\ }\href {https://doi.org/https://doi.org/10.1038/nature07127} {\bibfield  {journal} {\bibinfo  {journal} {Nature}\ }\textbf {\bibinfo {volume} {453}},\ \bibinfo {pages} {1023} (\bibinfo {year} {2008})}\BibitemShut {NoStop}%
\bibitem [{\citenamefont {Cao}\ \emph {et~al.}(2022)\citenamefont {Cao}, \citenamefont {Zhao}, \citenamefont {Wang}, \citenamefont {Zhang}, \citenamefont {Ng},\ and\ \citenamefont {Hanzo}}]{Cao_QI}%
  \BibitemOpen
  \bibfield  {author} {\bibinfo {author} {\bibfnamefont {Y.}~\bibnamefont {Cao}}, \bibinfo {author} {\bibfnamefont {Y.}~\bibnamefont {Zhao}}, \bibinfo {author} {\bibfnamefont {Q.}~\bibnamefont {Wang}}, \bibinfo {author} {\bibfnamefont {J.}~\bibnamefont {Zhang}}, \bibinfo {author} {\bibfnamefont {S.~X.}\ \bibnamefont {Ng}},\ and\ \bibinfo {author} {\bibfnamefont {L.}~\bibnamefont {Hanzo}},\ }\bibfield  {title} {\bibinfo {title} {The evolution of quantum key distribution networks: On the road to the qinternet},\ }\href {https://doi.org/10.1109/COMST.2022.3144219} {\bibfield  {journal} {\bibinfo  {journal} {IEEE Communications Surveys and Tutorials}\ }\textbf {\bibinfo {volume} {24}},\ \bibinfo {pages} {839} (\bibinfo {year} {2022})}\BibitemShut {NoStop}%
\bibitem [{\citenamefont {Azuma}\ \emph {et~al.}(2023)\citenamefont {Azuma}, \citenamefont {Economou}, \citenamefont {Elkouss}, \citenamefont {Hilaire}, \citenamefont {Jiang}, \citenamefont {Lo},\ and\ \citenamefont {Tzitrin}}]{azuma_QI}%
  \BibitemOpen
  \bibfield  {author} {\bibinfo {author} {\bibfnamefont {K.}~\bibnamefont {Azuma}}, \bibinfo {author} {\bibfnamefont {S.~E.}\ \bibnamefont {Economou}}, \bibinfo {author} {\bibfnamefont {D.}~\bibnamefont {Elkouss}}, \bibinfo {author} {\bibfnamefont {P.}~\bibnamefont {Hilaire}}, \bibinfo {author} {\bibfnamefont {L.}~\bibnamefont {Jiang}}, \bibinfo {author} {\bibfnamefont {H.-K.}\ \bibnamefont {Lo}},\ and\ \bibinfo {author} {\bibfnamefont {I.}~\bibnamefont {Tzitrin}},\ }\bibfield  {title} {\bibinfo {title} {Quantum repeaters: From quantum networks to the quantum internet},\ }\href {https://doi.org/10.1103/RevModPhys.95.045006} {\bibfield  {journal} {\bibinfo  {journal} {Rev. Mod. Phys.}\ }\textbf {\bibinfo {volume} {95}},\ \bibinfo {pages} {045006} (\bibinfo {year} {2023})}\BibitemShut {NoStop}%
\bibitem [{\citenamefont {Koudia}(2023)}]{koudia2023quantum}%
  \BibitemOpen
  \bibfield  {author} {\bibinfo {author} {\bibfnamefont {S.}~\bibnamefont {Koudia}},\ }\bibfield  {title} {\bibinfo {title} {The quantum internet: an efficient stabilizer states distribution scheme},\ }\href {https://doi.org/10.1088/1402-4896/ad1565} {\bibfield  {journal} {\bibinfo  {journal} {Physica Scripta}\ }\textbf {\bibinfo {volume} {99}},\ \bibinfo {pages} {015115} (\bibinfo {year} {2023})}\BibitemShut {NoStop}%
\bibitem [{\citenamefont {Sun}\ \emph {et~al.}(2025)\citenamefont {Sun}, \citenamefont {Cheng},\ and\ \citenamefont {Zhang}}]{Sun_2025}%
  \BibitemOpen
  \bibfield  {author} {\bibinfo {author} {\bibfnamefont {J.}~\bibnamefont {Sun}}, \bibinfo {author} {\bibfnamefont {L.}~\bibnamefont {Cheng}},\ and\ \bibinfo {author} {\bibfnamefont {S.-X.}\ \bibnamefont {Zhang}},\ }\bibfield  {title} {\bibinfo {title} {Stabilizer ground states for simulating quantum many-body physics: theory, algorithms, and applications},\ }\href {https://doi.org/10.22331/q-2025-06-24-1782} {\bibfield  {journal} {\bibinfo  {journal} {Quantum}\ }\textbf {\bibinfo {volume} {9}},\ \bibinfo {pages} {1782} (\bibinfo {year} {2025})}\BibitemShut {NoStop}%
\bibitem [{\citenamefont {Gu}\ \emph {et~al.}(2024)\citenamefont {Gu}, \citenamefont {Oliviero},\ and\ \citenamefont {Leone}}]{Gu_many_body}%
  \BibitemOpen
  \bibfield  {author} {\bibinfo {author} {\bibfnamefont {A.}~\bibnamefont {Gu}}, \bibinfo {author} {\bibfnamefont {S.~F.~E.}\ \bibnamefont {Oliviero}},\ and\ \bibinfo {author} {\bibfnamefont {L.}~\bibnamefont {Leone}},\ }\bibfield  {title} {\bibinfo {title} {Doped stabilizer states in many-body physics and where to find them},\ }\href {https://doi.org/10.1103/PhysRevA.110.062427} {\bibfield  {journal} {\bibinfo  {journal} {Phys. Rev. A}\ }\textbf {\bibinfo {volume} {110}},\ \bibinfo {pages} {062427} (\bibinfo {year} {2024})}\BibitemShut {NoStop}%
\bibitem [{\citenamefont {Antu}\ and\ \citenamefont {Zhou}(2025)}]{Antu_2025}%
  \BibitemOpen
  \bibfield  {author} {\bibinfo {author} {\bibfnamefont {S.~B.}\ \bibnamefont {Antu}}\ and\ \bibinfo {author} {\bibfnamefont {S.}~\bibnamefont {Zhou}},\ }\bibfield  {title} {\bibinfo {title} {Stabilizer codes for heisenberg-limited many-body hamiltonian estimation},\ }\href {https://doi.org/10.22331/q-2025-06-05-1766} {\bibfield  {journal} {\bibinfo  {journal} {Quantum}\ }\textbf {\bibinfo {volume} {9}},\ \bibinfo {pages} {1766} (\bibinfo {year} {2025})}\BibitemShut {NoStop}%
\bibitem [{\citenamefont {Liu}\ and\ \citenamefont {Winter}(2022)}]{Liu_many_body}%
  \BibitemOpen
  \bibfield  {author} {\bibinfo {author} {\bibfnamefont {Z.-W.}\ \bibnamefont {Liu}}\ and\ \bibinfo {author} {\bibfnamefont {A.}~\bibnamefont {Winter}},\ }\bibfield  {title} {\bibinfo {title} {Many-body quantum magic},\ }\href {https://doi.org/10.1103/PRXQuantum.3.020333} {\bibfield  {journal} {\bibinfo  {journal} {PRX Quantum}\ }\textbf {\bibinfo {volume} {3}},\ \bibinfo {pages} {020333} (\bibinfo {year} {2022})}\BibitemShut {NoStop}%
\bibitem [{\citenamefont {Hein}\ \emph {et~al.}(2006)\citenamefont {Hein}, \citenamefont {D{\"u}r}, \citenamefont {Eisert}, \citenamefont {Raussendorf}, \citenamefont {Van~den Nest},\ and\ \citenamefont {Briegel}}]{hein2006entanglement}%
  \BibitemOpen
  \bibfield  {author} {\bibinfo {author} {\bibfnamefont {M.}~\bibnamefont {Hein}}, \bibinfo {author} {\bibfnamefont {W.}~\bibnamefont {D{\"u}r}}, \bibinfo {author} {\bibfnamefont {J.}~\bibnamefont {Eisert}}, \bibinfo {author} {\bibfnamefont {R.}~\bibnamefont {Raussendorf}}, \bibinfo {author} {\bibfnamefont {M.}~\bibnamefont {Van~den Nest}},\ and\ \bibinfo {author} {\bibfnamefont {H.-J.}\ \bibnamefont {Briegel}},\ }\bibfield  {title} {\bibinfo {title} {Entanglement in graph states and its applications},\ }in\ \href@noop {} {\emph {\bibinfo {booktitle} {Quantum computers, algorithms and chaos}}}\ (\bibinfo  {publisher} {IOS Press},\ \bibinfo {year} {2006})\ pp.\ \bibinfo {pages} {115--218}\BibitemShut {NoStop}%
\bibitem [{\citenamefont {Hein}\ \emph {et~al.}(2004)\citenamefont {Hein}, \citenamefont {Eisert},\ and\ \citenamefont {Briegel}}]{hein2004}%
  \BibitemOpen
  \bibfield  {author} {\bibinfo {author} {\bibfnamefont {M.}~\bibnamefont {Hein}}, \bibinfo {author} {\bibfnamefont {J.}~\bibnamefont {Eisert}},\ and\ \bibinfo {author} {\bibfnamefont {H.~J.}\ \bibnamefont {Briegel}},\ }\bibfield  {title} {\bibinfo {title} {Multiparty entanglement in graph states},\ }\href {https://doi.org/10.1103/PhysRevA.69.062311} {\bibfield  {journal} {\bibinfo  {journal} {Phys. Rev. A}\ }\textbf {\bibinfo {volume} {69}},\ \bibinfo {pages} {062311} (\bibinfo {year} {2004})}\BibitemShut {NoStop}%
\bibitem [{\citenamefont {Shettell}\ and\ \citenamefont {Markham}(2020)}]{shettell2020graph}%
  \BibitemOpen
  \bibfield  {author} {\bibinfo {author} {\bibfnamefont {N.}~\bibnamefont {Shettell}}\ and\ \bibinfo {author} {\bibfnamefont {D.}~\bibnamefont {Markham}},\ }\bibfield  {title} {\bibinfo {title} {Graph states as a resource for quantum metrology},\ }\href@noop {} {\bibfield  {journal} {\bibinfo  {journal} {Physical review letters}\ }\textbf {\bibinfo {volume} {124}},\ \bibinfo {pages} {110502} (\bibinfo {year} {2020})}\BibitemShut {NoStop}%
\bibitem [{\citenamefont {Contreras-Tejada}\ \emph {et~al.}(2019)\citenamefont {Contreras-Tejada}, \citenamefont {Palazuelos},\ and\ \citenamefont {De~Vicente}}]{contreras2019resource}%
  \BibitemOpen
  \bibfield  {author} {\bibinfo {author} {\bibfnamefont {P.}~\bibnamefont {Contreras-Tejada}}, \bibinfo {author} {\bibfnamefont {C.}~\bibnamefont {Palazuelos}},\ and\ \bibinfo {author} {\bibfnamefont {J.~I.}\ \bibnamefont {De~Vicente}},\ }\bibfield  {title} {\bibinfo {title} {Resource theory of entanglement with a unique multipartite maximally entangled state},\ }\href@noop {} {\bibfield  {journal} {\bibinfo  {journal} {Physical review letters}\ }\textbf {\bibinfo {volume} {122}},\ \bibinfo {pages} {120503} (\bibinfo {year} {2019})}\BibitemShut {NoStop}%
\bibitem [{\citenamefont {D'Hondt}\ and\ \citenamefont {Panangaden}(2004)}]{d2004computational}%
  \BibitemOpen
  \bibfield  {author} {\bibinfo {author} {\bibfnamefont {E.}~\bibnamefont {D'Hondt}}\ and\ \bibinfo {author} {\bibfnamefont {P.}~\bibnamefont {Panangaden}},\ }\bibfield  {title} {\bibinfo {title} {The computational power of the w and ghz states},\ }\href@noop {} {\bibfield  {journal} {\bibinfo  {journal} {arXiv preprint quant-ph/0412177}\ } (\bibinfo {year} {2004})}\BibitemShut {NoStop}%
\bibitem [{\citenamefont {Kitaev}(2003)}]{kitaev2003fault}%
  \BibitemOpen
  \bibfield  {author} {\bibinfo {author} {\bibfnamefont {A.~Y.}\ \bibnamefont {Kitaev}},\ }\bibfield  {title} {\bibinfo {title} {Fault-tolerant quantum computation by anyons},\ }\href {https://doi.org/https://doi.org/10.1016/S0003-4916(02)00018-0} {\bibfield  {journal} {\bibinfo  {journal} {Annals of physics}\ }\textbf {\bibinfo {volume} {303}},\ \bibinfo {pages} {2} (\bibinfo {year} {2003})}\BibitemShut {NoStop}%
\bibitem [{\citenamefont {Bonilla~Ataides}\ \emph {et~al.}(2021)\citenamefont {Bonilla~Ataides}, \citenamefont {Tuckett}, \citenamefont {Bartlett}, \citenamefont {Flammia},\ and\ \citenamefont {Brown}}]{bonilla2021xzzx}%
  \BibitemOpen
  \bibfield  {author} {\bibinfo {author} {\bibfnamefont {J.~P.}\ \bibnamefont {Bonilla~Ataides}}, \bibinfo {author} {\bibfnamefont {D.~K.}\ \bibnamefont {Tuckett}}, \bibinfo {author} {\bibfnamefont {S.~D.}\ \bibnamefont {Bartlett}}, \bibinfo {author} {\bibfnamefont {S.~T.}\ \bibnamefont {Flammia}},\ and\ \bibinfo {author} {\bibfnamefont {B.~J.}\ \bibnamefont {Brown}},\ }\bibfield  {title} {\bibinfo {title} {The xzzx surface code},\ }\href {https://doi.org/https://doi.org/10.1038/s41467-021-22274-1} {\bibfield  {journal} {\bibinfo  {journal} {Nature communications}\ }\textbf {\bibinfo {volume} {12}},\ \bibinfo {pages} {2172} (\bibinfo {year} {2021})}\BibitemShut {NoStop}%
\bibitem [{\citenamefont {Castelnovo}(2013)}]{castelnuovo_toric}%
  \BibitemOpen
  \bibfield  {author} {\bibinfo {author} {\bibfnamefont {C.}~\bibnamefont {Castelnovo}},\ }\bibfield  {title} {\bibinfo {title} {Negativity and topological order in the toric code},\ }\href {https://doi.org/10.1103/PhysRevA.88.042319} {\bibfield  {journal} {\bibinfo  {journal} {Phys. Rev. A}\ }\textbf {\bibinfo {volume} {88}},\ \bibinfo {pages} {042319} (\bibinfo {year} {2013})}\BibitemShut {NoStop}%
\bibitem [{\citenamefont {Zhao~et al.}(2022)}]{zhao_surface}%
  \BibitemOpen
  \bibfield  {author} {\bibinfo {author} {\bibfnamefont {Y.}~\bibnamefont {Zhao~et al.}},\ }\bibfield  {title} {\bibinfo {title} {Realization of an error-correcting surface code with superconducting qubits},\ }\href {https://doi.org/10.1103/PhysRevLett.129.030501} {\bibfield  {journal} {\bibinfo  {journal} {Phys. Rev. Lett.}\ }\textbf {\bibinfo {volume} {129}},\ \bibinfo {pages} {030501} (\bibinfo {year} {2022})}\BibitemShut {NoStop}%
\bibitem [{\citenamefont {Gottesman}(1998)}]{gottesman1998}%
  \BibitemOpen
  \bibfield  {author} {\bibinfo {author} {\bibfnamefont {D.}~\bibnamefont {Gottesman}},\ }\href {https://arxiv.org/abs/quant-ph/9807006} {\bibinfo {title} {The heisenberg representation of quantum computers}} (\bibinfo {year} {1998}),\ \Eprint {https://arxiv.org/abs/quant-ph/9807006} {arXiv:quant-ph/9807006 [quant-ph]} \BibitemShut {NoStop}%
\bibitem [{\citenamefont {Brandhofer}\ \emph {et~al.}(2025)\citenamefont {Brandhofer}, \citenamefont {Polian}, \citenamefont {Barz},\ and\ \citenamefont {Bhatti}}]{brandhofer2025hardware}%
  \BibitemOpen
  \bibfield  {author} {\bibinfo {author} {\bibfnamefont {S.}~\bibnamefont {Brandhofer}}, \bibinfo {author} {\bibfnamefont {I.}~\bibnamefont {Polian}}, \bibinfo {author} {\bibfnamefont {S.}~\bibnamefont {Barz}},\ and\ \bibinfo {author} {\bibfnamefont {D.}~\bibnamefont {Bhatti}},\ }\bibfield  {title} {\bibinfo {title} {Hardware-efficient preparation of architecture-specific graph states on near-term quantum computers},\ }\href {https://doi.org/https://doi.org/10.1038/s41598-024-82715-x} {\bibfield  {journal} {\bibinfo  {journal} {Scientific Reports}\ }\textbf {\bibinfo {volume} {15}},\ \bibinfo {pages} {2095} (\bibinfo {year} {2025})}\BibitemShut {NoStop}%
\bibitem [{\citenamefont {Nielsen}\ and\ \citenamefont {Chuang}(2000)}]{nielsen2000quantum}%
  \BibitemOpen
  \bibfield  {author} {\bibinfo {author} {\bibfnamefont {M.~A.}\ \bibnamefont {Nielsen}}\ and\ \bibinfo {author} {\bibfnamefont {I.~L.}\ \bibnamefont {Chuang}},\ }\href {https://doi.org/https://doi.org/10.1017/CBO9780511976667} {\emph {\bibinfo {title} {Quantum computing and quantum information}}}\ (\bibinfo  {publisher} {Cambridge University Press, Cambridge},\ \bibinfo {year} {2000})\BibitemShut {NoStop}%
\bibitem [{\citenamefont {Benenti}\ \emph {et~al.}(2019)\citenamefont {Benenti}, \citenamefont {Casati}, \citenamefont {Rossini},\ and\ \citenamefont {Strini}}]{benenti2019principles}%
  \BibitemOpen
  \bibfield  {author} {\bibinfo {author} {\bibfnamefont {G.}~\bibnamefont {Benenti}}, \bibinfo {author} {\bibfnamefont {G.}~\bibnamefont {Casati}}, \bibinfo {author} {\bibfnamefont {D.}~\bibnamefont {Rossini}},\ and\ \bibinfo {author} {\bibfnamefont {G.}~\bibnamefont {Strini}},\ }\href@noop {} {\emph {\bibinfo {title} {Principles of Quantum Computation and Information: A Comprehensive Textbook}}}\ (\bibinfo  {publisher} {World Scientific},\ \bibinfo {year} {2019})\BibitemShut {NoStop}%
\bibitem [{\citenamefont {Wallman}\ \emph {et~al.}(2015)\citenamefont {Wallman}, \citenamefont {Granade}, \citenamefont {Harper},\ and\ \citenamefont {Flammia}}]{Wallman_2015}%
  \BibitemOpen
  \bibfield  {author} {\bibinfo {author} {\bibfnamefont {J.}~\bibnamefont {Wallman}}, \bibinfo {author} {\bibfnamefont {C.}~\bibnamefont {Granade}}, \bibinfo {author} {\bibfnamefont {R.}~\bibnamefont {Harper}},\ and\ \bibinfo {author} {\bibfnamefont {S.~T.}\ \bibnamefont {Flammia}},\ }\bibfield  {title} {\bibinfo {title} {Estimating the coherence of noise},\ }\href {https://doi.org/10.1088/1367-2630/17/11/113020} {\bibfield  {journal} {\bibinfo  {journal} {New Journal of Physics}\ }\textbf {\bibinfo {volume} {17}},\ \bibinfo {pages} {113020} (\bibinfo {year} {2015})}\BibitemShut {NoStop}%
\bibitem [{\citenamefont {Ball}\ \emph {et~al.}(2016)\citenamefont {Ball}, \citenamefont {Stace}, \citenamefont {Flammia},\ and\ \citenamefont {Biercuk}}]{ball_2016}%
  \BibitemOpen
  \bibfield  {author} {\bibinfo {author} {\bibfnamefont {H.}~\bibnamefont {Ball}}, \bibinfo {author} {\bibfnamefont {T.~M.}\ \bibnamefont {Stace}}, \bibinfo {author} {\bibfnamefont {S.~T.}\ \bibnamefont {Flammia}},\ and\ \bibinfo {author} {\bibfnamefont {M.~J.}\ \bibnamefont {Biercuk}},\ }\bibfield  {title} {\bibinfo {title} {Effect of noise correlations on randomized benchmarking},\ }\href {https://doi.org/10.1103/PhysRevA.93.022303} {\bibfield  {journal} {\bibinfo  {journal} {Phys. Rev. A}\ }\textbf {\bibinfo {volume} {93}},\ \bibinfo {pages} {022303} (\bibinfo {year} {2016})}\BibitemShut {NoStop}%
\bibitem [{\citenamefont {Kueng}\ \emph {et~al.}(2016)\citenamefont {Kueng}, \citenamefont {Long}, \citenamefont {Doherty},\ and\ \citenamefont {Flammia}}]{kueng_2016}%
  \BibitemOpen
  \bibfield  {author} {\bibinfo {author} {\bibfnamefont {R.}~\bibnamefont {Kueng}}, \bibinfo {author} {\bibfnamefont {D.~M.}\ \bibnamefont {Long}}, \bibinfo {author} {\bibfnamefont {A.~C.}\ \bibnamefont {Doherty}},\ and\ \bibinfo {author} {\bibfnamefont {S.~T.}\ \bibnamefont {Flammia}},\ }\bibfield  {title} {\bibinfo {title} {Comparing experiments to the fault-tolerance threshold},\ }\href {https://doi.org/10.1103/PhysRevLett.117.170502} {\bibfield  {journal} {\bibinfo  {journal} {Phys. Rev. Lett.}\ }\textbf {\bibinfo {volume} {117}},\ \bibinfo {pages} {170502} (\bibinfo {year} {2016})}\BibitemShut {NoStop}%
\bibitem [{\citenamefont {Wallman}(2016)}]{wallman_2016}%
  \BibitemOpen
  \bibfield  {author} {\bibinfo {author} {\bibfnamefont {J.~J.}\ \bibnamefont {Wallman}},\ }\href {https://arxiv.org/abs/1511.00727} {\bibinfo {title} {Bounding experimental quantum error rates relative to fault-tolerant thresholds}} (\bibinfo {year} {2016}),\ \Eprint {https://arxiv.org/abs/1511.00727} {arXiv:1511.00727 [quant-ph]} \BibitemShut {NoStop}%
\bibitem [{\citenamefont {Greenbaum}\ and\ \citenamefont {Dutton}(2017)}]{greenbaum2017modeling}%
  \BibitemOpen
  \bibfield  {author} {\bibinfo {author} {\bibfnamefont {D.}~\bibnamefont {Greenbaum}}\ and\ \bibinfo {author} {\bibfnamefont {Z.}~\bibnamefont {Dutton}},\ }\bibfield  {title} {\bibinfo {title} {Modeling coherent errors in quantum error correction},\ }\href {https://doi.org/10.1088/2058-9565/aa9a06} {\bibfield  {journal} {\bibinfo  {journal} {Quantum Science and Technology}\ }\textbf {\bibinfo {volume} {3}},\ \bibinfo {pages} {015007} (\bibinfo {year} {2017})}\BibitemShut {NoStop}%
\bibitem [{\citenamefont {Huang}\ \emph {et~al.}(2019)\citenamefont {Huang}, \citenamefont {Doherty},\ and\ \citenamefont {Flammia}}]{Huang_2019}%
  \BibitemOpen
  \bibfield  {author} {\bibinfo {author} {\bibfnamefont {E.}~\bibnamefont {Huang}}, \bibinfo {author} {\bibfnamefont {A.~C.}\ \bibnamefont {Doherty}},\ and\ \bibinfo {author} {\bibfnamefont {S.}~\bibnamefont {Flammia}},\ }\bibfield  {title} {\bibinfo {title} {Performance of quantum error correction with coherent errors},\ }\href {https://doi.org/10.1103/PhysRevA.99.022313} {\bibfield  {journal} {\bibinfo  {journal} {Phys. Rev. A}\ }\textbf {\bibinfo {volume} {99}},\ \bibinfo {pages} {022313} (\bibinfo {year} {2019})}\BibitemShut {NoStop}%
\bibitem [{\citenamefont {Bravyi}\ \emph {et~al.}(2018)\citenamefont {Bravyi}, \citenamefont {Englbrecht}, \citenamefont {K{\"o}nig},\ and\ \citenamefont {Peard}}]{bravyi2018correcting}%
  \BibitemOpen
  \bibfield  {author} {\bibinfo {author} {\bibfnamefont {S.}~\bibnamefont {Bravyi}}, \bibinfo {author} {\bibfnamefont {M.}~\bibnamefont {Englbrecht}}, \bibinfo {author} {\bibfnamefont {R.}~\bibnamefont {K{\"o}nig}},\ and\ \bibinfo {author} {\bibfnamefont {N.}~\bibnamefont {Peard}},\ }\bibfield  {title} {\bibinfo {title} {Correcting coherent errors with surface codes},\ }\href {https://doi.org/https://doi.org/10.1038/s41534-018-0106-y} {\bibfield  {journal} {\bibinfo  {journal} {npj Quantum Information}\ }\textbf {\bibinfo {volume} {4}},\ \bibinfo {pages} {55} (\bibinfo {year} {2018})}\BibitemShut {NoStop}%
\bibitem [{\citenamefont {Kern}\ \emph {et~al.}(2005)\citenamefont {Kern}, \citenamefont {Alber},\ and\ \citenamefont {Shepelyansky}}]{kern2005quantum}%
  \BibitemOpen
  \bibfield  {author} {\bibinfo {author} {\bibfnamefont {O.}~\bibnamefont {Kern}}, \bibinfo {author} {\bibfnamefont {G.}~\bibnamefont {Alber}},\ and\ \bibinfo {author} {\bibfnamefont {D.~L.}\ \bibnamefont {Shepelyansky}},\ }\bibfield  {title} {\bibinfo {title} {Quantum error correction of coherent errors by randomization},\ }\href {https://doi.org/https://doi.org/10.1140/epjd/e2004-00196-9} {\bibfield  {journal} {\bibinfo  {journal} {The European Physical Journal D-Atomic, Molecular, Optical and Plasma Physics}\ }\textbf {\bibinfo {volume} {32}},\ \bibinfo {pages} {153} (\bibinfo {year} {2005})}\BibitemShut {NoStop}%
\bibitem [{\citenamefont {Cai}\ \emph {et~al.}(2023)\citenamefont {Cai}, \citenamefont {Babbush}, \citenamefont {Benjamin}, \citenamefont {Endo}, \citenamefont {Huggins}, \citenamefont {Li}, \citenamefont {McClean},\ and\ \citenamefont {O'Brien}}]{endo_QEM}%
  \BibitemOpen
  \bibfield  {author} {\bibinfo {author} {\bibfnamefont {Z.}~\bibnamefont {Cai}}, \bibinfo {author} {\bibfnamefont {R.}~\bibnamefont {Babbush}}, \bibinfo {author} {\bibfnamefont {S.~C.}\ \bibnamefont {Benjamin}}, \bibinfo {author} {\bibfnamefont {S.}~\bibnamefont {Endo}}, \bibinfo {author} {\bibfnamefont {W.~J.}\ \bibnamefont {Huggins}}, \bibinfo {author} {\bibfnamefont {Y.}~\bibnamefont {Li}}, \bibinfo {author} {\bibfnamefont {J.~R.}\ \bibnamefont {McClean}},\ and\ \bibinfo {author} {\bibfnamefont {T.~E.}\ \bibnamefont {O'Brien}},\ }\bibfield  {title} {\bibinfo {title} {Quantum error mitigation},\ }\href {https://doi.org/10.1103/RevModPhys.95.045005} {\bibfield  {journal} {\bibinfo  {journal} {Rev. Mod. Phys.}\ }\textbf {\bibinfo {volume} {95}},\ \bibinfo {pages} {045005} (\bibinfo {year} {2023})}\BibitemShut {NoStop}%
\bibitem [{\citenamefont {Campbell}(2017)}]{campbell_QEM}%
  \BibitemOpen
  \bibfield  {author} {\bibinfo {author} {\bibfnamefont {E.}~\bibnamefont {Campbell}},\ }\bibfield  {title} {\bibinfo {title} {Shorter gate sequences for quantum computing by mixing unitaries},\ }\href {https://doi.org/10.1103/PhysRevA.95.042306} {\bibfield  {journal} {\bibinfo  {journal} {Phys. Rev. A}\ }\textbf {\bibinfo {volume} {95}},\ \bibinfo {pages} {042306} (\bibinfo {year} {2017})}\BibitemShut {NoStop}%
\bibitem [{\citenamefont {Endo}\ \emph {et~al.}(2018)\citenamefont {Endo}, \citenamefont {Benjamin},\ and\ \citenamefont {Li}}]{endo_practical_QEM}%
  \BibitemOpen
  \bibfield  {author} {\bibinfo {author} {\bibfnamefont {S.}~\bibnamefont {Endo}}, \bibinfo {author} {\bibfnamefont {S.~C.}\ \bibnamefont {Benjamin}},\ and\ \bibinfo {author} {\bibfnamefont {Y.}~\bibnamefont {Li}},\ }\bibfield  {title} {\bibinfo {title} {Practical quantum error mitigation for near-future applications},\ }\href {https://doi.org/10.1103/PhysRevX.8.031027} {\bibfield  {journal} {\bibinfo  {journal} {Phys. Rev. X}\ }\textbf {\bibinfo {volume} {8}},\ \bibinfo {pages} {031027} (\bibinfo {year} {2018})}\BibitemShut {NoStop}%
\bibitem [{\citenamefont {Yoshioka}\ \emph {et~al.}(2022)\citenamefont {Yoshioka}, \citenamefont {Hakoshima}, \citenamefont {Matsuzaki}, \citenamefont {Tokunaga}, \citenamefont {Suzuki},\ and\ \citenamefont {Endo}}]{endo_QEM_coherent}%
  \BibitemOpen
  \bibfield  {author} {\bibinfo {author} {\bibfnamefont {N.}~\bibnamefont {Yoshioka}}, \bibinfo {author} {\bibfnamefont {H.}~\bibnamefont {Hakoshima}}, \bibinfo {author} {\bibfnamefont {Y.}~\bibnamefont {Matsuzaki}}, \bibinfo {author} {\bibfnamefont {Y.}~\bibnamefont {Tokunaga}}, \bibinfo {author} {\bibfnamefont {Y.}~\bibnamefont {Suzuki}},\ and\ \bibinfo {author} {\bibfnamefont {S.}~\bibnamefont {Endo}},\ }\bibfield  {title} {\bibinfo {title} {Generalized quantum subspace expansion},\ }\href {https://doi.org/10.1103/PhysRevLett.129.020502} {\bibfield  {journal} {\bibinfo  {journal} {Phys. Rev. Lett.}\ }\textbf {\bibinfo {volume} {129}},\ \bibinfo {pages} {020502} (\bibinfo {year} {2022})}\BibitemShut {NoStop}%
\bibitem [{\citenamefont {Orsucci}\ \emph {et~al.}(2016)\citenamefont {Orsucci}, \citenamefont {Tiersch},\ and\ \citenamefont {Briegel}}]{orsucci_coherent}%
  \BibitemOpen
  \bibfield  {author} {\bibinfo {author} {\bibfnamefont {D.}~\bibnamefont {Orsucci}}, \bibinfo {author} {\bibfnamefont {M.}~\bibnamefont {Tiersch}},\ and\ \bibinfo {author} {\bibfnamefont {H.~J.}\ \bibnamefont {Briegel}},\ }\bibfield  {title} {\bibinfo {title} {Estimation of coherent error sources from stabilizer measurements},\ }\href {https://doi.org/10.1103/PhysRevA.93.042303} {\bibfield  {journal} {\bibinfo  {journal} {Phys. Rev. A}\ }\textbf {\bibinfo {volume} {93}},\ \bibinfo {pages} {042303} (\bibinfo {year} {2016})}\BibitemShut {NoStop}%
\bibitem [{\citenamefont {Temme}\ \emph {et~al.}(2017)\citenamefont {Temme}, \citenamefont {Bravyi},\ and\ \citenamefont {Gambetta}}]{temme_PEC}%
  \BibitemOpen
  \bibfield  {author} {\bibinfo {author} {\bibfnamefont {K.}~\bibnamefont {Temme}}, \bibinfo {author} {\bibfnamefont {S.}~\bibnamefont {Bravyi}},\ and\ \bibinfo {author} {\bibfnamefont {J.~M.}\ \bibnamefont {Gambetta}},\ }\bibfield  {title} {\bibinfo {title} {Error mitigation for short-depth quantum circuits},\ }\href {https://doi.org/10.1103/PhysRevLett.119.180509} {\bibfield  {journal} {\bibinfo  {journal} {Phys. Rev. Lett.}\ }\textbf {\bibinfo {volume} {119}},\ \bibinfo {pages} {180509} (\bibinfo {year} {2017})}\BibitemShut {NoStop}%
\bibitem [{\citenamefont {Van Den~Berg}\ \emph {et~al.}(2023)\citenamefont {Van Den~Berg}, \citenamefont {Minev}, \citenamefont {Kandala},\ and\ \citenamefont {Temme}}]{van2023probabilistic}%
  \BibitemOpen
  \bibfield  {author} {\bibinfo {author} {\bibfnamefont {E.}~\bibnamefont {Van Den~Berg}}, \bibinfo {author} {\bibfnamefont {Z.~K.}\ \bibnamefont {Minev}}, \bibinfo {author} {\bibfnamefont {A.}~\bibnamefont {Kandala}},\ and\ \bibinfo {author} {\bibfnamefont {K.}~\bibnamefont {Temme}},\ }\bibfield  {title} {\bibinfo {title} {Probabilistic error cancellation with sparse pauli--lindblad models on noisy quantum processors},\ }\href@noop {} {\bibfield  {journal} {\bibinfo  {journal} {Nature physics}\ }\textbf {\bibinfo {volume} {19}},\ \bibinfo {pages} {1116} (\bibinfo {year} {2023})}\BibitemShut {NoStop}%
\bibitem [{\citenamefont {Filippov}\ \emph {et~al.}(2023)\citenamefont {Filippov}, \citenamefont {Leahy}, \citenamefont {Rossi},\ and\ \citenamefont {García-Pérez}}]{filippov2023}%
  \BibitemOpen
  \bibfield  {author} {\bibinfo {author} {\bibfnamefont {S.}~\bibnamefont {Filippov}}, \bibinfo {author} {\bibfnamefont {M.}~\bibnamefont {Leahy}}, \bibinfo {author} {\bibfnamefont {M.~A.~C.}\ \bibnamefont {Rossi}},\ and\ \bibinfo {author} {\bibfnamefont {G.}~\bibnamefont {García-Pérez}},\ }\href {https://arxiv.org/abs/2307.11740} {\bibinfo {title} {Scalable tensor-network error mitigation for near-term quantum computing}} (\bibinfo {year} {2023}),\ \Eprint {https://arxiv.org/abs/2307.11740} {arXiv:2307.11740 [quant-ph]} \BibitemShut {NoStop}%
\bibitem [{\citenamefont {Bultrini}\ \emph {et~al.}(2023)\citenamefont {Bultrini}, \citenamefont {Gordon}, \citenamefont {Czarnik}, \citenamefont {Arrasmith}, \citenamefont {Cerezo}, \citenamefont {Coles},\ and\ \citenamefont {Cincio}}]{bultrini2023unifying}%
  \BibitemOpen
  \bibfield  {author} {\bibinfo {author} {\bibfnamefont {D.}~\bibnamefont {Bultrini}}, \bibinfo {author} {\bibfnamefont {M.~H.}\ \bibnamefont {Gordon}}, \bibinfo {author} {\bibfnamefont {P.}~\bibnamefont {Czarnik}}, \bibinfo {author} {\bibfnamefont {A.}~\bibnamefont {Arrasmith}}, \bibinfo {author} {\bibfnamefont {M.}~\bibnamefont {Cerezo}}, \bibinfo {author} {\bibfnamefont {P.~J.}\ \bibnamefont {Coles}},\ and\ \bibinfo {author} {\bibfnamefont {L.}~\bibnamefont {Cincio}},\ }\bibfield  {title} {\bibinfo {title} {Unifying and benchmarking state-of-the-art quantum error mitigation techniques},\ }\href {https://doi.org/https://doi.org/10.22331/q-2023-06-06-1034} {\bibfield  {journal} {\bibinfo  {journal} {Quantum}\ }\textbf {\bibinfo {volume} {7}},\ \bibinfo {pages} {1034} (\bibinfo {year} {2023})}\BibitemShut {NoStop}%
\bibitem [{\citenamefont {Arab}(2024)}]{arab2024lecture}%
  \BibitemOpen
  \bibfield  {author} {\bibinfo {author} {\bibfnamefont {A.~R.}\ \bibnamefont {Arab}},\ }\bibfield  {title} {\bibinfo {title} {Lecture notes on quantum entanglement: From stabilizer states to stabilizer channels},\ }\href@noop {} {\bibfield  {journal} {\bibinfo  {journal} {Frontiers of Physics}\ }\textbf {\bibinfo {volume} {19}},\ \bibinfo {pages} {51203} (\bibinfo {year} {2024})}\BibitemShut {NoStop}%
\bibitem [{\citenamefont {Larocca}\ \emph {et~al.}(2025)\citenamefont {Larocca}, \citenamefont {Thanasilp}, \citenamefont {Wang}, \citenamefont {Sharma}, \citenamefont {Biamonte}, \citenamefont {Coles}, \citenamefont {Cincio}, \citenamefont {McClean}, \citenamefont {Holmes},\ and\ \citenamefont {Cerezo}}]{Larocca25}%
  \BibitemOpen
  \bibfield  {author} {\bibinfo {author} {\bibfnamefont {M.}~\bibnamefont {Larocca}}, \bibinfo {author} {\bibfnamefont {S.}~\bibnamefont {Thanasilp}}, \bibinfo {author} {\bibfnamefont {S.}~\bibnamefont {Wang}}, \bibinfo {author} {\bibfnamefont {K.}~\bibnamefont {Sharma}}, \bibinfo {author} {\bibfnamefont {J.}~\bibnamefont {Biamonte}}, \bibinfo {author} {\bibfnamefont {P.~J.}\ \bibnamefont {Coles}}, \bibinfo {author} {\bibfnamefont {L.}~\bibnamefont {Cincio}}, \bibinfo {author} {\bibfnamefont {J.~R.}\ \bibnamefont {McClean}}, \bibinfo {author} {\bibfnamefont {Z.}~\bibnamefont {Holmes}},\ and\ \bibinfo {author} {\bibfnamefont {M.}~\bibnamefont {Cerezo}},\ }\bibfield  {title} {\bibinfo {title} {Barren plateaus in variational quantum computing},\ }\href {https://doi.org/10.1038/s42254-025-00813-9} {\bibfield  {journal} {\bibinfo  {journal} {Nature Reviews Physics}\ }\textbf {\bibinfo {volume} {7}},\ \bibinfo {pages} {174–189} (\bibinfo {year} {2025})}\BibitemShut {NoStop}%
\bibitem [{\citenamefont {Puig}\ \emph {et~al.}(2025)\citenamefont {Puig}, \citenamefont {Drudis}, \citenamefont {Thanasilp},\ and\ \citenamefont {Holmes}}]{Puig25}%
  \BibitemOpen
  \bibfield  {author} {\bibinfo {author} {\bibfnamefont {R.}~\bibnamefont {Puig}}, \bibinfo {author} {\bibfnamefont {M.}~\bibnamefont {Drudis}}, \bibinfo {author} {\bibfnamefont {S.}~\bibnamefont {Thanasilp}},\ and\ \bibinfo {author} {\bibfnamefont {Z.}~\bibnamefont {Holmes}},\ }\bibfield  {title} {\bibinfo {title} {Variational quantum simulation: A case study for understanding warm starts},\ }\bibfield  {journal} {\bibinfo  {journal} {PRX Quantum}\ }\textbf {\bibinfo {volume} {6}},\ \href {https://doi.org/10.1103/prxquantum.6.010317} {10.1103/prxquantum.6.010317} (\bibinfo {year} {2025})\BibitemShut {NoStop}%
\bibitem [{\citenamefont {Krantz}\ \emph {et~al.}(2019)\citenamefont {Krantz}, \citenamefont {Kjaergaard}, \citenamefont {Yan}, \citenamefont {Orlando}, \citenamefont {Gustavsson},\ and\ \citenamefont {Oliver}}]{krantz2019quantum}%
  \BibitemOpen
  \bibfield  {author} {\bibinfo {author} {\bibfnamefont {P.}~\bibnamefont {Krantz}}, \bibinfo {author} {\bibfnamefont {M.}~\bibnamefont {Kjaergaard}}, \bibinfo {author} {\bibfnamefont {F.}~\bibnamefont {Yan}}, \bibinfo {author} {\bibfnamefont {T.~P.}\ \bibnamefont {Orlando}}, \bibinfo {author} {\bibfnamefont {S.}~\bibnamefont {Gustavsson}},\ and\ \bibinfo {author} {\bibfnamefont {W.~D.}\ \bibnamefont {Oliver}},\ }\bibfield  {title} {\bibinfo {title} {A quantum engineer's guide to superconducting qubits},\ }\href@noop {} {\bibfield  {journal} {\bibinfo  {journal} {Applied Physics Reviews}\ }\textbf {\bibinfo {volume} {6}},\ \bibinfo {pages} {021318} (\bibinfo {year} {2019})}\BibitemShut {NoStop}%
\bibitem [{\citenamefont {Breuer}\ \emph {et~al.}(2002)\citenamefont {Breuer}, \citenamefont {Petruccione} \emph {et~al.}}]{breuer2002theory}%
  \BibitemOpen
  \bibfield  {author} {\bibinfo {author} {\bibfnamefont {H.-P.}\ \bibnamefont {Breuer}}, \bibinfo {author} {\bibfnamefont {F.}~\bibnamefont {Petruccione}}, \emph {et~al.},\ }\href {https://doi.org/10.1093/acprof:oso/9780199213900.001.0001} {\emph {\bibinfo {title} {The theory of open quantum systems}}}\ (\bibinfo  {publisher} {Oxford University Press on Demand},\ \bibinfo {year} {2002})\BibitemShut {NoStop}%
\bibitem [{ibm(2025)}]{ibm_quantum_exp}%
  \BibitemOpen
  \href {https://quantum-computing.ibm.com/services/resources} {\bibinfo {title} {{IBM quantum compute resources}}} (\bibinfo {year} {2025})\BibitemShut {NoStop}%
\bibitem [{goo(2025)}]{google_quantum_AI}%
  \BibitemOpen
  \href {https://quantumai.google} {\bibinfo {title} {{Google Quantum AI}}} (\bibinfo {year} {2025})\BibitemShut {NoStop}%
\bibitem [{rig(2025)}]{rigetti}%
  \BibitemOpen
  \href {https://www.rigetti.com} {\bibinfo {title} {{Rigetti}}} (\bibinfo {year} {2025})\BibitemShut {NoStop}%
\bibitem [{qua(2025)}]{quandela_cloud}%
  \BibitemOpen
  \href {https://cloud.quandela.com} {\bibinfo {title} {{Quandela Cloud}}} (\bibinfo {year} {2025})\BibitemShut {NoStop}%
\bibitem [{Xan(2025)}]{Xanadu}%
  \BibitemOpen
  \href {https://www.xanadu.ai} {\bibinfo {title} {{Xanadu}}} (\bibinfo {year} {2025})\BibitemShut {NoStop}%
\bibitem [{Psi(2025)}]{PsiQuantum}%
  \BibitemOpen
  \href {https://www.psiquantum.com} {\bibinfo {title} {{PsiQuantum}}} (\bibinfo {year} {2025})\BibitemShut {NoStop}%
\bibitem [{ION(2025)}]{IONQ}%
  \BibitemOpen
  \href {https://ionq.com} {\bibinfo {title} {{IONQ}}} (\bibinfo {year} {2025})\BibitemShut {NoStop}%
\bibitem [{Qua(2025)}]{Quantinuum}%
  \BibitemOpen
  \href {https://www.quantinuum.com} {\bibinfo {title} {{Quantinuum}}} (\bibinfo {year} {2025})\BibitemShut {NoStop}%
\bibitem [{Pas(2025)}]{Pasqual}%
  \BibitemOpen
  \href {https://www.pasqal.com} {\bibinfo {title} {{Pasqual}}} (\bibinfo {year} {2025})\BibitemShut {NoStop}%
\bibitem [{QuE(2025)}]{QuEra}%
  \BibitemOpen
  \href {https://www.quera.com} {\bibinfo {title} {{QuEra}}} (\bibinfo {year} {2025})\BibitemShut {NoStop}%
\bibitem [{\citenamefont {D\"ur}\ \emph {et~al.}(2005)\citenamefont {D\"ur}, \citenamefont {Hein}, \citenamefont {Cirac},\ and\ \citenamefont {Briegel}}]{Dur}%
  \BibitemOpen
  \bibfield  {author} {\bibinfo {author} {\bibfnamefont {W.}~\bibnamefont {D\"ur}}, \bibinfo {author} {\bibfnamefont {M.}~\bibnamefont {Hein}}, \bibinfo {author} {\bibfnamefont {J.~I.}\ \bibnamefont {Cirac}},\ and\ \bibinfo {author} {\bibfnamefont {H.-J.}\ \bibnamefont {Briegel}},\ }\bibfield  {title} {\bibinfo {title} {Standard forms of noisy quantum operations via depolarization},\ }\href {https://doi.org/10.1103/PhysRevA.72.052326} {\bibfield  {journal} {\bibinfo  {journal} {Phys. Rev. A}\ }\textbf {\bibinfo {volume} {72}},\ \bibinfo {pages} {052326} (\bibinfo {year} {2005})}\BibitemShut {NoStop}%
\bibitem [{\citenamefont {Wallman}\ and\ \citenamefont {Emerson}(2016)}]{wallman2016}%
  \BibitemOpen
  \bibfield  {author} {\bibinfo {author} {\bibfnamefont {J.~J.}\ \bibnamefont {Wallman}}\ and\ \bibinfo {author} {\bibfnamefont {J.}~\bibnamefont {Emerson}},\ }\bibfield  {title} {\bibinfo {title} {Noise tailoring for scalable quantum computation via randomized compiling},\ }\href {https://doi.org/10.1103/PhysRevA.94.052325} {\bibfield  {journal} {\bibinfo  {journal} {Phys. Rev. A}\ }\textbf {\bibinfo {volume} {94}},\ \bibinfo {pages} {052325} (\bibinfo {year} {2016})}\BibitemShut {NoStop}%
\bibitem [{\citenamefont {Hashim}\ \emph {et~al.}(2021)\citenamefont {Hashim}, \citenamefont {Naik}, \citenamefont {Morvan}, \citenamefont {Ville}, \citenamefont {Mitchell}, \citenamefont {Kreikebaum}, \citenamefont {Davis}, \citenamefont {Smith}, \citenamefont {Iancu}, \citenamefont {O'Brien}, \citenamefont {Hincks}, \citenamefont {Wallman}, \citenamefont {Emerson},\ and\ \citenamefont {Siddiqi}}]{hashim2021}%
  \BibitemOpen
  \bibfield  {author} {\bibinfo {author} {\bibfnamefont {A.}~\bibnamefont {Hashim}}, \bibinfo {author} {\bibfnamefont {R.~K.}\ \bibnamefont {Naik}}, \bibinfo {author} {\bibfnamefont {A.}~\bibnamefont {Morvan}}, \bibinfo {author} {\bibfnamefont {J.-L.}\ \bibnamefont {Ville}}, \bibinfo {author} {\bibfnamefont {B.}~\bibnamefont {Mitchell}}, \bibinfo {author} {\bibfnamefont {J.~M.}\ \bibnamefont {Kreikebaum}}, \bibinfo {author} {\bibfnamefont {M.}~\bibnamefont {Davis}}, \bibinfo {author} {\bibfnamefont {E.}~\bibnamefont {Smith}}, \bibinfo {author} {\bibfnamefont {C.}~\bibnamefont {Iancu}}, \bibinfo {author} {\bibfnamefont {K.~P.}\ \bibnamefont {O'Brien}}, \bibinfo {author} {\bibfnamefont {I.}~\bibnamefont {Hincks}}, \bibinfo {author} {\bibfnamefont {J.~J.}\ \bibnamefont {Wallman}}, \bibinfo {author} {\bibfnamefont {J.}~\bibnamefont {Emerson}},\ and\ \bibinfo {author} {\bibfnamefont {I.}~\bibnamefont {Siddiqi}},\ }\bibfield  {title} {\bibinfo {title} {Randomized compiling for scalable quantum computing on a noisy
  superconducting quantum processor},\ }\href {https://doi.org/10.1103/PhysRevX.11.041039} {\bibfield  {journal} {\bibinfo  {journal} {Phys. Rev. X}\ }\textbf {\bibinfo {volume} {11}},\ \bibinfo {pages} {041039} (\bibinfo {year} {2021})}\BibitemShut {NoStop}%
\bibitem [{\citenamefont {Bergholm}\ \emph {et~al.}(2022)\citenamefont {Bergholm}, \citenamefont {Izaac}, \citenamefont {Schuld}, \citenamefont {Gogolin}, \citenamefont {Ahmed},\ and\ \citenamefont {et~al.}}]{pennylane2022}%
  \BibitemOpen
  \bibfield  {author} {\bibinfo {author} {\bibfnamefont {V.}~\bibnamefont {Bergholm}}, \bibinfo {author} {\bibfnamefont {J.}~\bibnamefont {Izaac}}, \bibinfo {author} {\bibfnamefont {M.}~\bibnamefont {Schuld}}, \bibinfo {author} {\bibfnamefont {C.}~\bibnamefont {Gogolin}}, \bibinfo {author} {\bibfnamefont {S.}~\bibnamefont {Ahmed}},\ and\ \bibinfo {author} {\bibfnamefont {V.~A.}\ \bibnamefont {et~al.}},\ }\href@noop {} {\bibinfo {title} {Pennylane: Automatic differentiation of hybrid quantum-classical computations}} (\bibinfo {year} {2022}),\ \Eprint {https://arxiv.org/abs/1811.04968} {arXiv:1811.04968 [quant-ph]} \BibitemShut {NoStop}%
\bibitem [{\citenamefont {Cai}\ and\ \citenamefont {Benjamin}(2019)}]{Cai_2019}%
  \BibitemOpen
  \bibfield  {author} {\bibinfo {author} {\bibfnamefont {Z.}~\bibnamefont {Cai}}\ and\ \bibinfo {author} {\bibfnamefont {S.~C.}\ \bibnamefont {Benjamin}},\ }\bibfield  {title} {\bibinfo {title} {Constructing smaller pauli twirling sets for arbitrary error channels},\ }\bibfield  {journal} {\bibinfo  {journal} {Scientific Reports}\ }\textbf {\bibinfo {volume} {9}},\ \href {https://doi.org/10.1038/s41598-019-46722-7} {10.1038/s41598-019-46722-7} (\bibinfo {year} {2019})\BibitemShut {NoStop}%
\bibitem [{\citenamefont {Di~Matteo}(2014)}]{olivia_dimatteo}%
  \BibitemOpen
  \bibfield  {author} {\bibinfo {author} {\bibfnamefont {O.}~\bibnamefont {Di~Matteo}},\ }\href {https://www.google.com/search?client=safari&rls=en&q=Olivia+Di+Matteo%2C+CS+867%2FQIC+890%2C+4+November+2014&ie=UTF-8&oe=UTF-8} {\bibinfo {title} {{A short introduction to unitary 2-designs}}} (\bibinfo {year} {2014})\BibitemShut {NoStop}%
\bibitem [{\citenamefont {Bravyi}\ and\ \citenamefont {Maslov}(2021)}]{Bravyi_2021}%
  \BibitemOpen
  \bibfield  {author} {\bibinfo {author} {\bibfnamefont {S.}~\bibnamefont {Bravyi}}\ and\ \bibinfo {author} {\bibfnamefont {D.}~\bibnamefont {Maslov}},\ }\bibfield  {title} {\bibinfo {title} {Hadamard-free circuits expose the structure of the clifford group},\ }\href {https://doi.org/10.1109/tit.2021.3081415} {\bibfield  {journal} {\bibinfo  {journal} {IEEE Transactions on Information Theory}\ }\textbf {\bibinfo {volume} {67}},\ \bibinfo {pages} {4546–4563} (\bibinfo {year} {2021})}\BibitemShut {NoStop}%
\bibitem [{\citenamefont {Gross}\ \emph {et~al.}(2007)\citenamefont {Gross}, \citenamefont {Audenaert},\ and\ \citenamefont {Eisert}}]{Gross_2007}%
  \BibitemOpen
  \bibfield  {author} {\bibinfo {author} {\bibfnamefont {D.}~\bibnamefont {Gross}}, \bibinfo {author} {\bibfnamefont {K.}~\bibnamefont {Audenaert}},\ and\ \bibinfo {author} {\bibfnamefont {J.}~\bibnamefont {Eisert}},\ }\bibfield  {title} {\bibinfo {title} {Evenly distributed unitaries: On the structure of unitary designs},\ }\bibfield  {journal} {\bibinfo  {journal} {Journal of Mathematical Physics}\ }\textbf {\bibinfo {volume} {48}},\ \href {https://doi.org/10.1063/1.2716992} {10.1063/1.2716992} (\bibinfo {year} {2007})\BibitemShut {NoStop}%
\bibitem [{\citenamefont {Emerson}\ \emph {et~al.}(2005)\citenamefont {Emerson}, \citenamefont {Alicki},\ and\ \citenamefont {Życzkowski}}]{Emerson_2005}%
  \BibitemOpen
  \bibfield  {author} {\bibinfo {author} {\bibfnamefont {J.}~\bibnamefont {Emerson}}, \bibinfo {author} {\bibfnamefont {R.}~\bibnamefont {Alicki}},\ and\ \bibinfo {author} {\bibfnamefont {K.}~\bibnamefont {Życzkowski}},\ }\bibfield  {title} {\bibinfo {title} {Scalable noise estimation with random unitary operators},\ }\href {https://doi.org/10.1088/1464-4266/7/10/021} {\bibfield  {journal} {\bibinfo  {journal} {Journal of Optics B: Quantum and Semiclassical Optics}\ }\textbf {\bibinfo {volume} {7}},\ \bibinfo {pages} {S347–S352} (\bibinfo {year} {2005})}\BibitemShut {NoStop}%
\bibitem [{\citenamefont {Dür}\ \emph {et~al.}(2005)\citenamefont {Dür}, \citenamefont {Hein}, \citenamefont {Cirac},\ and\ \citenamefont {Briegel}}]{D_r_2005}%
  \BibitemOpen
  \bibfield  {author} {\bibinfo {author} {\bibfnamefont {W.}~\bibnamefont {Dür}}, \bibinfo {author} {\bibfnamefont {M.}~\bibnamefont {Hein}}, \bibinfo {author} {\bibfnamefont {J.~I.}\ \bibnamefont {Cirac}},\ and\ \bibinfo {author} {\bibfnamefont {H.-J.}\ \bibnamefont {Briegel}},\ }\bibfield  {title} {\bibinfo {title} {Standard forms of noisy quantum operations via depolarization},\ }\bibfield  {journal} {\bibinfo  {journal} {Physical Review A}\ }\textbf {\bibinfo {volume} {72}},\ \href {https://doi.org/10.1103/physreva.72.052326} {10.1103/physreva.72.052326} (\bibinfo {year} {2005})\BibitemShut {NoStop}%
\end{thebibliography}%

\end{document}